\documentclass[journal]{IEEEtran}
\usepackage{tabularx}
\usepackage{makecell}

\usepackage{threeparttable}
\usepackage{algpseudocode}
\usepackage{setspace}

\usepackage{amsfonts}
\usepackage{amssymb}
\usepackage{amsthm}
\usepackage{cite}
\usepackage[cmex10]{amsmath}
\usepackage{float}
\usepackage{color}
\usepackage{stfloats,fancyhdr}
\usepackage{amsmath,bm}
\usepackage{algorithm}
\usepackage{changepage}
\usepackage[normalem]{ulem}
\usepackage{balance}
\usepackage{bbm}

\newtheorem{problem}{Problem}
\newtheorem{theorem}{Theorem}
\newtheorem{lemma}{Lemma}
\newtheorem{proposition}{Proposition}

\newtheorem{definition}{Definition}

\newtheorem{remark}{Remark}
\newtheorem{example}{Example}

\ifCLASSINFOpdf
\usepackage[pdftex]{graphicx}
\DeclareGraphicsExtensions{.pdf,.jpeg,.png}
\else
\usepackage[dvips]{graphicx}
\DeclareGraphicsExtensions{.eps}
\fi

\usepackage{mathtools}
\mathtoolsset{showonlyrefs=true}

\title{\huge Goal-oriented Transmission Scheduling: Structure-guided DRL with a Unified Dual On-policy and Off-policy Approach}
\author{\large Jiazheng Chen,~{\IEEEmembership{Graduate Student Member,~IEEE,}} Wanchun Liu*,~\IEEEmembership{Member,~IEEE,}
\thanks{The work of W. Liu was supported by the Australian Research Council's Discovery Early Career Researcher Award (DECRA) Project DE230100016. \textit{(Corresponding author: W. Liu.)}}
\thanks{J. Chen and W. Liu are with the School of Electrical and Information Engineering, The University of Sydney, Sydney, Australia (e-mail: jiazheng.chen@sydney.edu.au; wanchun.liu@sydney.edu.au).
}
} 
\begin{document}

\maketitle

\begin{abstract}
Goal-oriented communications prioritize application-driven objectives over data accuracy, enabling intelligent next-generation wireless systems. Efficient scheduling in multi-device, multi-channel systems poses significant challenges due to high-dimensional state and action spaces. We address these challenges by deriving key structural properties of the optimal solution to the goal-oriented scheduling problem, incorporating Age of Information (AoI) and channel states. Specifically, we establish the monotonicity of the optimal state value function—a measure of long-term system performance—w.r.t. channel states and prove its asymptotic convexity w.r.t. AoI states. Additionally, we derive the monotonicity of the optimal policy w.r.t. channel states, advancing the theoretical framework for optimal scheduling.
Leveraging these insights, we propose the structure-guided unified dual on-off policy DRL (SUDO-DRL), a hybrid algorithm that combines the stability of on-policy training with the sample efficiency of off-policy methods. Through a novel structural property evaluation framework, SUDO-DRL enables effective and scalable training, addressing the complexities of large-scale systems.
Numerical results show SUDO-DRL improves system performance by up to 45\% and reduces convergence time by 40\% compared to state-of-the-art methods. It also effectively handles scheduling in much larger systems, where off-policy DRL fails and on-policy benchmarks exhibit significant performance loss, demonstrating its scalability and efficacy in goal-oriented communications.
\end{abstract}

% \vspace{-0.6cm}

\section{Introduction} \label{sec:intro}
Conventional communications focus on accurate bit-by-bit data transmission, achieving near-Shannon-capacity efficiency in 5G networks. However, as we transition to 6G, \textbf{goal-oriented communications} emerge as a transformative paradigm, prioritizing application-driven objectives over raw data accuracy~\cite{gunduz2022beyond}. This shift is critical for enabling intelligent and efficient next-generation networks.
Goal-oriented communications encompass two main categories: \textbf{human-centric} and \textbf{machine-centric}. Human-centric applications, such as extended reality (XR)~\cite{chen2023aoii} and augmented reality (AR)~\cite{wang2024goalar}, focus on preserving semantic meaning for accurate human comprehension. Machine-centric applications, including industrial Internet of Things (IIoT)~\cite{Huang2020UpDown} and autonomous driving~\cite{xu2022autodrive}, prioritize transmitting information that directly optimizes system performance. This paradigm shift transcends traditional communication approaches, enabling smarter, more efficient interactions between humans, machines, and their environments~\cite{li2023towardsgoal}.

\subsection{Scheduling in Goal-oriented Communications}
Efficient scheduling is crucial in goal-oriented communications to maximize the performance of systems with limited communication resources. Scheduling determines how devices share communication channels, directly impacting the timeliness and relevance of transmitted information—key factors for achieving system goals. Unlike conventional communication systems that typically use throughput, latency, or reliability as performance metrics, goal-oriented communications adopt application-specific metrics that evaluate the importance of the transmitted information in achieving the desired objective. Among these, the age of information (AoI)~\cite{kaul2012aoi}, which measures the freshness of data, is particularly important for machine-type applications where outdated messages can become irrelevant or even harmful to system performance.

To address scheduling challenges, optimizing scheduling policies has garnered significant attention in the communications community~\cite{gunduz2022beyond}. Existing works often use AoI and related metrics to guide scheduling decisions. For example, in~\cite{wang2024goal}, a control cost minimization problem in a single-loop network is reformulated as an AoI-based optimization problem and solved using Markov decision processes (MDPs) with value iteration. Beyond AoI, other metrics such as the value of information (VoI)\cite{soleymani2019value} and mean square error (MSE)\cite{wu2018optimalmulti} have been introduced to assess information importance in various contexts. In~\cite{soleymani2019value}, an optimal scheduling and power allocation problem is proposed for a single-sensor-single-controller system, maximizing VoI through an event-triggered policy. Similarly,~\cite{wu2018optimalmulti} addresses remote estimation in a multi-sensor system by formulating an MSE minimization problem, solved using Whittle's index heuristic to derive a suboptimal policy.
However, these approaches have notable limitations. Heuristic methods, while computationally efficient, cannot guarantee optimality. On the other hand, conventional dynamic programming methods, such as value and policy iteration, are computationally infeasible for large-scale systems with high-dimensional state and action spaces. These challenges underscore the need for scalable and optimal solutions tailored to the complexities of modern goal-oriented communication systems.

\subsection{Off-policy and On-policy DRL Solutions}
In recent years, deep reinforcement learning (DRL) has emerged as a powerful tool for solving large-scale MDPs by leveraging deep neural networks (NNs) to approximate functions~\cite{leong2020DRL, holm2023goal, bai2023dqntransport}. Scheduling policies derived from DRL significantly outperform heuristic approaches, offering more optimal solutions in complex systems.
Notably, deep Q-networks (DQN), a fundamental off-policy DRL algorithm, have been applied to solve optimal scheduling problems in various multi-device, multi-channel systems, such as remote state estimation systems~\cite{leong2020DRL, holm2023goal} and intelligent transportation systems~\cite{bai2023dqntransport}. Building on the conventional DQN, the guided exploration-based branching dueling Q-network (GE-SDQN)\cite{xie2023gesdqn} improves exploration efficiency and scheduling performance while reducing the number of neurons required for large-scale systems. In\cite{he2024age}, the deep deterministic policy gradient (DDPG) algorithm—a widely used off-policy DRL method with an actor-critic framework—is employed to address AoI minimization in larger-scale systems, surpassing the limitations of DQN in scalability.

All these methods rely on off-policy DRL, where the current policy is updated using data collected from past policies. While off-policy methods exhibit high sample efficiency and effective exploration by reusing past data, they also introduce training instability and bias due to potential discrepancies between the behaviors of the current and past policies.
In contrast, on-policy DRL provides a more stable learning process by updating policies using data generated exclusively by the current policy. This approach, which discards previously collected data after each update to keep training data closely aligned with the current policy, ensures greater stability in stochastic environments.
For example, the trust region policy optimization (TRPO) algorithm, an on-policy DRL method with an actor-critic framework, is employed in~\cite{peng2022trpo} to derive channel and power allocation policies aimed at minimizing the sum of AoI and power consumption. TRPO enhances training stability by constraining both the direction and magnitude of policy updates, offering theoretical guarantees for policy improvement. Similarly, the proximal policy optimization (PPO) algorithm, a computationally simpler variant of TRPO, is utilized in~\cite{pang2022drl} to solve scheduling problems in large-scale remote state estimation systems where off-policy DRL methods face significant challenges.
While on-policy methods offer stability and have been effectively applied in specific scenarios, they suffer from poor data efficiency because generated data is discarded after each update. This inefficiency, combined with insufficient exploration, can sometimes hinder performance, particularly in scenarios requiring unbiased sampling~\cite{thomas2014bias}.

\subsection{Initial Studies on Structure-Enhanced DRL Algorithms}
Most existing works apply general DRL algorithms to solve specific scheduling problems without incorporating domain-specific insights, focusing instead on brute-force optimization techniques. As a result, these algorithms are prone to getting stuck in local minima, leading to performance losses compared to the theoretical optimal policy. A major limitation of these approaches is the lack of investigation into the structural properties of optimal policies, which, if utilized, could significantly enhance the efficiency and effectiveness of DRL algorithms.

Recently, there has been growing interest in leveraging the structural properties of optimal policies to improve DRL-based solutions. For example, a basic single-sensor transmission scheduling problem in a remote estimation system under communication constraints is analyzed in~\cite{wu2019single}, where the authors identified the monotonicity and submodularity of the state-action value function of the optimal policy. Monotonicity implies that taking certain actions consistently improves system performance, and submodularity reflects diminishing returns when applying multiple actions. Building on this, the analysis is extended to a multi-sensor remote state estimation system over AWGN channels in~\cite{wu2020optimalmulti}, where the threshold property of the optimal policy is formally proven, showing that sensors are scheduled based on well-defined thresholds.

More recently, pioneering works have focused on developing theoretical properties of optimal policies to guide DRL algorithms in efficiently discovering optimal solutions. In~\cite{chen2022seDRL}, the authors studied a multi-sensor remote estimation system over fading channels and proved that the optimal policy exhibits a threshold structure. They then proposed a structure-enhanced DRL algorithm that leverages this property to achieve improved performance compared to traditional DRL methods. A follow-up study in~\cite{chen2023semantic} further demonstrated that the state-action value function is monotonic with respect to both AoI and channel states. This insight led to the development of a monotonicity-enforced DDPG algorithm, which enhances convergence speed and performance over baseline methods.
However, these works rely exclusively on off-policy DRL, which suffers from inherent instability, particularly when applied to large-scale dynamic decision-making problems. Consequently, the proposed algorithms are limited to solving scheduling problems in systems with a scale of up to twenty sensors and ten channels.

In this paper, we focus on a multi-device goal-oriented transmission scheduling problem over fading channels and delve deeper into exploring the structural properties of the optimal policy. By leveraging these theoretical insights, we aim to develop a hybrid DRL algorithm that integrates the strengths of both off-policy and on-policy DRL methods, enabling the efficient resolution of large-scale systems.
The main contributions of this work are summarized as follows:
\begin{enumerate}
    \item We derive key structural properties of the optimal solution to the formulated MDP for the goal-oriented scheduling problem, which accounts for both AoI and channel states of all devices. Specifically, we establish the monotonicity of the optimal state value function w.r.t. channel states, complementing the monotonicity w.r.t. AoI states derived in our earlier work~\cite{chen2022seDRL}. Additionally, we prove the asymptotic convexity of the state value function w.r.t. AoI states, representing the first result in the literature to explore convexity in transmission scheduling problems. Finally, we derive the monotonicity of the optimal policy w.r.t. channel states, further advancing the theoretical understanding of optimal scheduling in goal-oriented communications.

    \item We propose the structure-guided unified dual on-off policy DRL (SUDO-DRL), a novel hybrid algorithm that leverages the derived structural insights to solve the scheduling problem efficiently. SUDO-DRL uniquely combines the stability of on-policy training and the sample efficiency of off-policy methods through a unified loss function. A structural property evaluation framework is introduced to derive critic-monotonicity, critic-convexity, and actor-monotonicity scores, which are incorporated into the on-policy loss function. For the off-policy component, the structural scores guide replay buffer management by selectively storing transitions from good policies and enabling priority-based sampling, significantly enhancing training effectiveness and efficiency.

    \item The proposed SUDO-DRL demonstrates robust performance improvements in goal-oriented communication systems. Numerical experiments reveal that it enhances system performance by 25\% to 45\%, while reducing convergence time by approximately 40\% compared to state-of-the-art methods. Furthermore, by leveraging the advantages of both on-policy and off-policy training methods, SUDO-DRL effectively addresses scheduling problems in systems with up to 40 devices and 20 channels—a scale where benchmark off-policy algorithms fail to converge, and state-of-the-art on-policy DRL exhibits significant performance loss—underscoring its scalability and effectiveness in large-scale scenarios.
\end{enumerate}

\textbf{Outline:} The system model of the goal-oriented communication system is introduced in Section~\ref{sec: sys}. The transmission scheduling problem formulation and the definition of value functions are presented in Section~\ref{sec: problem formulation and value function}. The structural properties of value functions and optimal policies are proven in Section~\ref{sec: proof of structure}. To solve the formulated problem, the structure-guided unified on-off policy DRL algorithm is developed in Section~\ref{sec: DRL}. The results of the numerical experiments are given and analyzed in Section~\ref{sec: simulation}. Finally, the conclusion is presented in Section~\ref{sec: conclusion}.

\section{System Model} \label{sec: sys}
We consider a wireless goal-oriented communication system with $N$ edge devices (e.g. cameras or sensors) and a remote destination (e.g. a base station or remote estimator) as illustrated in Fig.~\ref{fig: system model}.
The devices transmit local data to the remote destination through $M$ channels (e.g. subcarriers) where $M<N$.

\begin{figure}
    \centering
    \includegraphics[width=0.97\linewidth]{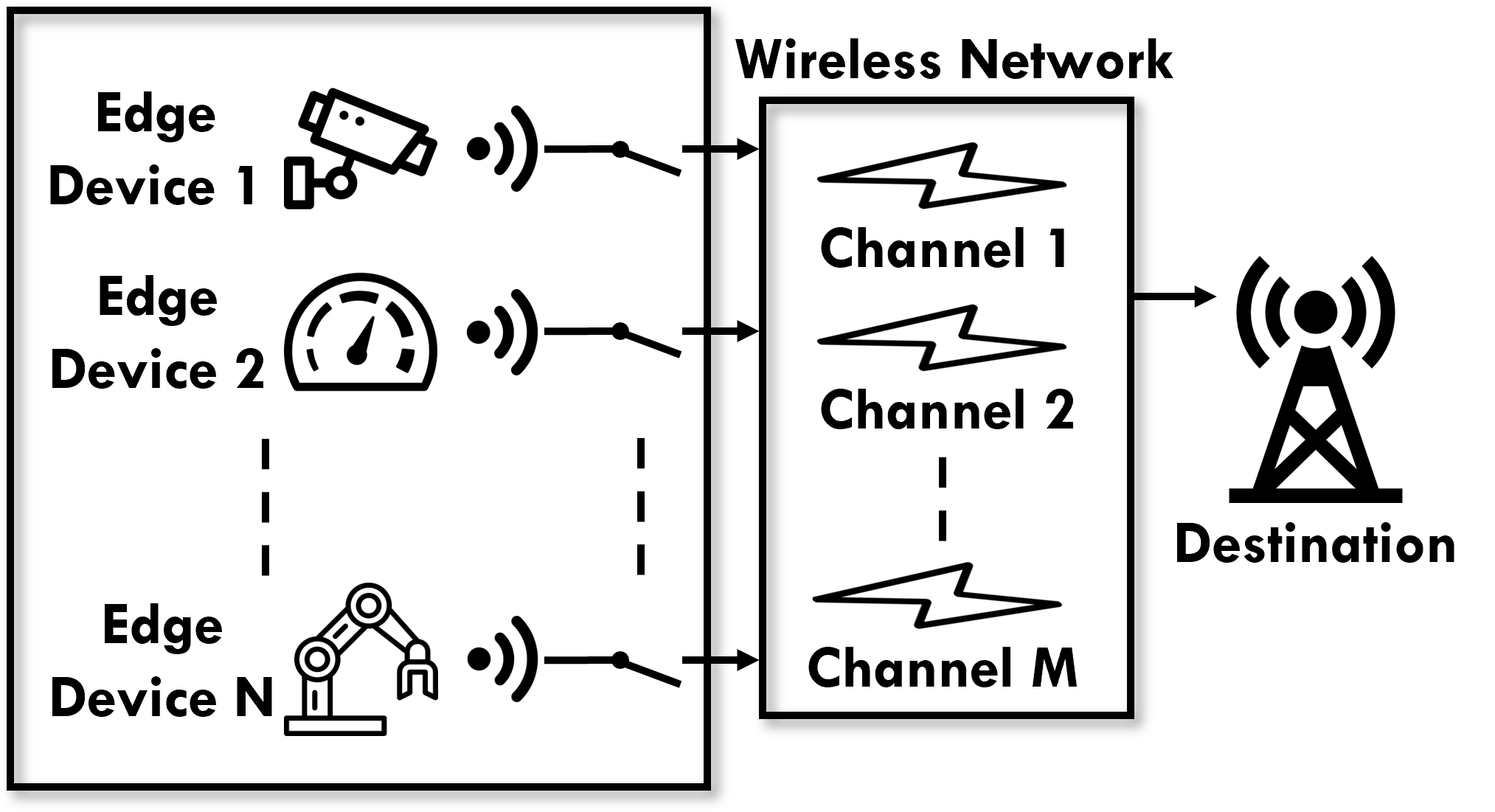}
    \caption{Goal-oriented communication system with $N$ edge devices, $M$ channels, and a remote destination}
    \label{fig: system model}
\end{figure}

\subsection{Communication Model}\label{subsec: communication model}
In this paper, wireless channels are modeled as independent and identically distributed (i.i.d.) block fading channels, where the channel state remains constant during each packet transmission, but changes independently between each transmission.
At time step $t$, we denote the system channel state between device $n$ and the remote destination at channel $m$ as $g_{n,m,t} \in \mathcal{G} \triangleq \left\{ 1, \dots, \bar{g} \right\}$ with $\bar{g}$ quantization levels.
The overall system channel state is represented by an $N \times M$ matrix $\mathbf{G}_{t}$ with $g_{n,m,t}$ being the element in the $m$th column and $n$th row.
The channel state $g_{n,m,t}$ follows the probability distribution:
\begin{equation}
    \operatorname{P}(g_{n,m,t}=i) = q^{i}_{n,m}, \forall t,
    \label{eq: channel distribution}
\end{equation}
where $\sum_{i=1}^{\bar{g}} q^{i}_{n,m} =1, \forall n,m$.
The packet drop rate for state $g_{n,m,t}$ is denoted as $\psi_{n,m,t}$, with higher channel states corresponding to higher packet drop rates.
The remote destination acquires the instantaneous channel state $\mathbf{G}_{t}$ using standard channel estimation methods~\cite{Coleri2002ChannelEstimation}.

The channel assignment for device $n$ at time $t$ is denoted~as
\begin{equation}
    a_{n,t} = 
    \begin{cases}
        0,  & \text{if no channel is allocated to device $n$,} \\
        m,  & \text{if channel $m$ is allocated to device $n$.}
    \end{cases}
\label{eq:action}
\end{equation}
This assignment satisfies the constraint:
\begin{equation}
    \sum_{n=1}^{N} \mathbbm{1}\left( a_{n,t} = m \right) = 1,
    \sum_{m=1}^{M} \mathbbm{1}\left( a_{n,t} = m \right) \leq 1,
    \label{eq: action constraint}
\end{equation} where $\mathbbm{1} (\cdot)$ is the indicator function.
The constraint ensures that each channel is assigned to only one device, and each device is allocated at most one channel.

\subsection{Goal-oriented Communication Performance Metric} \label{subsec: metric}
We define $\delta_{n,t} \in \{1, 2, \dots\}$ as the AoI of the device $n$ at time $t$, which refers to the time elapsed since the last successful receive of device packet at the destination~\cite{liu2021remoteMF, tong2022age}:
\begin{equation}
\!\!\!\!\!\delta_{n,t+1} \!=\!\!
\begin{cases}
\!1, & {\text{if remote destination receive}}\\
     & { \text{device $n$'s packet at time $t$}} \\
\!\delta_{n,t}\!+\!1, & {\text{otherwise.}}
\end{cases}
\label{eq: aoi}
\end{equation}
To characterize the importance of information in a goal-oriented communication system, we define a \textbf{positive cost function} $c_{n} (\delta_{n,t}), \forall n \in \{1, 2, \dots, N\}$. This cost function, a critical performance metric for device $n$, is \textbf{non-decreasing} with respect to AoI and varies based on the system's goals, with lower values indicating better performance.

Next, we provide an example of a goal-oriented communication system~\cite{chen2023semantic}.

\begin{example}[Remote state estimation system]\label{example: remote estimation}
    This system consists of a remote estimator that reconstructs information sent by $N$ sensors,where each sensor $n$ measures a corresponding dynamic process modeled as a discrete-time linear time-invariant (LTI) system~\cite{leong2020DRL, liu2022stability}:
    \begin{equation}
        \mathbf{x}_{n,t+1}  = \mathbf{A}_{n} \mathbf{x}_{n, t} + \mathbf{w}_{n, t},  \ 
        \mathbf{y}_{n, t}  = \mathbf{C}_{n} \mathbf{x}_{n, t} + \mathbf{v}_{n, t},
        \label{eq:LTI}
    \end{equation}
    where $\mathbf{x}_{n, t} \in \mathbb{R}^{r_{n}}$ and $\mathbf{y}_{n, t} \in \mathbb{R}^{e_{n}}$ are process $n$'s state and its measurement of sensor $n$, respectively; $\mathbf{A}_{n} \in \mathbb{R}^{r_{n} \times r_{n}}$ and $\mathbf{C}_{n} \in \mathbb{R}^{e_{n} \times r_{n}}$ are the system matrix
    and the measurement matrix, respectively; $\mathbf{w}_{n, t} \in \mathbb{R}^{r_{n}}$ and $\mathbf{v}_{n, t} \in \mathbb{R}^{e_{n}}$ are the process disturbance and the measurement noise modeled as independent and identically distributed (i.i.d) zero-mean Gaussian random vectors $\mathcal{N}(\mathbf{0},\mathbf{W}_{n})$ and $\mathcal{N}(\mathbf{0},\mathbf{V}_{n})$, respectively.
    
    Due  to measurement imperfections ($\mathbf{y}_{n, t} \neq \mathbf{x}_{n, t}$), sensor $n$ generates a local state estimate $\hat{\mathbf{x}}_{n,t}^{\text{local}}$ based on the raw measurements $\{\mathbf{y}_{n, t} \}$ by executing a local Kalman filter.
    Limited wireless channels and packet dropouts may prevent some local estimates from reaching the remote estimator, which then computes a remote state estimate $\hat{\mathbf{x}}_{n,t}$ using a minimum mean-square error (MMSE) estimator, where the estimation error covariance at time $t$ is
    \begin{equation}
        \mathbf{P}_{n, t} \triangleq \mathbb{E} \left[ \left(\hat{\mathbf{x}}_{n, t} - \mathbf{x}_{n, t}\right) \left(\hat{\mathbf{x}}_{n, t} - \mathbf{x}_{n, t} \right)^{\top} \right] 
          = h_{n}^{\delta_{n,t}}(\bar{\mathbf{P}}_{n}),
    \end{equation}
    where $h_{n}(\mathbf{X}) = \mathbf{A}_{n} \mathbf{X} \mathbf{A}_{n}^{\top} + \mathbf{W}_{n}$, $h^{\delta+1}_{n}(\cdot) =h_n(h^{\delta}_{n}(\cdot))$, and $\bar{\mathbf{P}}_{n}$ is a constant depending on $\mathbf{A}_n$, $\mathbf{C}_n$, $\mathbf{W}_n$, and $\mathbf{V}_n$~\cite{liu2021remoteMF}.

    Since the goal of the remote state estimation system is to provide high-quality remote estimation, the cost function of process $n$ is defined as the \textbf{estimation mean-square error (MSE)}, i.e., 
    \begin{equation}
        c_n(\delta_{n,t}) \triangleq \operatorname{Tr}(\mathbf{P}_{n,t}) = \operatorname{Tr}\left(h_{n}^{\delta_{n,t}}(\bar{\mathbf{P}}_{n})\right).
        \label{eq: cost function}
    \end{equation}
\end{example}

\section{Goal-oriented Transmission Scheduling}\label{sec: problem formulation and value function}
Our goal is to determine a dynamic scheduling policy $\pi(\cdot)$ that, based on the AoI state $\bm{\delta}_{t}\triangleq \{\delta_{1,t},\dots,\delta_{N,t}\}$ and the channel state $\mathbf{G}_{t}$, minimizes the infinity-horizon expected sum of cost functions across all $N$ devices, with a discount factor $\gamma\in(0,1)$. The problem is formulated as follows:
\begin{problem}
\begin{equation}
    \min_{\pi} \lim_{T \to \infty} \mathbb{E}^\pi \left[ \sum_{t=1}^{T} \sum_{n=1}^N \gamma^t c_n(\delta_{n,t}) \right].
\end{equation}
\label{pro1}
\end{problem}

\subsection{MDP Formulartion} \label{subsec: mdp fomulation}

In Problem~\ref{pro1}, the channel states are assumed to be i.i.d. over time, and the cost $c_n(\delta_{n,t})$ is defined as a function solely dependent on the Markovian AoI state $\delta_{n,t}$, as described in~\eqref{eq: aoi}. Consequently, Problem~\ref{pro1} is a sequential decision-making problem that satisfies the Markov property, allowing it to be formulated as an MDP as below.

$\bullet$ \textbf{States:} At time $t$, given the instantaneous AoI state vector $\bm{\delta}_{t} \in \mathbb{N}^N$ and the real-time full channel state matrix $\mathbf{G}_{t} \in \mathcal{G}^{N \times M}$, the MDP state is defined as $\mathbf{s}_{t} \triangleq (\bm{\delta}_{t}, \mathbf{G}_{t})$ and the state space is $\mathcal{S} \triangleq \mathbb{N}^N \times \mathcal{G}^{N \times M}$.

$\bullet$ \textbf{Actions:} Given a policy function $\pi(\cdot)$, which maps a state to an action, the action at time $t$ is defined as $\mathbf{a}_{t} = \pi(\mathbf{s}_{t}) = (a_{1,t}, \dots, a_{N,t}) \in \left\{ 0, 1, 2, \dots, M \right\}^N$, subject to the constraint~\eqref{eq: action constraint}. Under this constraint, the action space is $\mathcal{A} \subset \left\{ 0, 1, 2, \dots, M \right\}^N$ with the size of $N!/(N-M)!$.

$\bullet$ \textbf{Transitions:} In the MDP, the probability of transitioning to the next state $\mathbf{s}_{t+1}$ from the current state $\mathbf{s}_t$ after executing the action $\mathbf{a}_t$ is denoted as the state transition probability $\operatorname{P}(\mathbf{s}_{t+1}|\mathbf{s}_{t}, \mathbf{a}_{t})$.
Since the optimal policy of an infinite-horizon MDP is stationary~\cite[Chapter~6]{puterman2014markov}, this state transition is independent of the time $t$, i.e., time homogeneous.
For simplicity, the subscript $t$ is dropped and the states at the current and next time steps are notated by $\mathbf{s}$ and $\mathbf{s}^{+}$, respectively.
Then, state transition probability is given as
\begin{equation}
    \operatorname{P}(\mathbf{s}^+|\mathbf{s}, \mathbf{a}) = \operatorname{P}(\bm{\delta}^+|\bm{\delta},\mathbf{G}, \mathbf{a}) \operatorname{P}(\mathbf{G}^+),
\end{equation}
where $\operatorname{P}(\mathbf{G}^+)$ is the probability of the channel state matrix and can be derived by using~\eqref{eq: channel distribution}, and $\operatorname{P}(\bm{\delta}^+|\bm{\delta},\mathbf{G}, \mathbf{a})$ is the AoI state vector transition probability:
\begin{equation}\label{eq:transision_prob}
    \operatorname{P}(\bm{\delta}^+|\bm{\delta},\mathbf{G}, \mathbf{a}) = \Pi_{n=1}^{N} \operatorname{P}(\delta^+_{n}|\delta_{n}, \mathbf{G}, a_{n})
\end{equation}
where 
\begin{align}
    \!\!\!\!\!\operatorname{P}(\delta^+_{n}|\delta_{n}, \mathbf{g}_n, a_{n})  \!=\! 
    \begin{cases}
        1 - \psi_{n, m}, & \!\!\!\text{if $\delta^+_{n} = 1, a_{n} = m$,} \\
        \psi_{n, m}, & \!\!\!\text{if $\delta^+_{n} = \delta_{n} \!+\! 1, a_{n} \!= \!m$,}\\
        1,          & \!\!\!\text{if $\delta^+_{n} = \delta_{n} \!+\! 1, a_{n} \!=\! 0$,} \\
        0, & \!\!\!{\text{otherwise,}}
    \end{cases}
    \label{eq: transition p 2}
\end{align}
where $\mathbf{g}_n$ is sensor $n$'s channel vector, i.e., the $n$th row of $\mathbf{G}$. \eqref{eq: transition p 2} is obtained from~\eqref{eq:action} and~\eqref{eq: aoi}.

$\bullet$ \textbf{Costs:} The immediate cost is defined as the sum of the cost of all devices at time $t$, i.e., $c(\mathbf{s}_{t}) = \sum_{n=1}^N c_n(\delta_{n,t})$.

\subsection{Value Functions for the Optimal Policy} \label{subsec: value function}
For the optimal scheduling policy of the MDP, i.e., $\pi^*(\cdot)$, we define the optimal \textbf{action-value function} $Q(\mathbf{s}_{t}, \mathbf{a}_{t}): \mathcal{S} \times \mathcal{A} \to \mathbb{R}$ and the optimal \textbf{state-value function}  $\upsilon^{*}(\mathbf{s}_{t}): \mathcal{S} \to \mathbb{R}$ as below.

Given the current state $\mathbf{s}_{t}$ and action $\mathbf{a}_{t}$, the action-value function, also known as Q function, represents the expected cumulative discounted cost of executing action $\mathbf{a}_{t}$ and following the optimal policy $\pi^*(\cdot)$, i.e., 
\begin{equation}
    \!\!\!\!Q(\mathbf{s}_{t}, \mathbf{a}_{t}) \!=\!  \mathbb{E} \left[\sum_{\tilde{t}=t}^{\infty} \gamma^{\tilde{t}-t} c(\mathbf{s}_{\tilde{t}}) \Big\vert \mathbf{a}_{t}, \mathbf{a}_{\tilde{t}} \!=\! \pi^{*}(\mathbf{s}_{\tilde{t}}),\forall \tilde{t}\!>\!t \right]\!,\!
    \label{eq: def Q}
\end{equation}
which satisfies the Bellman optimality equation:
\begin{equation}
    \!\!Q(\mathbf{s}_{t},\! \mathbf{a}_{t})  \!=\! c(\mathbf{s}_{t}) \!+\! \gamma\! \sum_{\mathbf{s}_{t+1}}\!\operatorname{Pr}(\mathbf{s}_{t+1}|\mathbf{s}_{t}, \!\mathbf{a}_{t}) \min_{\mathbf{a}_{t+1}\in \mathcal{A}} Q \!\left(\mathbf{s}_{t+1},\! \mathbf{a}_{t+1}\right)\!.\!\!\!
    \label{eq:Bellman_Q}
\end{equation}
The optimal action given the optimal policy $\pi^*(\cdot)$ is
\begin{equation}
    \!\!\!\mathbf{a}_{t}^{*} 
    \triangleq \pi^{*}(\mathbf{s}_{t}) 
    = \mathop{\arg\min}_{\mathbf{a}_{t}\in \mathcal{A}} Q (\mathbf{s}_{t}, \mathbf{a}_{t}).
    \label{eq:optimal V action}
\end{equation}
Then, the optimal state-value function, also called the optimal V function, is defined as
\begin{equation}
    \upsilon^{*}(\mathbf{s}_{t}) = Q(\mathbf{s}_{t}, \mathbf{a}^*_{t}),
    \label{eq: optimal state-value function}
\end{equation}
which is the minimum expected discounted sum of the future cost starting in state $\mathbf{s}_{t}$ under the optimal policy $\pi^*(\cdot)$.
Based on~\eqref{eq:Bellman_Q},~\eqref{eq:optimal V action} and~\eqref{eq: optimal state-value function}, the following inequality holds:
\begin{equation}
    Q(\mathbf{s}_{t}, \mathbf{a}_{t}) \geq \upsilon^{*}(\mathbf{s}_{t}).
    \label{ineq:V&Q}
\end{equation}

Conventional MDP algorithms, such as value iteration and policy iteration, solve MDP problems by computing the optimal V function $\upsilon^{*}(\cdot)$ in~\eqref{eq: optimal state-value function} or the optimal policy $\pi^*(\cdot)$ in~\eqref{eq:optimal V action}. However, these methods are highly computationally complex for large state and action spaces.
As a result, conventional algorithms are infeasible for solving the formulated MDP problem, even for relatively small systems—for example, a 10-device-5-channel system with infinite states and $N!/(N-M)! = 30240$ actions~\cite{puterman2014markov}.

\section{Structural Properties of Optimal Policy} \label{sec: proof of structure}
%We derive some structural properties of the optimal V function (including monotonicity in Theorem~\ref{theo: mono channel}, and convexity in Theorem~\ref{theo: convexity, 2-1} and~\ref{theo: convexity, N-M}) and the optimal scheduling policy when devices are connected by a gateway (a greedy structure in Theorem~\ref{theo: structure gateway}).
In this section, we derive structural properties of the optimal V function and the optimal scheduling policy, which will be leveraged in the design of advanced DRL algorithms in Section~\ref{sec: DRL}.
Similar to the MDP formulation in Section~\ref{subsec: mdp fomulation}, we use $\mathbf{a}$, $\mathbf{s}$, and $\mathbf{s}^{+}$ to denote $\mathbf{a}_{t}, \mathbf{s}_{t},$ and $\mathbf{s}_{t+1}$, respectively, for simplicity in notation.

\subsection{Monotonicity of Optimal V function} \label{subsec: V function mono}
In our earlier work, we established the following result about the monotonicity of the optimal V function w.r.t. the AoI state:
\begin{lemma}[Monotonicity of optimal V function w.r.t. AoI states~\cite{chen2022seDRL}]\label{lemma:monotone V}
    Consider states $\mathbf{s} = (\bm{\delta}, \mathbf{G})$ and $\mathbf{s}'_{\text{AoI}} = (\bm{\delta}'_{(n)}, \mathbf{G})$, where $\bm{\delta}'_{(n)}$ is identical to $\bm{\delta}$ except for the $n$th AoI state, which is $\delta'_{n}$, and $\delta'_{n}\geq \delta_{n}$, then, the optimal V function holds the inequality:
    \begin{equation}
        \upsilon^{*}(\mathbf{s}'_{\text{AoI}}) \geq \upsilon^{*}(\mathbf{s}).
    \end{equation}
\end{lemma}

We now prove that the optimal V function is also monotonically decreasing in terms of the channel states as below.

\begin{theorem}[Monotonicity of the optimal V function w.r.t. channel states]\label{theo: mono channel}
    Consider states $\mathbf{s} = (\bm{\delta}, \mathbf{G})$ and $\mathbf{s}'_{\text{Ch}} = (\bm{\delta}, \mathbf{G}'_{(n,m)})$, where $\mathbf{G}$ and $\mathbf{G}'_{(n,m)}$ are identical except for the element in the $n$th row and $m$th column $g_{n,m} < g'_{n,m}$. The optimal V function holds the inequality:
    \begin{equation}
         \upsilon^{*}(\mathbf{s}'_{\text{Ch}}) \geq \upsilon^{*}(\mathbf{s}).
         \label{eq: mono channel}
    \end{equation}
\end{theorem}

\begin{proof}
    To prove Theorem~\ref{theo: mono channel} based on~\eqref{ineq:V&Q}, it is sufficient to prove
    \begin{equation}
        Q(\mathbf{s}'_{\text{Ch}}, \mathbf{a}^{*}) \geq Q(\mathbf{s}, \mathbf{a}^{*}),
        \label{eq: mono channel, Q}
    \end{equation}
    where $\mathbf{a}^{*}$ is the optimal action given the state $\mathbf{s}$, i.e., $\mathbf{a}^{*} = \pi^{*}(\mathbf{s})$.
    From~\eqref{eq: transition p 2}, we have the transition probability of the AoI state
    \begin{align}
        \operatorname{P}(\bm{\delta}^{+}|\bm{\delta}, \mathbf{G}, \mathbf{a}) 
        & = \operatorname{P} (\delta_{n}^{+}|\delta_{n}, \mathbf{g}_{n},  a_{n}) \\
        & \quad \times \operatorname{P} \left(\bm{\delta}_{\backslash\!\{n\}}^{+}|\bm{\delta}_{\backslash\!\{n\}}, \!\mathbf{G}_{\backslash\!\{n\}}, \!\mathbf{a}_{\backslash\!\{n\}}\right)\!, 
        \label{eq: transition delta}
    \end{align}
    where $\mathbf{a}_{\backslash\!\{n\}} \triangleq (a_{1}, \dots, a_{n-1}, a_{n+1}, \dots, a_{N})$ and $\bm{\delta}_{\backslash\!\{n\}} \triangleq (\delta_{1}, \dots, \delta_{n-1}, \delta_{n+1}, \dots, \delta_{N})$ denote all actions and AoI states without the device $n$, respectively,  and $\mathbf{G}_{\backslash\!\{n\}} \triangleq (\mathbf{g}_{1}, \dots, \mathbf{g}_{n-1}, \mathbf{g}_{n+1}, \dots,  \mathbf{g}_{N})$.
    By substituting~\eqref{eq: optimal state-value function} and~\eqref{eq: transition delta} in the right-hand side of~\eqref{eq:Bellman_Q}, we derive that
    \begin{align}
        Q(\mathbf{s}, \mathbf{a}) &= c(\mathbf{s}) \!+\! \gamma \sum_{\mathbf{G}^{+}} \sum_{\bm{\delta}^{+}} \operatorname{P} (\mathbf{G}^{+}) \operatorname{P} (\bm{\delta}^{+}|\bm{\delta}, \!\mathbf{G}, \mathbf{a}) \upsilon^{*}(\mathbf{s}^{+}) \\
         & = c(\mathbf{s}) + \gamma \sum_{\mathbf{G}^{+}} \sum_{\delta_{n}^{+}}  \sum_{\bm{\delta}_{\backslash\!\{n\}}^{+}} \operatorname{P} (\mathbf{G}^{+}) \operatorname{P} (\delta_{n}^{+}|\delta_{n}, \mathbf{g}_{n}, a_{n}) \!\! \\
        & \qquad \times \operatorname{P} (\bm{\delta}_{\backslash\!\{n\}}^{+}|\bm{\delta}_{\backslash\!\{n\}}, \mathbf{G}_{\backslash\!\{n\}}, \!\mathbf{a}_{\backslash\!\{n\}}) \upsilon^{*}(\mathbf{s}^{+}).
        \label{eq: Q i -i}
    \end{align}
    To simplify the notation, we denote $\mathbf{G}'_{(n,m)}$ with $\mathbf{G}'$ and proceed to prove~\eqref{eq: mono channel, Q} by considering different cases of the optimal action $\mathbf{a}^{*}$.
    \begin{itemize}
        \item[(a)] If $a^{*}_{n} \neq m$, then
        \begin{align}
            & Q(\mathbf{s}'_{\text{Ch}}, \mathbf{a}^{*}) - Q(\mathbf{s}, \mathbf{a}^{*}) \\
            & = \left[c\left(\mathbf{s}'_{\text{Ch}}\right) - c\left(\mathbf{s}\right) \right] \\
            & \quad + \gamma \bigg[ \sum_{{\mathbf{G}'}^{+}} \sum_{\bm{\delta}^{+}} \operatorname{P} ({\mathbf{G}'}^{+}) \operatorname{P} (\bm{\delta}^{+}|\bm{\delta}, \mathbf{G}'\!, \mathbf{a}) \upsilon^{*}( \bm{\delta}^{+},\!{\mathbf{G}'}^{+}) \\
            & \quad - \sum_{\mathbf{G}^{+}} \sum_{\bm{\delta}^{+}} \operatorname{P} (\mathbf{G}^{+}) \operatorname{P} (\bm{\delta}^{+}|\bm{\delta}, \mathbf{G}, \mathbf{a}) \upsilon^{*}(\bm{\delta}^{+}, \mathbf{G}^{+})\bigg] \!\!\\
            & = 0,
        \end{align}
        where the last equality holds follows from the facts that $c\left(\mathbf{s}'_{\text{Ch}}\right) = c\left(\mathbf{s}\right)$,
        $\operatorname{P} (\bm{\delta}^{+}|\bm{\delta}, \mathbf{G}'\!, \mathbf{a}) = \operatorname{P} (\bm{\delta}^{+}|\bm{\delta}, \mathbf{G}, \mathbf{a})$ because $a^{*}_{n} \neq m$, and
        \begin{equation}
             \!\!\!\!\!\sum_{{\mathbf{G}'}^{+}} \operatorname{P} ({\mathbf{G}'}^{+}) \upsilon^{*}(\bm{\delta}^{+},\!{\mathbf{G}'}^{+}) \!=\! \sum_{\mathbf{G}^{+}} \operatorname{P} ({\mathbf{G}}^{+}) \upsilon^{*}(\bm{\delta}^{+}, \mathbf{G}^{+})
            \label{eq: H=H'},
        \end{equation}
        which holds under the assumption of i.i.d. fading channels. 
        
        \item[(b)] If $a^{*}_{n} = m$, then we have
        \begin{align}
            % & \ \geq Z(\mathbf{s}, \mathbf{a}^{1}; \upsilon^{0}) - Z(\mathbf{s}'_{\text{Ch}}, \mathbf{a}^{1}; \upsilon^{0}) \\
            & \!\!\!\!Q(\mathbf{s}'_{\text{Ch}}, \mathbf{a}^{*}) - Q(\mathbf{s}, \mathbf{a}^{*})\\
            & \!\!\!\!= \left[c\left(\mathbf{s}'_{\text{Ch}}\right) - c\left(\mathbf{s}\right) \right] \\
            & \!\!+\! \gamma \bigg[ \sum_{{\mathbf{G}}^{+}} \!\sum_{\bm{\delta}_{\backslash\!\{\!n\!\}}^{+}}  \sum_{\delta_{n}^{+}} \operatorname{P} \left({\mathbf{G}}^{+}\right) \!\operatorname{P} \left(\bm{\delta}_{\backslash\!\{\!n\!\}}^{+}|\bm{\delta}_{\backslash\!\{\!n\!\}}, \!\mathbf{G}_{\backslash\!\{\!n\!\}}, \!\mathbf{a}_{\backslash\!\{\!n\!\}}\!\right) \!\!\\
            & \qquad \quad \times \operatorname{P} (\delta_{n}^{+}|\delta_{n},\! \mathbf{g}'_{n},\! a_{n}) \upsilon^{*}(\bm{\delta}^{+}, \mathbf{G}^{+})\\
            & \!\!-\! \sum_{\mathbf{G}^{+}} \!\sum_{\bm{\delta}_{\backslash\!\{\!n\!\}}^{+}}  \sum_{\delta_{n}^{+}}  \operatorname{P} (\mathbf{G}^{+}) \operatorname{P} (\bm{\delta}_{\backslash\!\{\!n\!\}}^{+}|\bm{\delta}_{\backslash\!\{\!n\!\}}, \!\mathbf{G}_{\backslash\!\{\!n\!\}}, \!\mathbf{a}_{\backslash\!\{\!n\!\}}) \\
            & \qquad \quad \times \operatorname{P} (\delta_{n}^{+}|\delta_{n},\! \mathbf{g}_{n},\! a_{n}) \upsilon^{*}(\bm{\delta}^{+}, \mathbf{G}^{+}) \bigg]  \\
            & \geq 0,
        \end{align}
        where the  equality is derived based on~\eqref{eq: Q i -i} and~\eqref{eq: H=H'}, and the  inequality is based on the following:
        \begin{align}
            & \!\!\!\!\!(1\!-\!\psi'_{n,m}) \upsilon^{*}(1, \bm{\delta}_{\backslash\!\{\!n\!\}}^{+}, \!{\mathbf{G}}^{+}) \!+\! \psi'_{n,m} \upsilon^{*}(\delta_{n}\!+\!\!1, \bm{\delta}^{+}_{\backslash\!\{\!n\!\}}, \!{\mathbf{G}}^{+}) \!\! \\
            & \!\!\!\!\!\!\geq \!(1\!-\!\psi_{n,m}) \upsilon^{*}(1, \bm{\delta}_{\backslash\!\{\!n\!\}}^{+}, \!{\mathbf{G}}^{+}) \!+\! \psi_{n,m} \upsilon^{*}(\delta_{n}\!\!+\!\!1, \!\bm{\delta}^{+}_{\backslash\!\{\!n\!\}},\! {\mathbf{G}}^{+}). \!\!\!
        \end{align}
        This holds due to $\psi'_{n,m} \geq \psi_{n,m}$, $a^{*}_{n} = m$ and Lemma~\ref{lemma:monotone V}.
    \end{itemize}
\end{proof} 

From Lemma~\ref{lemma:monotone V} and Theorem~\ref{theo: mono channel}, the optimal V function monotonically increases with both the AoI and channel states. 

\subsection{Convexity of Cost function and Optimal V function} \label{subsec: V function convex}
Since the input state of the optimal V function takes only discrete values, we define its convexity as below.
\begin{definition}[Discrete convexity of optimal V function and cost function w.r.t. AoI]\label{def: convexity}
    Consider states $\mathbf{s} = (\bm{\delta}, \mathbf{G})$, $\mathbf{s}'_{\text{AoI}} = (\bm{\delta}'_{(n)}, \mathbf{G})$, and $\mathbf{s}''_{\text{AoI}} = (\bm{\delta}''_{(n)}, \mathbf{G})$, where $\bm{\delta} = (\delta_1, \dots, \delta_{n}, \dots, \delta_N)$, $\bm{\delta}'_{(n)} = (\delta_1, \dots, \delta'_{n}, \dots, \delta_N)$, $\bm{\delta}''_{(n)} = (\delta_1, \dots, \delta''_{n}, \dots, \delta_N)$, and $\delta'_{n} \geq \delta_{n} \geq \delta''_{n}$. The cost function and the optimal V function, exhibiting convexity, are defined as satisfying the following inequalities:
    \begin{align}
        \alpha c(\mathbf{s}''_{\text{AoI}}) + (1-\alpha) c(\mathbf{s}'_{\text{AoI}}) & \geq c(\mathbf{s}),
        \label{ineq: def cost convexity}
        \\
                \alpha \upsilon^{*}(\mathbf{s}''_{\text{AoI}}) + (1-\alpha) \upsilon^{*}(\mathbf{s}'_{\text{AoI}}) & \geq \upsilon^{*}(\mathbf{s}), \label{ineq: def value convexity}
        % & \ \forall \alpha \in [0, 1], \ \alpha\delta''_{n} + (1-\alpha)\delta'_{n} = \delta_{n},
    \end{align}
    for any $n\in\{1,\dots,N\}$, where $\alpha \in [0, 1]$ and $\alpha\delta''_{n} + (1-\alpha)\delta'_{n} = \delta_{n}$.
\end{definition}

\subsubsection{Cost function Convexity}

It is important to highlight that cost functions are often convex in practical applications. This implies that the cost can grow increasingly rapidly as the AoI state increases. For problems that aim to optimize overall AoI performance, the cost function is typically a linear function of AoI, which satisfies the convexity property defined above.
In the context of the remote state estimation problem presented in Example~\ref{example: remote estimation}, we provide rigorous proof below to demonstrate that the cost function exhibits convexity when the AoI becomes large.

\begin{lemma}[Asymptotic convexity of the cost function w.r.t. AoI in a remote state estimation system of Example~\ref{example: remote estimation}]\label{lemma: cost convex}
    \textbf{1) The convexity of each device's cost function:} For device $n$, the cost function 
    \begin{equation}
        c_n(\delta) = \operatorname{Tr}\left(h^{\delta}_n(\bar{\mathbf{P}}_n)\right), \label{ineq: each cost convexity}
    \end{equation}
    as defined in~\eqref{eq: cost function},  is asymptotically convex, i.e., the inequality
    $\alpha c_n(\delta') + (1-\alpha) c_n(\delta'') \geq c_n(\delta)$ holds  when $\alpha \in [0, 1]$ and $\alpha\delta''_{n} + (1-\alpha)\delta'_{n} = \delta_{n}$, under the condition $\delta' \geq \delta \geq \delta'' \gg 1$.
    \textbf{2) The convexity of the overall cost function:} For states $\mathbf{s}$, $\mathbf{s}'_{\text{AoI}}$, and $\mathbf{s}''_{\text{AoI}}$ defined in Definition~\ref{def: convexity}, the inequality~\eqref{ineq: def cost convexity} holds under the condition  $\delta'_{n} \geq \delta_{n} \geq \delta''_{n} \gg 1$.
\end{lemma}

\begin{proof}
    See Appendix~\ref{proof: cost convexity}.
\end{proof}
    
\subsubsection{Optimal V function Convexity}
For a system with two devices and one channel, the convexity of the optimal V function is rigorously proven as follows:
\begin{theorem}[Convexity of the optimal V function w.r.t. AoI of a two-device-one-channel systems]\label{theo: convexity, 2-1}
    The optimal V function $\upsilon\left(\cdot\right)$ of a two-device-one-channel system is convex, provided the cost function satisfies convexity.
\end{theorem}

\begin{proof}
    See Appendix~\ref{proof: convexity, 2-1}.
\end{proof}

For a general system with multiple devices and multiple channels, proving the convexity becomes challenging due to the increased dimensionality of the state and action spaces. This higher dimensionality leads to a more complex set of transition states, making it difficult to directly verify the convexity property across all possible transitions. Instead, we establish the asymptotic convexity of the optimal V function as follows:
\begin{theorem}[Asymptotic convexity of the optimal V function w.r.t. AoI of a multi-device-multi-channel system]\label{theo: convexity, N-M}
    Consider states $\mathbf{s}$, $\mathbf{s}'_{\text{AoI}}$, and $\mathbf{s}''_{\text{AoI}}$ defined in Definition~\ref{def: convexity} with $\delta'_{n} \geq \delta_{n} \gg \delta''_{n}$.
Then, the optimal V function $\upsilon(\cdot)$ of a multi-device-multi-channel system exhibits asymptotic convexity for large AoI states, provided the cost function  $c\left(\cdot\right)$ is convex.
\end{theorem}

\begin{proof}
    See Appendix~\ref{proof: convexity, N-M}.
\end{proof}

The asymptotic convexity in Theorem~\ref{theo: convexity, N-M} is evaluated when the AoI states $\mathbf{s}'_{\text{AoI}}$ and  $\mathbf{s}$ are significantly larger compared to the reference state $\mathbf{s}''_{\text{AoI}}$.
In a special case where the devices have identical channel states (i.e., are co-located), we further establish the asymptotic convexity of the optimal value function when the evaluated states $\mathbf{s}'_{\text{AoI}}$,   $\mathbf{s}$ and $\mathbf{s}''_{\text{AoI}}$ all correspond to large AoI values. This result is formalized below:
\begin{proposition}[Asymptotic convexity of the optimal V function w.r.t. AoI for co-located devices]\label{prop: V convexity}
Given states $\mathbf{s}$, $\mathbf{s}'_{\text{AoI}}$, and $\mathbf{s}''_{\text{AoI}}$ defined in Definition~\ref{def: convexity} with $ \delta'_{n} \geq \delta_{n} \geq \delta''_{n} \gg 1$, and assuming that the devices experience identical channel conditions, the optimal V function $\upsilon(\cdot)$ of a multi-device-multi-channel system is convex, provided the cost function  $c\left(\cdot\right)$ is convex.
\end{proposition}

\begin{proof}
    See Appendix~\ref{proof: prop convexity}
\end{proof}

\begin{remark}[Why not channel state convexity?]
The convexity of the optimal V function with respect to channel states has not been derived because it is neither meaningful nor necessary in this context. The cost function, and thus the optimal V function, fundamentally depends on the AoI values rather than the channel states. Channel states play an indirect role, and in systems with independent and fluctuating channels, their specific values often become irrelevant, especially when a device is not using a particular channel. Additionally, analyzing convexity with respect to channel states would require comparing an overwhelming number of state combinations, making it impractical and adding no significant insight. Focusing on AoI, which directly impacts the system's performance, provides a more relevant and useful understanding.
\end{remark}

\subsection{Monotonicity of Optimal Policy} \label{subsec: action mono}
In addition to the properties of the optimal V function, our earlier work establishes the following monotonicity of the optimal policy w.r.t. the channel state:

\begin{theorem}[Monotonicity of optimal policy w.r.t. channel states~\cite{chen2022seDRL}]\label{theo: mono action ch}
    Consider states $\mathbf{s} = (\bm{\delta}, \mathbf{G})$ and $\mathbf{s}''_{\text{Ch}} = (\bm{\delta}, \mathbf{G}''_{(n,m)})$, where $\mathbf{G}$ and $\mathbf{G}''_{(n,m)}$ are identical except for the element in the $n$th row and $m$th column $g_{n,m} \geq g''_{n,m}$, and the corresponding optimal actions are $\mathbf{a}^{*}$ and ${\mathbf{a}''_{\text{Ch}}}^{*}$, respectively. 
    If $a^{*}_{n} = m \neq 0$, then the optimal action ${\mathbf{a}''_{\text{Ch}}}^{*}$ satisfies the following equality:
    \begin{equation}
        {a''_{\text{Ch},n}}^{*} = m.
    \end{equation}
\end{theorem}

This monotonicity demonstrates that if the optimal action for device $n$ is to schedule it to channel $m$ for state $\mathbf{s}$, then 
for another state $\mathbf{s}''_{\text{Ch}}$, where channel $m$ of device $n$ has better quality while all other state components remain identical to $\mathbf{s}$, device $n$ should still be scheduled to channel $m$.

Please note that we have also developed optimal policy monotonicity in terms of AoI in~\cite{chen2022seDRL} but only for some special cases, e.g., a two-device-single-channel scenario. Since no general results have been derived, we will not present them here or consider them in the design of our DRL algorithm in the subsequent section.

\subsection{Greedy Structure of Optimal Policy} \label{subsec: greedy structure}
In this section, beyond analyzing the properties of the optimal V function, we aim to establish the structure of the optimal scheduling policy. Deriving the structure of the optimal scheduling policy for a general multi-sensor, multi-channel system is not feasible due to the complexity of the problem. Instead, we focus on a special case involving co-located devices with identical channel states. This simplification, which focuses solely on the AoI states of different devices while disregarding variations in their channel states, allows us to address the problem more tractably and extract meaningful insights.

To achieve this, we first define the mandatory scheduling set as follows:

\begin{definition}[Mandatory scheduling set]\label{def: priority set}
Consider an $N$-device-$M$-channel system. If there exists a threshold $\bar{\delta}$ such that the asymptotic cost function satisfies the following ordering inequality:
    \begin{equation}
        c_{i_{1}} (\delta) \geq c_{i_{2}} (\delta) \dots \geq c_{i_{N}} (\delta), \ \forall \delta \geq \bar{\delta},
        \label{eq: asymptotic cost constraint}
    \end{equation}
    where $i_n\in \{1,\dots,N\}$ represents an device index, then the following holds:
    
    Given the AoI state $(\delta_{1}, \dots, \delta_{N})$, if there exists a largest number $\bar{N} \leq M$ such that the set $\mathcal{I} \triangleq \{i_1,\dots,i_{\bar{N}}\}$ includes the $\bar{N}$ devices with the largest AoI states, each greater than $\bar{\delta}$, then $\mathcal{I}$ is defined as the mandatory scheduling set.
\end{definition}
The mandatory scheduling set is time-varying due to the dynamics of the AoI states. If the set exists, it is intuitive that all devices within it should be scheduled. This is because the instantaneous cost of scheduling any device in the set exceeds that of any device outside the set. Moreover, leaving a device in the set unscheduled keeps resulting in a higher instantaneous cost than scheduling a device not belonging to the set, thereby increasing the future long-term cost.
Consequently, the optimal scheduling action aligns with a greedy action, which seeks to minimize the immediate cost at each time step. This alignment with a greedy action justifies referring to this structure as the \textbf{greedy structure of the optimal policy}. The result and its detailed proof are provided below.

\begin{theorem}[Asymptotic greedy structure of the optimal scheduling policy for co-located devices]\label{theo: structure gateway}
Consider a multi-device-multi-channel system with co-located devices. If the mandatory scheduling set in Definition~\ref{def: priority set} exists, the optimal policy schedules all devices within the set, i.e., $a^{*}_{n} \neq 0, \forall n \in \mathcal{I}$.
\end{theorem}

\begin{proof}
    See Appendix~\ref{proof: structure gateway}.
\end{proof}

\section{Structure-guided unified on-off policy DRL} \label{sec: DRL}
In this section, we leverage the theoretical results obtained to develop a structure-guided unified dual on-off policy (SUDO) DRL method. This approach combines the strengths of both off-policy and on-policy DRL, utilizing the state-of-the-art on-policy PPO algorithm, widely regarded as one of the most advanced DRL methods available.
First, we briefly overview PPO, which is less familiar than commonly used off-policy DRL methods. Next, we present the proposed SUDO algorithm.

\subsection{Overview of PPO Algorithm} \label{subsec: ppo}
A PPO agent has two neural networks (NNs): an actor NN and a critic NN. \textbf{The actor NN}, with the parameter set $\bm{\varphi}$, approximates the original deterministic policy $\pi^{*}(\mathbf{s})$ by a stochastic policy $\pi(\tilde{\mathbf{a}} | \mathbf{s};\bm{\varphi})$, which outputs a probability distribution over actions $\tilde{\mathbf{a}}$ given the state $\mathbf{s}$.
Note that in our original scheduling problem, the action $\mathbf{a}$ is selected from the discrete action space of size $N!/(N-M)!$. Here to apply the PPO algorithm, which operates in a continuous action space, we implement an action mapping method~\cite{pang2022drl}. This approach maps the $N$-dimensional continuous action $\tilde{\mathbf{a}}$ generated by the actor NN into a corresponding discrete action $\mathbf{a}$. 
For simplicity in notation, the process of obtaining $\mathbf{a}$ from the actor NN $\bm \varphi$ is represented as a stochastic function:
\begin{equation}\label{eq:f()}
 \mathbf{a} = f(\mathbf{s};\bm{\varphi}). 
\end{equation}
\textbf{The critic NN}, with the parameter set $\bm{\nu}$, approximates the optimal V function $\upsilon^{*} (\mathbf{s})$ as $\upsilon(\mathbf{s}; \bm{\nu})$, outputting the estimated value of the optimal V function for a given state $\mathbf{s}$.
Training a PPO agent involves \textbf{two iterative steps}: generating an experience trajectory and updating both NNs.

\underline{Step 1: Experience generation.} The PPO agent generates a trajectory of length $T$, resulting in the trajectory:
\begin{equation}
    \mathcal{T}_{\text{On}} \triangleq \{(\mathbf{s}_{t}, \tilde{\mathbf{a}}_{t}, c_{t}) \}_{t=0}^{T-1}.
    \label{eq: definition Tu}
\end{equation}
At each time step $t$, the actor NN uses the current stochastic policy $\pi(\tilde{\mathbf{a}}_{t}|\mathbf{s}_{t};\bm{\varphi}_{\text{old}})$ to sample a continuous action$\tilde{\mathbf{a}}_{t}$, which is then mapped to the discrete (real) scheduling action $\mathbf{a}_{t}$.
The next state $\mathbf{s}_{t+1}$ and the cost $c_{t}$  are obtained by executing the real action$\mathbf{a}_{t}$. The critic NN computes the estimated optimal V function of the state $\upsilon(\mathbf{s}_{t}; \bm{\nu})$.
Using the trajectory $\mathcal{T}_{\text{On}}$, the advantage function $A_{t}$ and the cost-to-go function $C_{t}$ are calculated as
\begin{align}
    A_{t} &= \!\sum_{\tilde{t}=0}^{T-t-1}(\gamma \lambda)^{\tilde{t}} \zeta_{t+\tilde{t}},
    \label{eq: advantage function}
    \\
    C_{t} &= c_{t} + \gamma \upsilon(\mathbf{s}_{t+1}; \bm{\nu}),
    \label{eq: cost to go}
\end{align}
where $\lambda$ is the generalized advantage estimation (GAE) parameter, and $\zeta_{t} = c_{t} + \gamma \upsilon(\mathbf{s}_{t+1}; \bm{\nu}) - \upsilon(\mathbf{s}_{t}; \bm{\nu})$.
The trajectory is then updated as
\begin{equation}
    \mathcal{T}'_{\text{On}} \triangleq \{(\mathbf{s}_{t}, \tilde{\mathbf{a}}_{t}, A_{t}, C_{t}) \}_{t=0}^{T-1}.
    \label{eq: definition Tu'}
\end{equation}

\underline{Step 2: NN update.} To update the actor and critic NNs, the PPO agent randomly samples $B_{1}$ data points from $\mathcal{T}'_{\text{On}}$ to create a mini-batch dataset:
\begin{equation}
    \left\{ (\mathbf{s}_{l}, \tilde{\mathbf{a}}_{l}, A_{l},C_{l}) \right\}_{l=1}^{B_{1}}.
    \label{PPO mini batch}
\end{equation}
For the critic NN, the temporal difference (TD) error is defined as:
\begin{equation}
    \mathsf{TD}_{l} \triangleq C_{l} - \upsilon(\mathbf{s}_{l}; \bm{\nu}),
    \label{eq: TD error}
\end{equation}
and the loss function is given by
\begin{equation}
    L(\bm{\nu}) = \frac{1}{B_{1}} \sum_{l=1}^{B_{1}} \mathsf{TD}^{2}_{l}.
    \label{eq: critic loss PPO}
\end{equation}
For the actor NN, the loss function is defined as:
\begin{align}
    L (\bm{\varphi}) = \frac{1}{B_{1}} \sum_{l=1}^{B_{1}} \Big( & \min \{ q(\mathbf{s}_{l}; \bm{\varphi}) A_{l}, \\
    & \text{clip}\left(q(\mathbf{s}_{l}; \bm{\varphi}), 1-\epsilon, 1+\epsilon \right) A_{l}\} \\
    & - \omega H(\mathbf{s}_{l};\bm{\varphi}) \Big),
    \label{eq: actor loss PPO}
\end{align}
where 
\begin{equation}
    q(\mathbf{s}_{l}; \bm{\varphi}) = \frac{\pi(\tilde{\mathbf{a}}_{l}|\mathbf{s}_{l};\bm{\varphi})}{\pi(\tilde{\mathbf{a}}_{l}|\mathbf{s}_{l};\bm{\varphi}_{\text{old}})}
\end{equation}
is the probability ratio, and
\begin{equation}
    \text{clip}(x, x_{\min}, x_{\max}) = \max\{ \min \{ x, x_{\max} \}, x_{\min} \}
\end{equation}
is a clip function. Here, $\epsilon$ is a clipping hyper-parameter, $\omega$ is the weight for the entropy loss, and
$$H(\mathbf{s}_{l};\bm{\varphi}) = \frac{1}{2} \ln(2\pi \cdot e \cdot \sigma_{l}^2)$$ represents the policy entropy loss used to encourage exploration, where $\sigma_{l}$ is the deviation for action $\tilde{\mathbf{a}}_{l}$ when in state $\mathbf{s}_{l}$ following the current policy.
The clip function ensures stable training by constraining large updates.

Finally, the critic and actor NNs are updated by minimizing their respective loss functions using gradient-based optimization methods, such as the Adam optimizer.

\subsection{Proposed SUDO-DRL} \label{subsec: SUDO DRL}
To leverage the advantages of on-policy DRL, known for its stable training performance, and off-policy DRL, which offers higher sampling efficiency by reusing past data and facilitates better exploration without getting trapped in local minima, the proposed SUDO-DRL algorithm innovatively integrates concepts from both on-policy and off-policy approaches.

Fundamentally, the effectiveness of SUDO-DRL lies in its carefully designed loss functions for the actor and critic NNs. These loss functions combine both on-policy and off-policy components as follows:
\begin{align}
    L_{\text{SUDO}}(\bm{\nu}) &= L_{\text{On}}(\bm{\nu}) + \beta_{1} L_{\text{Off}}(\bm{\nu}) 
    \label{eq: SUDO critic loss}
    \\
    L_{\text{SUDO}}(\bm{\varphi}) &= L_{\text{On}}(\bm{\varphi}) + \beta_{2} L_{\text{Off}}(\bm{\varphi}),
    \label{eq: SUDO actor loss}
\end{align}
where $\beta_{1}$ and $\beta_{2}$ are the hyperparameters to balance the contributions of the on-policy and off-policy loss functions.

In the following, we first present a holistic structural property evaluation framework based on the theoretical results discussed in the previous section.\footnote{Please note that although some of the theoretical results apply only to specific scenarios (e.g., Theorem~\ref{theo: convexity, N-M} holds in an asymptotic setting), we still utilize these structural results in designing SUDO-DRL. This is because these properties, even when holding under limited conditions, provide valuable guidance for improving the general performance and stability of the algorithm. The effectiveness of this approach will be further illustrated through performance improvements in the following section.} Building on this foundation, we then propose methods for constructing the on-policy and off-policy loss functions, respectively.

\subsubsection{Structural Property Evaluation Framework} \label{subsubsec: evaluate NN structure}
For each state-action pair $(\mathbf{s}, \mathbf{a})$, we evaluate the critic NN $\upsilon(\mathbf{s}; \bm{\nu})$ based on the proven structural properties of the optimal V function: monotonicity w.r.t. AoI state (Lemma~\ref{lemma:monotone V}) and channel state (Theorem~\ref{theo: mono channel}), as well as convexity w.r.t. AoI state (Theorem~\ref{theo: convexity, N-M}).
Additionally, we assess the monotonicity of the actor $\pi(\tilde{\mathbf{a}}|\mathbf{s};\bm{\varphi})$'s output action w.r.t. channel state in the vicinity of $\mathbf{s}$ (Theorem~\ref{theo: mono action ch}) using similar penalty metrics.
The overall evaluation framework is illustrated in Fig.~\ref{fig:structural evaluation metric}.

\begin{figure}[t]
    \centering
    \includegraphics[width=1\linewidth]{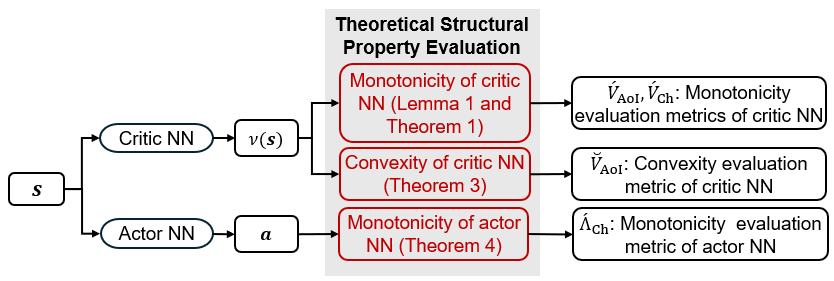}
    \caption{Critic and Actor NNs’ structural property evaluation framework.}
    \label{fig:structural evaluation metric}
\end{figure}
\emph{\underline{Monotonicity of the critic NN.}} For state $\mathbf{s} = (\bm{\delta}, \mathbf{G})$, we define $\acute{V}_{\text{AoI}}$ and $\acute{V}_{\text{Ch}}$ to evaluate the monotonicity of the critic NN w.r.t. AoI state and channel state, respectively:
\begin{align}
    \acute{V}_{\text{AoI}} & = \max\left(0,  \upsilon(\mathbf{s}; \bm{\nu})- \upsilon(\hat{\mathbf{s}}_{(n)}; \bm{\nu})\right)
    \label{eq: V_k,monotonicity aoi}
    \\
    \acute{V}_{\text{Ch}} & =  \max\left(0, \upsilon(\mathbf{s}; \bm{\nu})-\upsilon(\hat{\mathbf{s}}_{(n,m)}; \bm{\nu})\right),
    \label{eq: V_k,monotonicity channel}
\end{align}
where $\hat{\mathbf{s}}_{(n)} = (\hat{\bm{\delta}}_{(n)}, \mathbf{G})$ is identical to $\mathbf{s}$ except for a single-step increase in the $n$-th AoI, i.e.,
\begin{equation}
    \hat{\bm{\delta}}_{(n)} = (\delta_1, \dots, \delta_{n} + 1, \dots, \delta_N),
\end{equation}
and $\hat{\mathbf{s}}_{(n,m)} = (\bm{\delta}, \hat{\mathbf{G}}_{(n,m)})$, where $\hat{\mathbf{G}}_{(n,m)}$ is identical to $\mathbf{G}$ except for the element at the $n$th row and the $m$th column as $\min(g_{n,m} + 1, \bar{g})$.

The monotonicity metrics~\eqref{eq: V_k,monotonicity aoi} and \eqref{eq: V_k,monotonicity channel} indicate that when monotonicity is satisfied, the corresponding metric is zero. However, if monotonicity is violated, the penalty becomes positive and increases proportionally with the extent of the violation.

\emph{\underline{Convexity evaluation of the critic NN.}}
Similarly, the evaluation metric for convexity of the critic NN is defined as:
\begin{align}
    \breve{V}_{\text{AoI}} & =\!  \max\left(0, 2 \upsilon(\mathbf{s}; \bm{\nu}) -(\upsilon(\check{\mathbf{s}}_{(n)}; \bm{\nu}) + \upsilon(\hat{\mathbf{s}}_{(n)}; \bm{\nu})) \right)\!,
    \label{eq: V_k,convexity}
\end{align}
where $\check{\mathbf{s}}_{(n)} = (\check{\bm{\delta}}_{(n)}, \mathbf{G})$ is identical to $\mathbf{s}$ except for a single-step decrease in the $n$th AoI,
\begin{equation}
    \check{\bm{\delta}}_{(n)} = (\delta_1, \dots, \delta_{n} - 1, \dots, \delta_N).
\end{equation}

\emph{\underline{Monotonicity evaluation of the actor NN.}}
As established in Theorem~\ref{theo: mono action ch}, the monotonicity of the actor NN differs from the structural properties of the critic NN, which are evaluated over the entire state vector. Instead, the monotonicity of the actor NN is assessed for each device's action individually.

Given state $\mathbf{s}$ and a corresponding sampled action for device $n$, is $a_n=m \neq 0$, we define the state 
$\check{\mathbf{s}}_{(n,m)} = (\bm{\delta}, \check{\mathbf{G}}_{(n,m)})$, where $\check{\mathbf{G}}_{(n,m)}$ is identical to $\mathbf{G}$ except for the element at the $n$th row and $m$th column is $\max(g_{n,m} - 1, 1)$.
The actor NN's action for device $n$ at state $\check{\mathbf{s}}_{(n,m)}$ is then sampled as $a_{\text{ch},n}$.

The evaluation metric for device $n$'s action monotonicity w.r.t. the channel state is defined as
\begin{equation}    
\acute{\Lambda}_{\text{Ch},n} = \mathbbm{1} \left( a_{n} \neq 0  \text{ and } a_{\text{Ch},n} \neq  a_{n} \right).
    \label{eq: action mono channel}
\end{equation}

\emph{\uline{Sampling over the trajectory for structural property evaluation.}}  
To efficiently evaluate the structural properties of a trajectory, we sample data points from it rather than considering all data points. 
We randomly and uniformly sample \(K\) state-action pairs \(\left((\mathbf{s}_{1},\mathbf{a}_{1}), \dots, (\mathbf{s}_{k},\mathbf{a}_{k}), \dots, (\mathbf{s}_{K},\mathbf{a}_{K})\right)\) from the trajectory. For each state \(\mathbf{s}_{k}\), the structural evaluation metrics defined above, \eqref{eq: V_k,monotonicity aoi}, \eqref{eq: V_k,monotonicity channel},~\eqref{eq: V_k,convexity}, and~\eqref{eq: action mono channel}, require analyzing changes in different device AoI and channel states. To simplify this process, we uniformly sample \(\Xi\) devices for AoI-related evaluations and \(\Xi\) elements from the \(N \times M\) channel state matrix for channel state evaluations.

To account for these sampled points, we introduce subscripts \(k\) and \(\xi\) to the evaluation metrics, i.e., $\acute{V}_{\text{AoI}, k, \xi}$, $\acute{V}_{\text{Ch}, k, \xi}$, $\breve{V}_{\text{AoI}, k, \xi}$, and $\acute{\Lambda}_{\text{Ch}, k, \xi}$.
The sampled data will be used in both the on-policy and off-policy parts.

\emph{\underline{Trajectory structure evaluation.}}
Using the monotonicity evaluation metrics of the critic NN, i.e.,~\eqref{eq: V_k,monotonicity aoi} and \eqref{eq: V_k,monotonicity channel}, we define the \textbf{critic-monotonicity (CM) score}, which is calculated based on the sampled states as:
\begin{equation}
    \mathsf{CM} \triangleq \frac{\sum_{k=1}^{K} \sum_{\xi=1}^{\Xi} \left[ \mathbbm{1}(\acute{V}_{\text{AoI}, k, \xi} = 0) + \mathbbm{1}(\acute{V}_{\text{Ch}, k, \xi} = 0)\right]}{2K\Xi}\times 100.
    \label{eq: monotonicity score V}
\end{equation}
Similarly, based on the convexity evaluation metric in~\eqref{eq: V_k,convexity}, we define the \textbf{critic-convexity (CC) score} as:  
\begin{equation}
    \mathsf{CC} \triangleq \frac{\sum_{k=1}^{K} \sum_{\xi=1}^{\Xi} \mathbbm{1}(\breve{V}_{\text{AoI}, k, \xi} = 0)}{K\Xi}\times 100.
    \label{eq: convexity score V}
\end{equation}  
Then, using the monotonicity evaluation metric of the actor NN from~\eqref{eq: action mono channel}, we define the \textbf{actor-monotonicity (AM) score} as:  
\begin{equation}
    \mathsf{AM} \triangleq \frac{\sum_{k=1}^{K} \sum_{\xi=1}^{\Xi} \mathbbm{1}(\acute{\Lambda}_{\text{Ch}, k, \xi} = 0)}{K\Xi}\times 100.
    \label{eq: actor score}
\end{equation}  

A higher score for $\mathsf{CM}$, $\mathsf{CC}$ and $\mathsf{AM}$ indicates that the trajectory data aligns more closely with the theoretical structural properties of the optimal policy, suggesting that the policy being evaluated is closer to the optimal policy.
These scores will be used in the off-policy part to select trajectories for storage in a replay buffer.

\subsubsection{On-Policy Loss Function}  

The on-policy loss function in SUDO-DRL leverages the current trajectory to create a mini-batch data set of size \(B_{1}\), following the PPO algorithm described in Section~\eqref{subsec: ppo}. However, the key difference lies in the introduction of a penalty term for violations of the structural properties in the critic loss function. This penalty is computed based on \(\{\acute{V}_{\text{AoI}, k, \xi}\}\), \(\{\acute{V}_{\text{Ch}, k, \xi}\}\), and \(\{\breve{V}_{\text{AoI}, k, \xi}\}\) derived from the earlier structural property evaluation samplings.

The loss function for the critic NN in the on-policy component is defined as:  
\begin{align}
    L_{\text{On}}(\bm{\nu})  = \frac{1}{B_{1}} \!\sum_{l=1}^{B_{1}} \mathsf{TD}^{2}_{l} \!+\!  \frac{1}{K\Xi}\sum_{k=1}^{K} \sum_{\xi=1}^{\Xi}  \Big( &\acute{V}_{\text{AoI}, k, \xi}\! + \!\acute{V}_{\text{Ch}, k, \xi} \\
    & \!+\! \breve{V}_{\text{AoI}, k, \xi}\Big),
    \label{eq: critic loss on policy}
\end{align}
where the first term represents the TD loss, and the second term penalizes deviations from the theoretical structural properties. 

The loss function for the actor NN remains the same as the conventional PPO algorithm, as shown in~\eqref{eq: actor loss PPO}:  
\begin{align}
    L_{\text{On}} (\bm{\varphi}) = L (\bm{\varphi}).
    \label{eq: actor loss on policy}
\end{align}
Note that the evaluation metric $\acute{\Lambda}_{\text{Ch},n}$ for the actor NN is based on the executed action after mapping, rather than the action generated directly by the actor NN. Consequently, this metric cannot be directly incorporated into the loss function for training the critic NN.

\subsubsection{Off-policy reply buffer} \label{subsubsec: NN score}
Unlike on-policy DRL, which discards sampled data from old policies entirely, off-policy DRL retains this data in a replay buffer \(\mathcal{R}\) and samples from it for updating the actor and critic NNs.
However, data generated by a policy that is far from optimal can negatively impact training because it introduces bias and instability, hindering the convergence toward the optimal policy. To address this, the off-policy part of the proposed SUDO-DRL framework selectively stores high-quality data that aligns well with the theoretical structural properties of the optimal policy, ensuring more effective and stable training.

\emph{\uline{Structure-Guided Data Storage Scheme.}}  
Given the current trajectory $\{\mathbf{s}_t,\tilde{\mathbf{a}}_t,c_t\}^{T-1}_{t=0}$  with index \(u\), we first calculate the average structure scores of the past \(\bar{u}\) trajectories. The average CM score is computed based on~\eqref{eq: monotonicity score V} as:  
\begin{equation}
    \mathsf{CM}_{\text{Avg}, u} = \frac{1}{\bar{u}} \sum_{\tilde{u}=u-\bar{u}-1}^{u-1} \mathsf{CM}_{\tilde{u}}, 
    \label{eq: average score CM}
\end{equation}
Similarly, the average CC and AM scores are calculated as \(\mathsf{CC}_{\text{Avg}, u}\) and \(\mathsf{AM}_{\text{Avg}, u}\), respectively.  

Next, we define the condition for storing trajectory \(u\) in the replay buffer as:  
\begin{equation}
     \mathsf{CM}_{u} \geq \mathsf{CM}_{\text{Avg}, u},\ \mathsf{CC}_{u} \geq \mathsf{CC}_{\text{Avg}, u},\ \mathsf{AM}_{u} \geq \mathsf{AM}_{\text{Avg}, u}.
     \label{eq: score constraint for store}
\end{equation}  

If the trajectory scores satisfy the constraint~\eqref{eq: score constraint for store}, all transitions (i.e., state-action-cost-next-state tuples) within the trajectory is stored in \(\mathcal{R}\) as:  
\begin{equation}
    \mathcal{X}_{\text{Off}, t} \triangleq \left(\mathbf{s}_{t}, \tilde{\mathbf{a}}_{t}, c_{t}, \mathbf{s}_{t+1}, p_{t}\right), \ t=0,\dots,t-1,
\end{equation}
where \(p_{t}\) represents the \textbf{transition priority indicator}, defined based on the structural scores of the trajectory as:  
\begin{equation}
    p_{u} = \mathsf{CM}_{u} + \mathsf{CC}_{u} + \mathsf{AM}_{u}.
    \label{eq: priority of trajectory}
\end{equation}

\subsubsection{Off-Policy Replay Buffer Sampling and Loss Functions}  
To compute the loss functions for updating the actor and critic NNs, the off-policy component of SUDO-DRL samples a batch of size \(B_{2}\) from the replay buffer \(\mathcal{R}\) based on priority indicators as:
\begin{equation}
    \left\{\mathcal{X}_{\text{Off},b}\right\}_{b=1}^{B_{2}}
\end{equation} 
with the sampling probability of $\mathcal{X}_{\text{Off},b}$ as 
\begin{equation}
    P_{b} \triangleq \frac{p_{b}\cdot \varrho^{b}}{\sum_{b=1}^{R} \left(p_{b} \cdot \varrho^{b}\right)},
    \label{eq: sampling probability}
\end{equation}
where $R$ is the size of the replay buffer $\mathcal{R}$, and $\varrho \in (0, 1]$ is a hyperparameter that controls the decay rate of sampling priority to emphasize more recent trajectories. Trajectories with higher structure scores and greater recency are assigned higher sampling probabilities as determined by $\varrho$.

The loss function for the critic NN in the off-policy component is defined as:
\begin{equation}
    L_{\text{Off}}(\bm{\nu}) = \frac{1}{B_{2}} \sum_{b=1}^{B_{2}} \mathsf{TD}_{b},
    \label{eq: critic loss off-policy}
\end{equation}
where \(\mathsf{TD}_{b}\) is the TD error, as defined in~\eqref{eq: TD error}, and is computed based on the sampled transition \(\mathcal{X}_{\text{Off},b}\).  

The actor NN loss function is designed based on the soft actor-critic (SAC) DRL algorithm, which is an off-policy DRL method. In SAC, the critic NN outputs the state-action value function \(Q(\mathbf{s},\mathbf{a})\) instead of the state value function \(\upsilon(\mathbf{s})\). The state-action value function \(Q(\mathbf{s},\mathbf{a})\) evaluates the long-term expected cost starting from the current state-action pair. To adapt this for the current framework, we approximate \(Q(\mathbf{s},\mathbf{a})\) using \(\upsilon(\mathbf{s})\), as follows:
\begin{equation} \label{eq:Q_approx}
    Q(\mathbf{s},\mathbf{a}) \approx c_{} + \gamma \mathbb{E} \left[ \upsilon(\tilde{\mathbf{s}};\bm{\nu}) \right],
\end{equation}
where \(\tilde{\mathbf{s}}\) represents the next state generated by the environment based on the current state \(\mathbf{s}\) and the action \(\mathbf{a}\) sampled from the actor NN \(\bm{\varphi}\).  

Using this approximation, the actor NN loss function in the off-policy component is expressed as:
\begin{equation}
    L_{\text{Off}}(\bm{\varphi}) = \frac{1}{B_{2}} \sum_{b=1}^{B_{2}} \varpi \log(\pi(\tilde{\mathbf{a}}_b|\mathbf{s}_b;\bm{\varphi})) + \left(c_{b} + \gamma \upsilon(\tilde{\mathbf{s}}_{b+1};\bm{\nu})\right),
    \label{eq: actor loss off-policy}
\end{equation}
where \(\varpi\) is a hyperparameter that weights the entropy term, \(\log(\pi(\tilde{\mathbf{a}}_b|\mathbf{s}_b;\bm{\varphi}))\) is the entropy term encouraging action exploration, and the term \(c_{b} + \gamma \upsilon(\tilde{\mathbf{s}}_{b+1};\bm{\nu})\) approximates the expected long-term cost from \eqref{eq:Q_approx} to reduce computational complexity.

\subsubsection{Structure-Guided Action Selection for Pre-Training} \label{subsubsec: SG action selection}

In the pre-training stage, in addition to the procedures described for the formal training stage, we propose a structure-guided action selection method for the trajectory sampling process. The goal of pre-training is to quickly identify a ``good" initial policy to serve as a starting point for formal training, rather than beginning entirely from scratch.

To achieve this, we leverage the greedy structure outlined in Theorem~\ref{theo: structure gateway} to guide action selection during training. This approach prioritizes AoI differences between devices while disregarding channel state variations. Although effective and computationally efficient for pre-training, this policy is strictly suboptimal and limited to use in this stage.

At each time step, based on the current state and the properties of the system, we determine the mandatory scheduling set \(\mathcal{I}\) as defined in Definition~\ref{def: priority set}. Subsequently, the stochastic policy \(\pi(\tilde{\mathbf{a}}|\mathbf{s};\bm{\varphi})\) generates actions iteratively until either the set \(\mathcal{I}\) becomes empty or all devices in \(\mathcal{I}\) are scheduled, satisfying:
\begin{equation}
    a_{n} \neq 0, \forall n \in \mathcal{I}.
    \label{eq: action requirement}
\end{equation}
The selected virtual action \(\tilde{\mathbf{a}}\) is then stored in the trajectory \(\mathcal{T}_{\text{On}}\) for use during the pre-training stage.

The architecture and details of the SUDO-DRL algorithm are shown in Fig.~\ref{fig:SUDO-DRL} and Algorithm~\ref{alg:SUDO-DRL}, respectively.
\begin{figure}[t]
    \centering
    \includegraphics[width=1\linewidth]{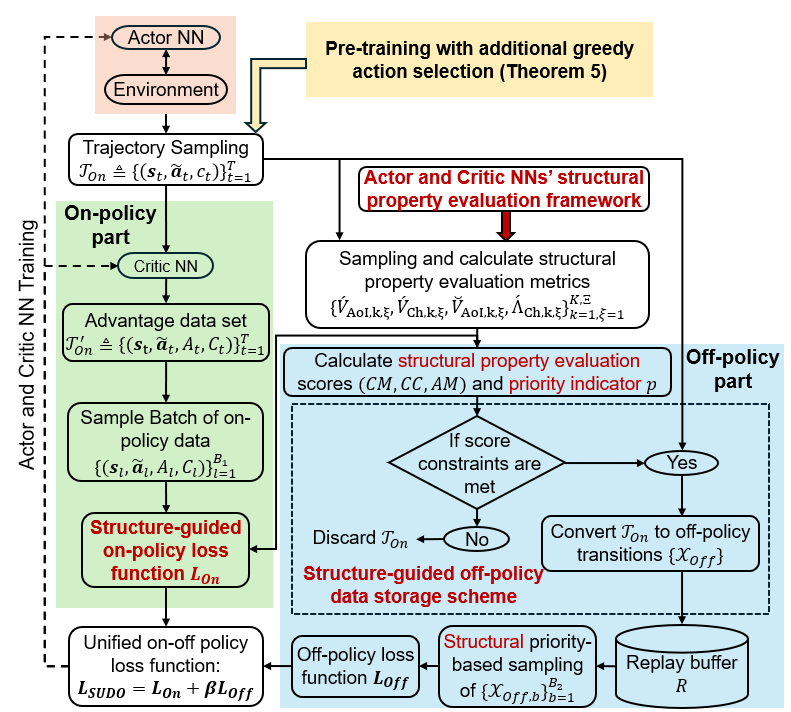}
    \caption{SUDO-DRL Architecture.}
    \label{fig:SUDO-DRL}
\end{figure}

\begin{algorithm}[t]
    \small
    \caption{\small{SUDO-DRL for transmission scheduling}}
    \label{alg:SUDO-DRL}
    \begin{algorithmic}[1]
        \State Initialize the environment with the goal-oriented communication system parameters
        \State Initialize critic and policy network with random weights $\bm{\nu}$ and $\bm{\varphi}$, respectively        
        \For {$\mathsf{episode}= 1, 2, \dots, I$}    
            \Statex \rule{8.2cm}{0.7pt}
            \Statex $\triangleright$ \textbf{Trajectory Sampling}
            \State Initialize state $\mathbf{s}_0$
            \For {$t = 0, 1, \dots, T$}
            \State Generate a trajectory $\mathcal{T}_{\text{On}} \triangleq \{ (\mathbf{s}_{t}, \tilde{\mathbf{a}}_{t}, c_{t})\}_{t=0}^{T-1}$ using the actor NN and the environment. If $\mathsf{episode} < I_1$ (pre-training), apply the structure-guided action selection from Section~\ref{subsec: greedy structure}.
            \EndFor
            \Statex \rule{8.2cm}{0.7pt}
            \Statex $\triangleright$ \textbf{Structural Property Evaluation}
            \State Calculate and generate a batch of structural property evaluation metrics $\{\acute{V}_{\text{AoI}, k, \xi}, \acute{V}_{\text{Ch}, k, \xi}, \breve{V}_{\text{AoI}, k, \xi}, \acute{\Lambda}_{\text{Ch},k,\xi}\}_{k=1, \xi=1}^{K, \Xi}$ based on \eqref{eq: V_k,monotonicity aoi},~\eqref{eq: V_k,monotonicity channel},~\eqref{eq: V_k,convexity}, and~\eqref{eq: actor score}
            \Statex \rule{8.2cm}{0.7pt}
            \Statex $\triangleright$ \textbf{On-policy Part}
            \For {$t = 0, 1, \dots, T$}
                \State Compute the advantage function $A_{t}$ in~\eqref{eq: advantage function} and cost-to-go function $C_{t}$ in~\eqref{eq: cost to go} and generate the data set $\mathcal{T}'_{\text{On}}$ in~\eqref{eq: definition Tu'}
            \EndFor
            \For {$l=1, \dots, B_{1}$}
                \State Sample a random mini-batch of data $ \{(\mathbf{s}_{l}, \tilde{\mathbf{a}}_{l}, A_{l}, C_{l}) \}_{l=1}^{B_{1}}$ from data set $\mathcal{T}'_{\text{On}}$ of the current trajectory
                \State Calculate the on-policy loss function $L_{\text{On}}(\bm{\nu})$ in~\eqref{eq: critic loss on policy} and $L_{\text{On}} (\bm{\varphi})$ in~\eqref{eq: actor loss on policy}
            \EndFor
            \Statex \rule{8.2cm}{0.7pt}
            \Statex $\triangleright$ \textbf{Off-policy Part}
            \State Calculate the structure score of the generated trajectory $\mathsf{CM}$, $\mathsf{CC}$, and $\mathsf{AM}$ according to~\eqref{eq: monotonicity score V},~\eqref{eq: convexity score V}, and~\eqref{eq: actor score}, and the priority indicator $p$ according to~\eqref{eq: priority of trajectory}
            \State Calculate the average structure scores $\mathsf{CM}_{\text{Avg}}$, $\mathsf{CC}_{\text{Avg}}$, and $\mathsf{AM}_{\text{Avg}}$ according to~\eqref{eq: average score CM}
            \State Store transitions $\{\mathcal{X}_{\text{Off},t}\}$ in $\mathcal{R}$ based on $\mathcal{T}_{\text{On}}$ and $p$, if the structure scores satisfy the constraints~\eqref{eq: score constraint for store}
            \For {$b= 1, \dots, B_{2}$}
                \State Sample a batch of transitions $\{\mathcal{X}_{\text{Off},b}\}_{b=1}^{B_{2}}$ based on sampling priority~\eqref{eq: priority of trajectory} from $\mathcal{R}$
                 \State Calculate the off-policy loss functions $L_{\text{Off}}(\bm{\nu})$ in~\eqref{eq: critic loss off-policy} and $L_{\text{Off}} (\bm{\varphi})$ in~\eqref{eq: actor loss off-policy}
            \EndFor
            \Statex \rule{8.2cm}{0.7pt}
            \Statex $\triangleright$ \textbf{Unified dual on-off policy-based parameter updating}
            \State Update $\bm{\nu}$ and $\bm{\varphi}$ by minimizing $L_{\text{SUDO}}(\bm{\nu})$ and $L_{\text{SUDO}}(\bm{\varphi})$ in~\eqref{eq: SUDO critic loss} and ~\eqref{eq: SUDO actor loss}, respectively
        \EndFor
    \end{algorithmic}
\end{algorithm}

% \parbox[t]{0.8\linewidth}

\section{Numerical Experiments} \label{sec: simulation}
In this section, we evaluate and compare the performance of the proposed SUDO-DRL with PPO, the benchmark on-policy DRL, and three off-policy DRL algorithms: DDPG~\cite{lillicrap2015DDPG}, structure-enhanced DDPG (SE-DDPG)~\cite{chen2022seDRL}, and type II monotonicity-regularized DDPG (MRII-DDPG)~\cite{chen2023semantic}. Notably, the latter two algorithms represent state-of-the-art structure-guided DRL approaches for addressing goal-oriented communication scheduling problems.

\subsection{Experiment Setups}
Our numerical experiments were conducted on a computational platform equipped with an Intel Core i7 9700 CPU @ 3.0 GHz, 32GB RAM, and an NVIDIA RTX 3060Ti GPU. The experimental environment is based on a remote state estimation system as described in Example~\ref{example: remote estimation}, with the estimation MSE considered as the performance metric. For this system, the dimensions of the process state and measurement are set to \(r_{n} = 2\) and \(c_{n} = 1\), respectively. The system matrices \(\mathbf{A}_{n}\) are randomly generated with spectral radii uniformly drawn from the range \((1, 1.3)\).

The discrete fading channel state is quantized into \(\bar{g} = 5\) levels, with corresponding packet drop rates set to \(0.2, 0.15, 0.1, 0.05,\) and \(0.01\). These values are derived from the Rayleigh distribution with a scale parameter randomly generated within the range \((0.5, 2)\)~\cite{huang2019deep}.

For a fair comparison, the actor and critic NNs of the SUDO-DRL and benchmark agents are implemented as fully connected NNs, each with three hidden layers, as described in~\cite{pang2022drl}. The input dimension of the actor NN matches the state dimension, i.e., \(N + N \times M\), for all algorithms. The output dimension is configured as \(2N\) for SUDO-DRL and PPO, and \(N\) for the benchmark off-policy DRL algorithms. Regarding the critic NN, the input dimension is \(N + N \times M\) for SUDO-DRL and PPO, while it is \(2N + N \times M\) for the benchmark off-policy DRL algorithms. The output dimension of the critic NN is set to \(1\) for all evaluated algorithms.

The training hyperparameters of the SUDO-DRL and PPO algorithms are summarized in Table~\ref{tab:Setup}.

\begin{table}[t]
    % \vspace{-6.7cm}
    \footnotesize
    \setlength\tabcolsep{7.5pt}
    \centering
    \caption{Summary of Training Hyperparameters}
    \renewcommand{\arraystretch}{1.3}
    \setlength{\tabcolsep}{1mm}
    \label{tab:Setup}
    \begin{tabular}{lr}
         \Xhline{1px}
         \textbf{Hyperparameters of SUDO-DRL and benchmarks} & Value \\
         \hline
         Critic NN learning rate & 0.001 \\
         Actor NN learning rate & 0.0001 \\
         Decay rate of learning rate                & 0.001 \\
         % \rowcolor[HTML]{EFEFEF} 
         % Initial channel exploration probability $\nu$ & 1 \\
         % $\nu$ decay rate                           & 0.99 \\
         % Minimum $\nu$                              & 0.01 \\
         Discount factor, $\gamma$                      & 0.99 \\
         % \rowcolor[HTML]{EFEFEF} 
         GAE parameter, $\lambda$        & 0.99 \\
         Clipping parameter, $\epsilon$                  & 0.2 \\
         Policy entropy loss weight, $\omega$  & 0.01 \\
         Decay rate of sampling priority, $\varrho$                & 0.95 \\
         Unified on-off policy loss function hyperparameter, $\beta_{1}, \beta_{2}$ & 0.9 \\
%         Number of epochs for on-policy part, $E_{1}$ & 2 \\
%         Number of epochs for on-policy part, $E_{2}$ & 2 \\
         On-policy and off-policy batch size, $B_{1}, B_2$ & 128 \\
%         Off-policy batch size, $B_{2}$ & 4 \\
         Number of sampled states for score scheme, $K$ & 50 \\
         Number of tested AoI and channel state, $\Xi$ & 4 \\
         Number of past trajectories for average score computing, $\bar{u}$ & 50 \\
         Time horizon of each episode, $T$ & 128 \\
         Size of replay buffer, $R$ & $200$ \\
         Number of episodes for pre-training, $I_{1}$ & $10 \times N$  \\
         Total number of episodes for training, $I$ & $10000$  \\
         Optimizer during training & Adam \\
         \Xhline{1px}
    \end{tabular}
\end{table}

\subsection{Performance Comparison of Different DRL Algorithms}
\begin{figure}[t]
    \centering
        \includegraphics[width=0.9\linewidth]{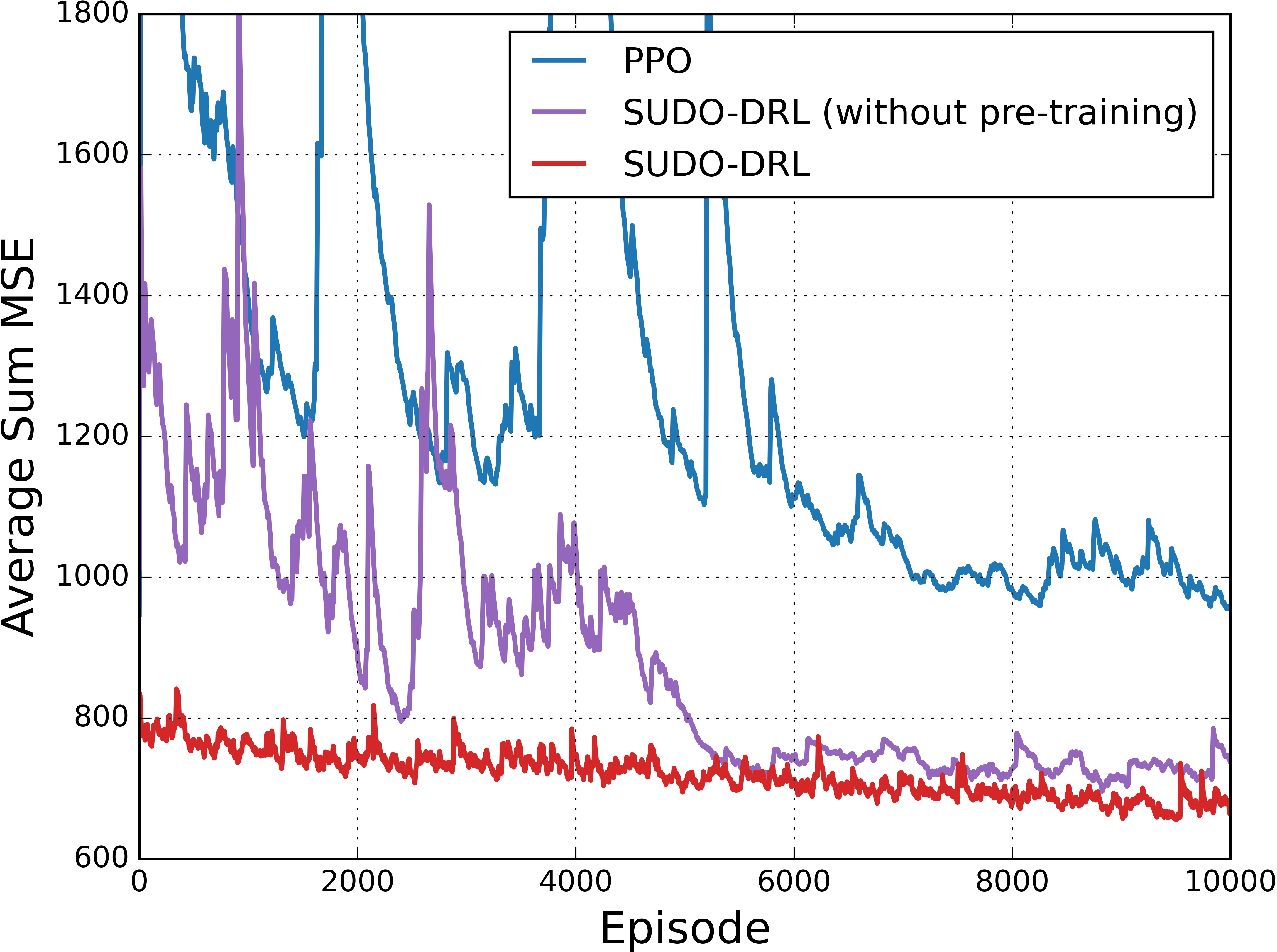}
        \caption{Average cost during training with $N\!\!=\!40, M\!\!=\!20$.}
        \label{fig:training curve DRL}
\end{figure}

\begin{table}[t]
    \centering
    \footnotesize
	\setlength\tabcolsep{7.5pt}
	\centering
    % \vspace{-0.1cm}
	\caption{Emperical average cost of the SUDO-DRL algorithm and the Benchmarks}
	\label{tab:test result}
	\setlength{\tabcolsep}{1.5mm}
	\begin{tabular}{c|c|c|c|c|c|c}
		\hline 
		\thead{System \\ Scale \\ $(N,M)$} &  \thead{Para.} & DDPG & \thead{SE-\\DDPG\\ \cite{chen2022seDRL}} & \thead{MRII- \\DDPG\\ \cite{chen2023semantic}} & PPO & \thead{\textbf{SUDO-DRL}} \\
		\hline
		$(10, 5)$ & 1 & 89.26 & 77.14 & 84.00 & 119.63 & \textbf{85.52} \\
		& 2 & 98.87 & 87.30 & 90.28 & 123.01 & \textbf{95.82} \\
		& 3 & $-$ & 106.60 & 119.55 & 195.37 & \textbf{121.31}  \\
		& 4 & 83.12 & 78.21 & 80.88 & 120.93 & \textbf{80.68} \\
		\hline
		 
		$(20, 10)$ & 5 & $-$ & 357.14 & 369.14 & 569.96  & \textbf{370.63}  \\
            & 6 & $-$ & $-$ & 445.87 & 584.37 & \textbf{426.29} \\
		  & 7 & $-$ & 407.20 & 441.73 & 731.94 & \textbf{417.90} \\
		  & 8 & $-$ & 290.72 & 307.94 & 376.16 & \textbf{308.91} \\
		  \hline

		$(30, 15)$ & 9 & $-$ & $-$ & $-$ & 805.45 & \textbf{519.78} \\
		& 10 & $-$ & $-$ & $-$ & 739.97 & \textbf{575.34}  \\
		& 11 & $-$ & $-$ & $-$ & 900.71 & \textbf{518.03} \\
		& 12 & $-$ & $-$ & $-$ & 901.42 &  \textbf{551.97} \\
		\hline
		$(40, 20)$ & 13 & $-$ & $-$ & $-$ & 1057.05 & \textbf{719.45} \\
		& 14 & $-$ & $-$ & $-$ & 971.35 & \textbf{689.81} \\
		& 15 & $-$ & $-$ & $-$ & 1291.54 & \textbf{994.80} \\
		& 16 & $-$ & $-$ & $-$ & 1012.20 & \textbf{771.91} \\
		\hline 
	\end{tabular}
\end{table}

Fig.~\ref{fig:training curve DRL} illustrates the average sum MSE cost during the training of the proposed SUDO-DRL algorithm, both with and without the pre-training stage, and compares it with the benchmark PPO algorithm under a 40-device-20-channel system setting. Notably, other off-policy benchmarks, such as DDPG-related algorithms, fail to converge in this large-scale setup, further highlighting the robustness of the proposed method.
The results demonstrate that the SUDO-DRL algorithm with pre-training significantly outperforms PPO, reducing the average cost by approximately \(35\%\). Additionally, the pre-training stage enables SUDO-DRL to converge in fewer episodes, achieving over \(40\%\) faster convergence compared to the variant without pre-training.
Moreover, the pre-training-based SUDO-DRL achieves a lower average cost compared to one without pre-training. This indicates that the pre-training phase effectively guides the policy towards a better initialization point, ultimately improving performance and stability during formal training.

In Table~\ref{tab:test result}, we examine the performance of the SUDO-DRL algorithm and benchmark algorithms. The performance is evaluated based on the empirical average MSE over 20,000-step simulations under 16 different system settings (i.e., Para. 1-16). These settings include parameters of dynamic processes for remote estimation and wireless channel statistics, i.e., $\mathbf{A}_{n}$, $\mathbf{C}_{n}$, and $q^{1}_{n,m}, \dots, q^{\bar{g}}_{n,m}$, as well as different system scales, i.e., $(N, M)$.
We observe that DDPG only works for the 10-device-5-channel system setting, while SE-DDPG and MRII-DDPG can converge up to 20-device-10-channel systems. However, both PPO and SUDO-DRL can handle large-scale systems with 40-device-20-channel settings. SUDO-DRL consistently achieves a 25\%-40\% reduction in average MSE compared to PPO, with the performance gap increasing as the system scale grows.
In particular, for small-scale 10-device-5-channel systems, we find that SUDO-DRL achieves performance comparable to advanced off-policy methods such as SE-DDPG and MRII-DDPG but is generally slightly worse. This is because off-policy methods are better suited for converging to optimal solutions in small-scale systems.

We also evaluate the effectiveness of the proposed structural property evaluation scheme during the training of SUDO-DRL, i.e., CM and CC scores for the critic NN and AM scores for the actor NN, as shown in Figs.~\ref{fig: monotonicity score curve DRL}, \ref{fig: convexity score curve DRL}, and \ref{fig: actor score curve DRL}.
For the critic NN monotonicity, we observe that SUDO-DRL achieves a full score (i.e., 100) very quickly, while PPO consistently reaches the full score only after 1000 episodes. For the critic NN convexity, SUDO-DRL guarantees a full score after 200 episodes, whereas PPO remains below 80 until the end of training.
For the actor NN monotonicity, SUDO-DRL ensures the property is satisfied after 2000 episodes, whereas PPO achieves less than 75 and shows no further improvement during training. These results indicate that PPO struggles to fully exploit the structural properties of the optimal policy, particularly in terms of critic convexity and actor monotonicity. This limitation is likely a key reason why SUDO-DRL outperforms PPO.

\begin{figure}[t]
    \centering
        \includegraphics[width=0.9\linewidth]{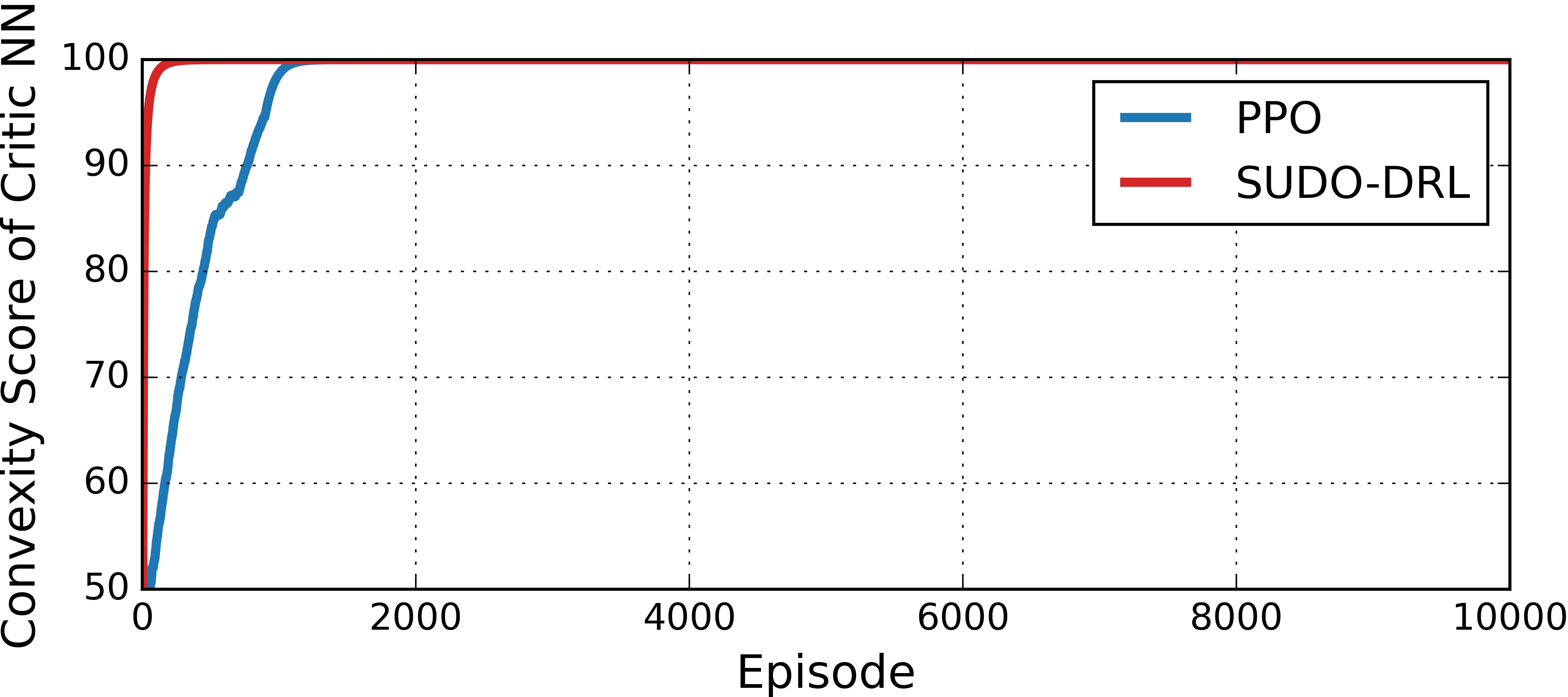}
        \caption{Critic monotonicity (CM) score during training with $N\!\!=\!40, M\!\!=\!20$.}
        \label{fig: monotonicity score curve DRL}
\end{figure}

\begin{figure}[t]
    \centering
        \includegraphics[width=0.9\linewidth]{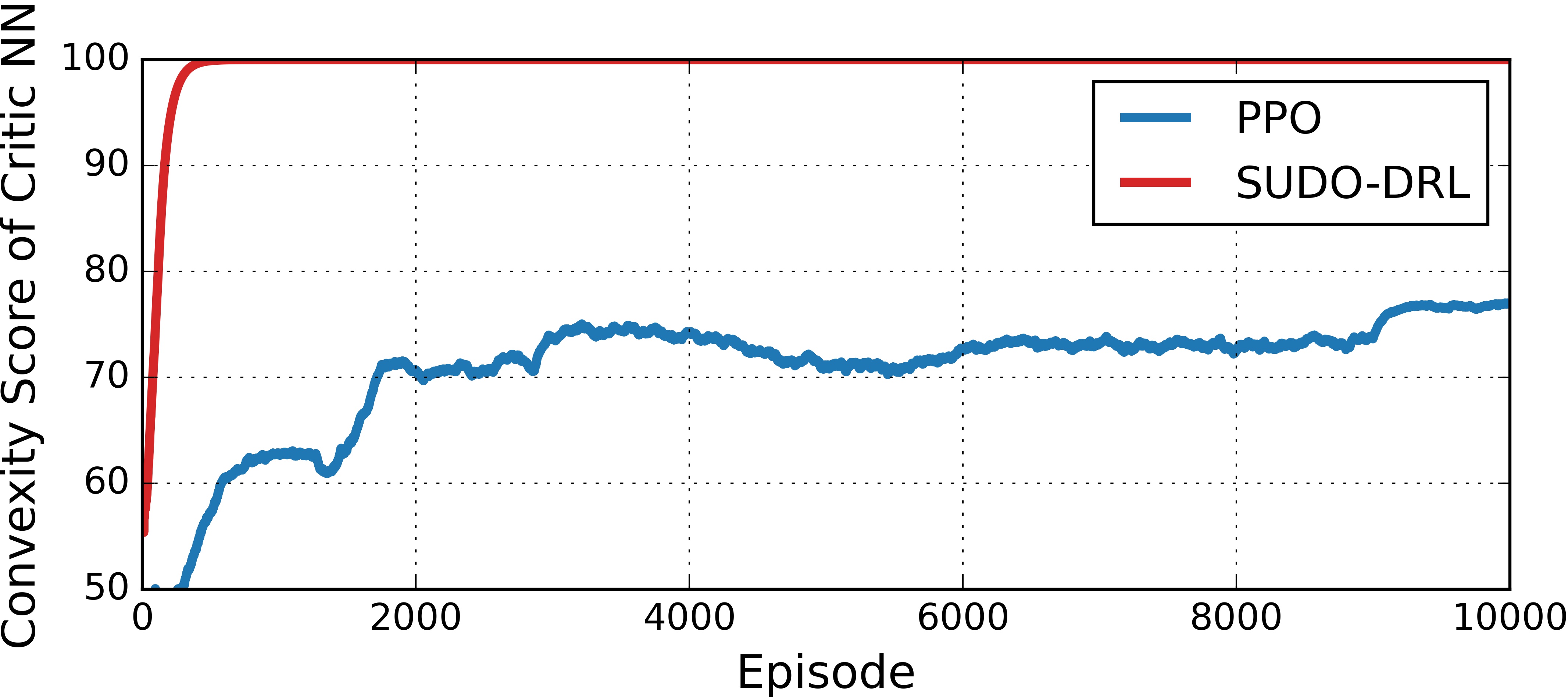}
        \caption{Critic convexity (CC) score during training with $N\!\!=\!40, M\!\!=\!20$.}
        \label{fig: convexity score curve DRL}
\end{figure}

\begin{figure}[t]
    \centering
        \includegraphics[width=0.9\linewidth]{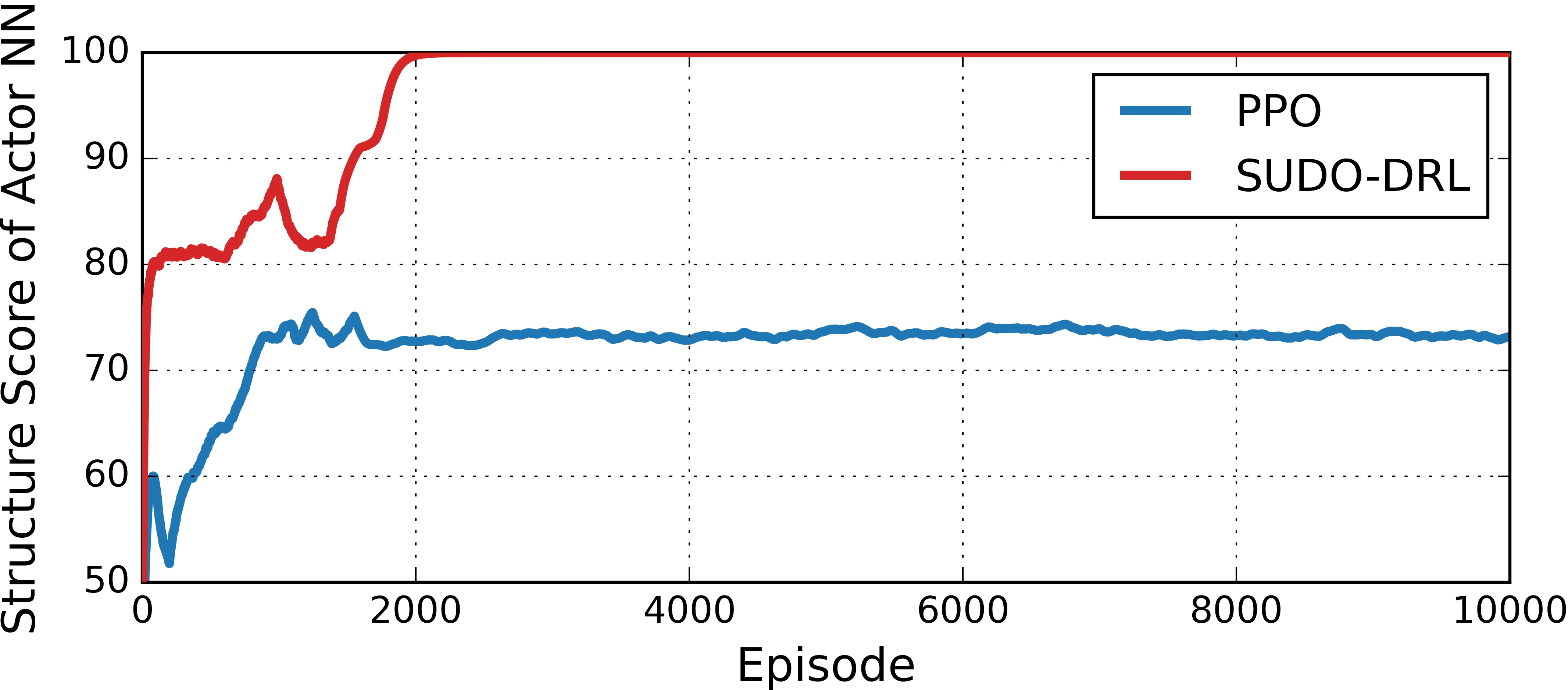}
        \caption{Actor monotonicity (AM) score during training with $N\!\!=\!40, M\!\!=\!20$.}
        \label{fig: actor score curve DRL}
\end{figure}

\section{Conclusion} \label{sec: conclusion}
We have derived key structural properties of the optimal solution to the goal-oriented scheduling problem, establishing monotonicity and asymptotic convexity for the optimal value function and policy. Leveraging these insights, we have developed SUDO-DRL, a hybrid algorithm combining on-policy stability and off-policy efficiency. SUDO-DRL has achieved up to 45\% performance improvement and 40\% faster convergence compared to state-of-the-art methods, while scaling effectively in large systems where other approaches fail. Our work has demonstrated the potential of SUDO-DRL to advance goal-oriented communications.
Future directions include exploring additional structural properties to further enhance the theoretical framework and extending SUDO-DRL to address comprehensive resource allocation problems, such as power allocation and advanced multiple access schemes like NOMA, to broaden its capabilities in goal-oriented communication scheduling.

\appendices
\section{Proof of Lemma~\ref{lemma: cost convex}} \label{proof: cost convexity}
The convexity of the overall cost function holds immediately if the individual cost function does. Therefore, we only need to prove the inequality~\eqref{ineq: each cost convexity}, which is sufficient to establish that  
\begin{equation}
    \operatorname{Tr} (h^{d}(\bar{\mathbf{P}})) + \operatorname{Tr} (h^{d+2}(\bar{\mathbf{P}})) \geq 2 \operatorname{Tr} (h^{d+1}(\bar{\mathbf{P}})).
    \label{ineq: cost P convexity 2}
\end{equation}

To proceed, we derive some linear algebra properties related to \(\bar{\mathbf{P}}\). Based on the properties of a stabilized Kalman filter~\cite{shi2012schedulin}, we have
\begin{equation}
    \bar{\mathbf{P}} = \mathbf{A} \bar{\mathbf{P}} \mathbf{A}^{\top} + \mathbf{W} - \mathbf{K},
\end{equation}
where  
\begin{equation}
\begin{aligned}
    \mathbf{K} = (\mathbf{A} \bar{\mathbf{P}} \mathbf{A}^{\top} + \mathbf{W}) \mathbf{C}^\top  & \left[\mathbf{C} (\mathbf{A} \bar{\mathbf{P}} \mathbf{A}^{\top} +  \mathbf{W}) \mathbf{C}^{\top} + \mathbf{V}\right]^{-1} \\
    &\times \mathbf{C} \left( \mathbf{A} \bar{\mathbf{P}} \mathbf{A}^{\top} + \mathbf{W} \right)^{\top},
\end{aligned}
\end{equation}
and \(\mathbf{K}\) is symmetric and positive definite.  
Thus, we have:  
\begin{equation}
    \bar{\mathbf{P}} - \mathbf{A} \bar{\mathbf{P}} \mathbf{A}^{\top} = \mathbf{W} - \mathbf{K}.
    \label{eq: stable P}
\end{equation}
From~\eqref{eq: stable P}, we derive:  
\begin{equation}
    \mathbf{A}^{2} \bar{\mathbf{P}} \mathbf{A}^{2^{\top}}  + \mathbf{A} \mathbf{W} \mathbf{A}^{\top} = \mathbf{A} \bar{\mathbf{P}} \mathbf{A}^{\top} + \mathbf{A} \mathbf{K} \mathbf{A}^{\top}.
    \label{eq: stable P 2}
\end{equation}

Next, we express \(\mathbf{A}\) in its Jordan normal form:  
\begin{equation}
    \mathbf{A} = \mathbf{F} \mathbf{J} \mathbf{F}^{-1},
    \label{eq: A jordan form}
\end{equation}
where \(\mathbf{J}\) is a block-diagonal matrix composed of Jordan blocks. Each block is a square, triangular matrix with a single eigenvalue of \(\mathbf{A}\) along its diagonal, ones on the superdiagonal, and zeros elsewhere.  

Now, returning to prove~\eqref{ineq: cost P convexity 2}, we proceed as follows:  
\begin{align}
    & \operatorname{Tr} (h^{d}(\bar{\mathbf{P}})) + \operatorname{Tr} (h^{d+2}(\bar{\mathbf{P}})) - 2 \operatorname{Tr} (h^{d+1}(\bar{\mathbf{P}})) \\
    &= \operatorname{Tr} \left[\mathbf{A}^{d} \!\left(\mathbf{A}^{2} \bar{\mathbf{P}} \mathbf{A}^{2^{\top}} + \bar{\mathbf{P}} - 2\mathbf{A} \bar{\mathbf{P}} \mathbf{A}^{\top} + \mathbf{A} \mathbf{W} \mathbf{A}^{\top} - \mathbf{W}\right) \mathbf{A}^{d^{\top}} \right] \\
    &= \operatorname{Tr} \left[\mathbf{A}^{d} \left( \mathbf{A} \mathbf{K} \mathbf{A}^{\top} - \mathbf{K} \right) \mathbf{A}^{d^{\top}} \right],
    \label{eq:remote_estimation_proof_simplify}
\end{align}
where the first equality uses \(h(\mathbf{X}) = \mathbf{A} \mathbf{X} \mathbf{A}^{\top} + \mathbf{W}\) and \(h^{\delta+1}(\cdot) = h(h^{\delta}(\cdot))\). The second equality substitutes~\eqref{eq: stable P} and~\eqref{eq: stable P 2}. 

After further simplification with $\mathbf{A}$'s Jordan form,~\eqref{eq:remote_estimation_proof_simplify} is equivalent to:  
\[
y(d+1) - y(d),
\]
where  
\begin{equation}
    y(d) \triangleq \operatorname{Tr}\left[\mathbf{F} \mathbf{J}^{d} \sqrt{\mathbf{K}} \left(\mathbf{F} \mathbf{J}^{d} \sqrt{\mathbf{K}}\right)^{\top}\right].
\end{equation}

Clearly, since \(\mathbf{A}\) has \(\bar{r}\) eigenvalues, \(y(d)\) can be expressed as the summation of \(\bar{r}\) terms, each with the form \(\lambda_{r}^{2d} \times (\text{polynomial in } d)\), where \(\lambda_{r}\) is the \(r\)th eigenvalue.  

The dominant term corresponds to the largest eigenvalue \(\lambda^{*}\) with a non-zero polynomial coefficient with a sufficiently large \(d\). As \(y(d)\) is strictly non-negative, the associated polynomial term is positive for large \(d\).  

Thus, \(y(d+1) - y(d) > 0\) for \(d \gg 1\), completing the proof.  

\section{Proof of Theorem~\ref{theo: convexity, 2-1}} \label{proof: convexity, 2-1}
Before proceeding further, we need a technical lemma regarding the property of the value function during value iteration.

During the conventional value iteration method, which theoretically guarantees the optimality of convergence~\cite{puterman1990markov}, we denote $\upsilon^{k}(\mathbf{s}) \in \mathcal{V}$ and $\upsilon^{0}(\mathbf{s})\in \mathcal{V}$ as the $k$-th iteration and the initial value function, respectively, where $\mathcal{V}$ is the measurable function set, i.e., $\mathcal{V}: \mathcal{S} \to \mathbb{R}$.
To execute the value iteration, we define a Bellman operation as follows:
\begin{align}
    \upsilon^{k+1} (\mathbf{s}) = c(\mathbf{s}) + \gamma \min_{\mathbf{a} \in \mathcal{A}} \left[ \sum_{\mathbf{s}^{+}} \operatorname{P}(\mathbf{s}^{+}|\mathbf{s}, \mathbf{a}) \upsilon^{k} (\mathbf{s}^{+})\right].
    \label{eq: bellman operator}
\end{align}
The optimality and convergence of value iterations are given below.
\begin{lemma}[Optimality and convergence of value iteration~\cite{puterman1990markov}]\label{lemma:converge of V*}
    If there exists an optimal policy, then the sequence $\{\upsilon^{k}(\cdot)\}$ produced by the Bellman operation~\eqref{eq: bellman operator} converges in norm to the unique optimal value function $\upsilon^{*}(\cdot) \in \mathcal{V}$, i.e., 
    \begin{equation}
        \lim_{k \to \infty} \upsilon^{k}(\cdot) = \upsilon^{*}(\cdot),
    \end{equation}
    for all initial value function $\upsilon^{0}(\cdot) \in \mathcal{V}$.
\end{lemma}

At the $k$th iteration of the value iteration, the action derived based on the previous value function $\upsilon^{k-1}(\cdot)$ is defined as
    \begin{align}
        \mathbf{a}^{k} \triangleq \pi^{k}(\mathbf{s}) = c(\mathbf{s}) + \gamma \mathop{\arg\max}_{\mathbf{a}\in \mathcal{A}}  \left[ \sum_{\mathbf{s}^{+}} \operatorname{P}(\mathbf{s}^{+}|\mathbf{s}, \mathbf{a})\upsilon^{k-1}(\mathbf{s}^{+})\right],
        \label{eq:optimal action V_t}
    \end{align}
where $\pi^{k}(\cdot)$ is the corresponding policy at the $k$th iteration.
    
    In order to prove Theorem~\ref{theo: convexity, 2-1}, we define a $Z$-function at $k$th value iteration similar to the Q-function, $Z(\mathbf{s}, \mathbf{a}; \upsilon^{k}): \mathcal{S} \times \mathcal{A} \times \mathcal{V} \to \mathbb{R}$:
    \begin{equation}
        Z(\mathbf{s}, \mathbf{a}; \upsilon^{k})= c(\mathbf{s}) + \gamma \sum_{\mathbf{s}^{+}} \operatorname{P}(\mathbf{s}^{+}|\mathbf{s}, \mathbf{a}) \upsilon^{k} (\mathbf{s}^{+}),
        \label{eq: Z-function}
    \end{equation}
    which can be derived to the same equation as~\eqref{eq: Q i -i} but replacing $\upsilon^{*}(\mathbf{s}^{+})$ with $\upsilon^{k}(\mathbf{s}^{+})$.
    From~\eqref{eq: bellman operator},~\eqref{eq:optimal action V_t}, and~\eqref{eq: Z-function}, the relationship of the value function and the $Z$-function is given as
    \begin{equation}
        \upsilon^{k+1} (\mathbf{s}) = Z(\mathbf{s}, \mathbf{a}^{k}; \upsilon^{k}) \leq Z(\mathbf{s}, \mathbf{a}; \upsilon^{k}). 
        \label{ineq:V&Z}
    \end{equation}
    
    Based on Lemma~\ref{lemma:converge of V*}, the convergence of the optimal V function does not depend on the initial value function $\upsilon^{0}(\mathbf{s})$ and the properties of $\upsilon^{0}(\mathbf{s})$ are propagated by the Bellman operation~\eqref{eq: bellman operator} to $\upsilon^{*}(\mathbf{s})$.
    So, to prove Theorem~\ref{theo: convexity, 2-1} from Lemma~\ref{lemma:converge of V*}, it is sufficient to prove that $\upsilon^{1}(\mathbf{s})$ holds the convexity i.e.,
    \begin{equation}
        \alpha \upsilon^{1}(\mathbf{s}''_{\text{AoI}}) + (1-\alpha) \upsilon^{1}(\mathbf{s}'_{\text{AoI}}) \geq \upsilon^{1}(\mathbf{s}),
        \label{eq: convexity, 1}
    \end{equation}
    under the assumption that $\upsilon^{0}(\mathbf{s})$ is convex, i.e.,
    \begin{equation}
        \alpha \upsilon^{0}(\mathbf{s}''_{\text{AoI}}) + (1-\alpha) \upsilon^{0}(\mathbf{s}'_{\text{AoI}}) \geq \upsilon^{0}(\mathbf{s}),
        \label{eq: convexity, 0}
    \end{equation}
    where $\mathbf{s}''_{\text{AoI}} = \left(\left(\delta''_{i},\delta_{j}\right), \mathbf{G}\right)$, $\mathbf{s}'_{\text{AoI}} = \left(\left(\delta'_{i},\delta_{j}\right), \mathbf{G}\right)$, $\mathbf{s} = \left(\left(\delta_{i},\delta_{j}\right), \mathbf{G}\right)$, and $\alpha\delta''_{i} + (1-\alpha)\delta'_{i} = \delta_{i}$ for $\alpha \in [0, 1]$ and $\delta'_{i} \geq \delta_{i} \geq \delta''_{i}$.

    Since the channel states are considered to be i.i.d. and the transition probability $\operatorname{P}(\mathbf{G}^{+})$ is dependent only on the probability distribution~\eqref{eq: channel distribution}, the channel states are constant during each iteration.
    Thus, we write the state as $\mathbf{s} = \bm{\delta}$, $\mathbf{s}'_{\text{AoI}} = \bm{\delta}'$, and $\mathbf{s}''_{\text{AoI}} = \bm{\delta}''$ for the writing simplicity.
    
    Thus, we prove~\eqref{eq: convexity, 1} by cases with different optimal actions of states for the first iteration, i.e., $\check{\mathbf{a}}^{1} = {\pi^{1}} (\mathbf{s}''_{\text{AoI}}), \hat{\mathbf{a}}^{1} = {\pi^{1}} (\mathbf{s}'_{\text{AoI}})$, in the following.
    \begin{itemize}
        \item[(a)] If $\check{a}^{1}_{i} = \hat{a}^{1}_{i} = 1$, then
        \begin{align}
            &\!\!\!\!\!\!\!\!\!\!\! \alpha \upsilon^{1}(\bm{\delta}'') + (1-\alpha) \upsilon^{1}(\bm{\delta}') - \upsilon^{1}(\bm{\delta}) \\
            &\!\!\!\!\!\!\!\!\!\!\! \geq \alpha Z(\bm{\delta}'', \check{\mathbf{a}}^{1}; \upsilon^{0}) + (1-\alpha) Z(\bm{\delta}', \hat{\mathbf{a}}^{1}; \upsilon^{0}) - Z(\bm{\delta}, \check{\mathbf{a}}^{1}; \upsilon^{0}) \\
            &\!\!\!\!\!\!\!\!\!\!\! = \left[\alpha c(\bm{\delta}'') + (1-\alpha) c(\bm{\delta}') - c(\bm{\delta}) \right] \\
            &\!\!\!\!\!\!\!\! + \! \gamma (1\!-\!\psi_{i,1}) \big[ \alpha \upsilon^{0}(1, \!\delta_{j} \!+\!\! 1) \!+\! (1\!-\!\alpha) \upsilon^{0}(1, \!\delta_{j} \!+\!\! 1)  \!\!\!\! \\
            & \qquad \qquad \!-\! \upsilon^{0}(1, \delta_{j} \!+\!\! 1)  \big] \\ 
            &\!\!\!\!\!\!\!\! + \! \gamma \psi_{i,1} \big[ \alpha \upsilon^{0}(\delta''_{i} \!+\! 1, \!\delta_{j} \!+\! 1) \!+\! (1\!-\!\!\alpha) \upsilon^{0}(\delta'_{i} \!+\! 1, \delta_{j} \!+\! 1) \!\! \\
            & \qquad \qquad - \upsilon^{0}(\delta_{i} + 1, \delta_{j} + 1) \big] \\
            &\!\!\!\!\!\!\!\! \geq 0,
        \end{align}
        where the first inequality is derived from and~\eqref{ineq:V&Z}, and the first equality is from~\eqref{eq: Z-function}, and the last inequality is from~\eqref{ineq: def cost convexity} and~\eqref{eq: convexity, 0}.
        
        \item[(b)] If $\check{a}^{1}_{j} = \hat{a}^{1}_{j} = 1$, then
        \begin{align}
            &\!\!\!\!\!\!\!\!\!\! \alpha \upsilon^{1}(\bm{\delta}'') + (1-\alpha) \upsilon^{1}(\bm{\delta}') - \upsilon^{1}(\bm{\delta}) \\
            &\!\!\!\!\!\!\!\!\!\!\!\! \geq \alpha Z(\bm{\delta}'', \check{\mathbf{a}}^{1}; \upsilon^{0}) + (1-\alpha) Z(\bm{\delta}', \hat{\mathbf{a}}^{1}; \upsilon^{0}) - Z(\bm{\delta}, \check{\mathbf{a}}^{1}; \upsilon^{0}) \\
            &\!\!\!\!\!\!\!\!\!\! = \left[\alpha c(\bm{\delta}'') + (1-\alpha) c(\bm{\delta}') - c(\bm{\delta}) \right] \\
            &\!\!\!\!\!\!\!\!\!\! +\! \gamma (1\!-\!\psi_{j,1}) \big[ \alpha \upsilon^{0}(\delta''_{i} \!\!+\!\! 1, \!1) \!+\! (1\!-\!\alpha) \upsilon^{0}(\delta'_{i} \!+\!\! 1, \!1) 
            \\ 
            & \qquad - \upsilon^{0}(\delta_{i} \!+\!\! 1, \!1)  \big] \!\!\!\! \\
            &\!\!\!\!\!\!\!\!\!\! + \! \gamma \psi_{j,1} \big[ \alpha \upsilon^{0}(\delta''_{i} \!+\! 1, \!\delta_{j} \!+\! 1) \!+\! (1\!-\!\!\alpha) \upsilon^{0}(\delta'_{i} \!+\! 1, \delta_{j} \!+\! 1) \!\! \\
            & \qquad  - \upsilon^{0}(\delta_{i} + 1, \delta_{j} + 1) \big] \\
            &\!\!\!\!\!\!\!\!\!\! \geq 0,
        \end{align}
        where the last inequality is derived based on~\eqref{ineq: def cost convexity},~\eqref{eq: convexity, 0}, and $\alpha(\delta''_{i} \!+\! 1) + (1-\alpha)(\delta'_{i} \!+\! 1) = \delta_{i} \!+\! 1$.
        
        \item[(c)] If $\check{a}^{1}_{j} = 1,  \hat{a}^{1}_{i} = 1$, then
        \begin{align}
            &\!\!\!\!\!\!\!\!\!\!\!\!\! \alpha \upsilon^{1}(\bm{\delta}'') + (1-\alpha) \upsilon^{1}(\bm{\delta}') - \upsilon^{1}(\bm{\delta}) \\
            &\!\!\!\!\!\!\!\!\!\!\!\!\! \geq \alpha Z(\bm{\delta}'', \check{\mathbf{a}}^{1}; \upsilon^{0}) + (1-\alpha) Z(\bm{\delta}', \hat{\mathbf{a}}^{1}; \upsilon^{0}) \\
            &\!\!\!\!\!\!\!\! - \alpha Z(\bm{\delta}, \check{\mathbf{a}}^{1}; \upsilon^{0}) -  (1-\alpha) Z(\bm{\delta}, \hat{\mathbf{a}}^{1}; \upsilon^{0})\\
            &\!\!\!\!\!\!\!\!\!\!\!\!\! = \left[\alpha c(\bm{\delta}'') + (1-\alpha) c(\bm{\delta}') - c(\bm{\delta}) \right] \\
            &\!\!\!\!\!\!\!\!\!\! +\! \alpha \gamma \big[(1\!-\!\psi_{j,1}) \upsilon^{0}(\delta''_{i}\!+\!1, 1) \!+\! (\psi_{i,1}\!-\!\psi_{j,1}) \upsilon^{0}(\delta_{i} \!+\! 1, \delta_{j} \!+\! 1) \!\! \\
            &\!\!\!\!\!\!\!\!\!\! - (1\!-\!\psi_{j,1}) \upsilon^{0}( \delta_{i} \!+\! 1, 1) \!-\! (\psi_{i,1}\!-\!\psi_{j,1})\upsilon^{0}(\delta''_{i} + 1, \delta_{j} + 1)\big] \\
            &\!\!\!\!\!\!\!\!\!\! +\! \gamma \psi_{i,1} \big[ \alpha \upsilon^{0}(\delta''_{i} \!+\! 1, \!\delta_{j} \!+\! 1) \!+\! (1\!-\!\!\alpha) \upsilon^{0}(\delta'_{i} \!+\! 1, \delta_{j} \!+\! 1) \!\! \\
            & \qquad \qquad - \upsilon^{0}(\delta_{i} + 1, \delta_{j} + 1) \big] \\
            & \!\!\!\!\!\!\!\!\!\!\!\!\! \geq 0
        \end{align}
        where the last inequality is based on~\eqref{ineq: def cost convexity},~\eqref{eq: convexity, 0}, and Lemma 3 in~\cite{chen2022seDRL}.

        \item[(d)] Based on Theorem 2 in~\cite{chen2022seDRL}, the case that $\check{a}^{1}_{i} = 1,  \hat{a}^{1}_{j} = 1$ cannot exist at all iteration during the value iteration.
    \end{itemize}
    Therefore, the Bellman operation~\eqref{eq: bellman operator} preserve the convexity of the value function $\upsilon^{0}(\mathbf{s})$ to the optimal V function $\upsilon^{*}(\mathbf{s})$.

\section{Proof of Theorem~\ref{theo: convexity, N-M}} \label{proof: convexity, N-M}
To prove Theorem~\ref{theo: convexity, N-M}, we require the lemma showing the optimal V function has the asymptotic monotonicity as below.
    \begin{lemma}[Asymptotic monotonicity of the optimal V function w.r.t. AoI state~\cite{chen2022seDRL}]\label{lemma:much larger}
    For states $\mathbf{s} = (\bm{\delta}, \mathbf{G})$ and $\mathbf{s}'_{\text{AoI}} = (\bm{\delta}'_{(i)}, \mathbf{G}) $, where $\delta'_{i} \gg \delta_{i}$, the optimal V function holds the inequality:
    \begin{equation}
        \upsilon^{*}(\mathbf{s}'_{\text{AoI}}) \gg \upsilon^{*}(\mathbf{s}).
    \end{equation}
    \end{lemma}
    
    Then, similar to the proof of Theorem~\ref{theo: convexity, 2-1}, to prove Theorem~\ref{theo: convexity, N-M} based on Lemma~\ref{lemma:converge of V*}, it is sufficient to prove that $\upsilon^{1}(\mathbf{s})$ is asymptotically convex, i.e., 
    \begin{equation}
        \alpha \upsilon^{1}(\mathbf{s}''_{\text{AoI}}) + (1-\alpha) \upsilon^{1}(\mathbf{s}'_{\text{AoI}}) \geq \upsilon^{1}(\mathbf{s}),
        \label{eq: asymptotic convexity, 1}
    \end{equation}
    under the assumption that the initial value function $\upsilon^{0}(\mathbf{s})$ has the asymptotic convexity, i.e.,
    \begin{equation}
        \alpha \upsilon^{0}(\mathbf{s}''_{\text{AoI}}) + (1-\alpha) \upsilon^{0}(\mathbf{s}'_{\text{AoI}}) \geq \upsilon^{0}(\mathbf{s}),
        \label{eq: asymptotic convexity, 0}
    \end{equation}
    where $\mathbf{s} = (\bm{\delta}, \mathbf{G})$, $\mathbf{s}'_{\text{AoI}} = (\bm{\delta}'_{(i)}, \mathbf{G})$, $\mathbf{s}''_{\text{AoI}} = (\bm{\delta}''_{(i)}, \mathbf{G})$, and $\alpha\delta''_{i} + (1-\alpha)\delta'_{i} = \delta_{i}$ for $\alpha \in [0, 1]$ and $\delta'_{i} \geq \delta_{i} \gg \delta''_{i}$.

    Therefore, in the following, we also write the states as $\mathbf{s} = \bm{\delta}$, $\mathbf{s}'_{\text{AoI}} = \bm{\delta}'$, and $\mathbf{s}''_{\text{AoI}} = \bm{\delta}''$, and prove~\eqref{eq: asymptotic convexity, 1} by cases (a) and (b) with different optimal actions $\check{\mathbf{a}}^{1} = {\pi^{1}} (\bm{\delta}''), \hat{\mathbf{a}}^{1} = {\pi^{1}} (\bm{\delta}')$.
    For writing simplicity, we write the AoI state as $\bm{\delta} = (\delta_{i}, \bm{\delta}_{\backslash\!\{\!i\!\}})$ in the following.
    \begin{itemize}
        \item[(a)] If $\check{\mathbf{a}}^{1} = \hat{\mathbf{a}}^{1}$, then
        \begin{align}
            &\!\!\!\!\!\!\!\!\!\!\! \alpha \upsilon^{1}(\bm{\delta}'') + (1-\alpha) \upsilon^{1}(\bm{\delta}') - \upsilon^{1}(\bm{\delta}) 
            \\
            &\!\!\!\!\!\!\!\!\!\!\! \geq \alpha Z(\bm{\delta}'', \check{\mathbf{a}}^{1}; \upsilon^{0}) + (1-\alpha) Z(\bm{\delta}', \hat{\mathbf{a}}^{1}; \upsilon^{0}) - Z(\bm{\delta}, \check{\mathbf{a}}^{1}; \upsilon^{0}) 
            \\
            &\!\!\!\!\!\!\!\!\!\!\! = \left[\alpha c(\bm{\delta}'') + (1-\alpha) c(\bm{\delta}') - c(\bm{\delta}) \right] 
            \\
            &\!\!\!\!\!\!\!\! + \gamma \sum_{\bm{\delta}_{\backslash\!\{\!i\!\}}^{+}} \operatorname{P} (\bm{\delta}_{\backslash\!\{i\}}^{+}|\bm{\delta}_{\backslash\!\{i\}}, \!\mathbf{G}_{\backslash\!\{i\}}, \!\check{\mathbf{a}}^{1}_{\backslash\!\{i\}}) 
            \\
            &\!\!\!\!\!\! \times \Bigg[ \alpha \sum_{{\delta_{i}''}^{+}} \operatorname{P} ({\delta_{i}''}^{+}|\delta''_{i},\! \mathbf{g}_{i},\! \check{\mathbf{a}}^{1}_{i}) \upsilon^{0}({\bm{\delta}''}^{+}) 
            \\
            &\!\!\!\!\!\! + (1-\alpha) \sum_{{\delta_{i}'}^{+}} \operatorname{P} ({\delta_{i}'}^{+}|\delta'_{i},\! \mathbf{g}_{i},\! \hat{\mathbf{a}}^{1}_{i}) \upsilon^{0}({\bm{\delta}'}^{+}) 
            \\
            &\!\!\!\!\!\! - \sum_{\delta_{i}^{+}} \operatorname{P} (\delta_{i}^{+}|\delta_{i},\! \mathbf{g}_{i},\! \check{\mathbf{a}}^{1}_{i}) \upsilon^{0}({\bm{\delta}}^{+}) \Bigg] 
            \\
            & \!\!\!\!\!\!\!\!\!\!\! \geq 0
        \end{align}
        where the first inequality is from~\eqref{ineq:V&Z}, and the first equality is from~\eqref{eq: Z-function} and $\check{\mathbf{a}}^{1} = \hat{\mathbf{a}}^{1}$, and the final inequality is from~\eqref{ineq: def cost convexity} and the following equality
        \begin{align}
            & \operatorname{P} (\delta_{i}''+1|\delta''_{i}, \mathbf{g}_{i}, \check{\mathbf{a}}^{1}_{i})  \upsilon^{0}(\delta_{i}''+1, \bm{\delta}^{+}_{\backslash\!\{i\}}) \\
            & + \operatorname{P} (\delta_{i}'+1|\delta'_{i}, \mathbf{g}_{i}, \hat{\mathbf{a}}^{1}_{i}) \upsilon^{0}(\delta_{i}'+1, \bm{\delta}^{+}_{\backslash\!\{i\}}) \\
            & \geq \operatorname{P} (\delta_{i}+1|\delta_{i}, \mathbf{g}_{i}, \check{\mathbf{a}}^{1}_{i}) \upsilon^{0}(\delta_{i}+1, \bm{\delta}^{+}_{\backslash\!\{i\}}),
            \label{ineq: asymptotic convexity case a1}
        \end{align}
        and 
        \begin{align}
            & \operatorname{P} (\delta_{i}''=1|\delta''_{i}, \mathbf{g}_{i}, \check{\mathbf{a}}^{1}_{i})  \upsilon^{0}(1, \bm{\delta}^{+}_{\backslash\!\{i\}}) \\
            & + \operatorname{P} (\delta_{i}'=1|\delta'_{i}, \mathbf{g}_{i}, \hat{\mathbf{a}}^{1}_{i}) \upsilon^{0}(1, \bm{\delta}^{+}_{\backslash\!\{i\}}) \\
            & \geq \operatorname{P} (\delta_{i}=1|\delta_{i}, \mathbf{g}_{i}, \check{\mathbf{a}}^{1}_{i}) \upsilon^{0}(1, \bm{\delta}^{+}_{\backslash\!\{i\}}),
            \label{ineq: asymptotic convexity case a2}
        \end{align}
        achieved by~\eqref{eq: asymptotic convexity, 0}, $\check{\mathbf{a}}^{1} = \hat{\mathbf{a}}^{1}$, and $\alpha(\delta''_{i} \!+\! 1) + (1-\alpha)(\delta'_{i} \!+\! 1) = \delta_{i} \!+\! 1$.
    
        \item[(b)] If $\check{\mathbf{a}}^{1} \neq \hat{\mathbf{a}}^{1}$, then
        \begin{align}
            &\!\!\!\!\!\!\!\!\!\!\! \alpha \upsilon^{1}(\bm{\delta}'') + (1-\alpha) \upsilon^{1}(\bm{\delta}') - \upsilon^{1}(\bm{\delta}) \\
            &\!\!\!\!\!\!\!\!\!\!\! \geq \alpha Z(\bm{\delta}'', \check{\mathbf{a}}^{1}; \upsilon^{0}) + (1-\alpha) Z(\bm{\delta}', \hat{\mathbf{a}}^{1}; \upsilon^{0}) - Z(\bm{\delta}, \hat{\mathbf{a}}^{1}; \upsilon^{0}) \\
            &\!\!\!\!\!\!\!\!\!\!\! = \left[\alpha c(\bm{\delta}'') + (1-\alpha) c(\bm{\delta}') - c(\bm{\delta}) \right] \\
            &\!\!\!\!\!\!\!\! + \gamma \Bigg[ \alpha \sum_{{\bm{\delta}''}^{+}} \operatorname{P} ({\bm{\delta}''}^{+}|\bm{\delta}'', \mathbf{G}, \check{\mathbf{a}}^{1}) \upsilon^{0}({\bm{\delta}''}^{+}) \\
            &\!\!\!\!\!\!\!\! + (1-\alpha) \sum_{{\bm{\delta}'}^{+}} \operatorname{P} ({\bm{\delta}'}^{+}|\bm{\delta}', \mathbf{G}, \hat{\mathbf{a}}^{1}) \upsilon^{0}({\bm{\delta}'}^{+}) \\
            &\!\!\!\!\!\!\!\! - \sum_{{\bm{\delta}}^{+}}  \operatorname{P} ({\bm{\delta}}^{+}|\bm{\delta}, \mathbf{G}, \hat{\mathbf{a}}^{1}) \upsilon^{0}({\bm{\delta}}^{+}) \Bigg]\\
            &\!\!\!\!\!\!\!\!\!\!\! \geq \gamma \Bigg[\alpha \sum_{{\bm{\delta}''}^{+}}  \operatorname{P} ({\bm{\delta}''}^{+}|\bm{\delta}'', \mathbf{G}, \hat{\mathbf{a}}^{1}) \upsilon^{0}({\bm{\delta}''}^{+}) \\
            &\!\!\!\!\!\!\!\! + (1-\alpha) \sum_{{\bm{\delta}'}^{+}}\operatorname{P} ({\bm{\delta}'}^{+}|\bm{\delta}', \mathbf{G}, \hat{\mathbf{a}}^{1}) \upsilon^{0}({\bm{\delta}'}^{+}) \\
            &\!\!\!\!\!\!\!\! - \sum_{{\bm{\delta}}^{+}} \operatorname{P} ({\bm{\delta}}^{+}|\bm{\delta}, \mathbf{G}, \hat{\mathbf{a}}^{1}) \upsilon^{0}({\bm{\delta}}^{+}) \Bigg]\\
            & \!\!\!\!\!\!\!\!\!\!\! \geq 0
        \end{align}
        where the first equality is derived based on~\eqref{eq: Z-function}, and the second inequality is based on~\eqref{ineq: def cost convexity}, the following inequality
        \begin{align}
            & \!\!\!\!\!\!\!\! \sum_{{\bm{\delta}'}^{+}}\operatorname{P} ({\bm{\delta}'}^{+}|\bm{\delta}', \mathbf{G}, \hat{\mathbf{a}}^{1}) \upsilon^{0}({\bm{\delta}'}^{+}) \\
            & \!\!\!\!\!\!\!\!\gg \sum_{{\bm{\delta}'}^{+}}\operatorname{P} ({\bm{\delta}''}^{+}|\bm{\delta}'', \mathbf{G}, \check{\mathbf{a}}^{1}) \upsilon^{0}({\bm{\delta}''}^{+})
        \end{align}
        achieved by Lemma~\ref{lemma:much larger} and $\delta'_{i} \gg \delta''_{i}$, and the last inequality is based on~\eqref{eq: asymptotic convexity, 0} and replacing $\check{\mathbf{a}}^{1}$ by $\hat{\mathbf{a}}^{1}$ in~\eqref{ineq: asymptotic convexity case a1} and~\eqref{ineq: asymptotic convexity case a2}.
    \end{itemize}
    
    Therefore, the asymptotic convexity of the value function $\upsilon^{0}(\mathbf{s})$ is preserved by the Bellman operation~\eqref{eq: bellman operator} to the optimal V function $\upsilon^{*}(\mathbf{s})$.

\section{Proof of Proposition~\ref{prop: V convexity}} \label{proof: prop convexity}

Before proving Proposition~\ref{prop: V convexity}, we develop the following Lemma.
\begin{lemma}\label{lemma: for prop convexity}
    If the devices are co-located, then for states $\dot{\mathbf{s}}'' = (\delta''_{i}, \dot{\bm{\delta}}_{\backslash\!\{i\}}, \mathbf{G})$, $\ddot{\mathbf{s}}' = (\delta'_{i}, \ddot{\bm{\delta}}_{\backslash\!\{i\}}, \mathbf{G})$, $\dot{\mathbf{s}} = (\delta_{i}, \dot{\bm{\delta}}_{\backslash\!\{i\}}, \mathbf{G})$, and $\ddot{\mathbf{s}} = (\delta_{i}, \ddot{\bm{\delta}}_{\backslash\!\{i\}}, \mathbf{G})$, where $\alpha\delta''_{i} + (1-\alpha)\delta'_{i} = \delta_{i}$, $\alpha \in [0, 1]$, and $\delta'_{i} \geq \delta_{i} \geq \delta''_{i} \gg 1$, then the optimal V function holds the following inequality:
    \begin{align}
        \alpha \upsilon(\dot{\mathbf{s}}'') + (1-\alpha) \upsilon(\ddot{\mathbf{s}}') \geq \alpha \upsilon(\dot{\mathbf{s}}) + (1-\alpha) \upsilon(\ddot{\mathbf{s}}).
    \end{align}
\end{lemma}

\begin{proof}
    See Appendix~\ref{proof: for prop convexity}.
\end{proof}

Next, similar to the proof of Theorem~\ref{theo: convexity, 2-1}, to prove Proposition~\ref{prop: V convexity} based on Lemma~\ref{lemma:converge of V*}, it is sufficient to prove that $\upsilon^{1}(\mathbf{s})$ holds the following inequality
\begin{equation}
    \alpha \upsilon^{1}(\mathbf{s}''_{\text{AoI}}) + (1-\alpha) \upsilon^{1}(\mathbf{s}'_{\text{AoI}}) \geq \upsilon^{1}(\mathbf{s}),
    \label{eq: prop convexity, 1}
\end{equation}
under the assumption that the initial value function $\upsilon^{0}(\mathbf{s})$ holds the inequality:
\begin{equation}
    \alpha \upsilon^{0}(\mathbf{s}''_{\text{AoI}}) + (1-\alpha) \upsilon^{0}(\mathbf{s}'_{\text{AoI}}) \geq \upsilon^{0}(\mathbf{s}),
    \label{eq: prop convexity, 0}
\end{equation}
where $\mathbf{s}''_{\text{AoI}} = (\bm{\delta}''_{(i)}, \mathbf{G})$, $\mathbf{s}'_{\text{AoI}} = (\bm{\delta}'_{(i)}, \mathbf{G})$ and $\mathbf{s} = (\bm{\delta}, \mathbf{G})$, and $\alpha\delta''_{i} + (1-\alpha)\delta'_{i} = \delta_{i}$ for $\alpha \in [0, 1]$ and $\delta'_{i} \geq \delta_{i} \geq \delta''_{i} \gg 1$.

We also write the states as $\mathbf{s} = \bm{\delta}$, $\mathbf{s}'_{\text{AoI}} = \bm{\delta}'$, and $\mathbf{s}''_{\text{AoI}} = \bm{\delta}''$, and prove~\eqref{eq: asymptotic convexity, 1} by cases (a), (b), and (c) with different optimal actions $\check{\mathbf{a}}^{1} = {\pi^{1}} (\bm{\delta}''), \hat{\mathbf{a}}^{1} = {\pi^{1}} (\bm{\delta}')$.
For writing simplicity, we write the AoI state as $\bm{\delta} = (\delta_{i}, \bm{\delta}_{\backslash\!\{\!i\!\}})$ in the following.

\begin{itemize}
    \item[(a)] If $\check{a}_{i}^{1} = m_{1}, \hat{a}_{i}^{1}= m_{2}$, then there are 2 cases with different packet drop rate: (a.1) $\psi_{i,m_{1}} \leq \psi_{i,m_{2}}$ and (a.2) $\psi_{i,m_{1}} > \psi_{i,m_{2}}$.
    \item[(a.1)] If $\psi_{i,m_{1}} \leq \psi_{i,m_{2}}$ and $\hat{a}_{j}^{1}= m_{1}$, we define another action $\dot{\mathbf{a}}^{1}$, where $\dot{a}_{i}^{1} = m_{1}, \dot{a}_{j}^{1} = m_{2}$, and $\dot{\mathbf{a}}^{1}_{\backslash\!\{\!i,j\!\}} = \hat{\mathbf{a}}^{1}_{\backslash\!\{\!i,j\!\}}$, then we have $\psi_{j,m_{1}} < \psi_{j,m_{2}}$ and
    \begin{align}
        & \hspace{-1cm}\operatorname{P} (\delta_{i}'\!+\!1|\delta_{i}', \!\mathbf{G}, \hat{a}_{i}^{1}) \sum_{\delta_{j}^{+}}  \operatorname{P} (\delta_{j}^{+}|\delta_{j}, \!\mathbf{G}, \hat{a}_{j}^{1}) \upsilon^{0} (\delta_{i}'\!+\!1, \bm{\delta}_{\backslash\!\{i\}}^{+}) 
        \\ 
        & \hspace{-1cm}\geq\! \operatorname{P} (\delta_{i}'\!+\!1|\delta_{i}', \!\mathbf{G}, \dot{a}_{i}^{1}) \!\sum_{\delta_{j}^{+}} \operatorname{P} (\delta_{j}^{+}|\delta_{j}, \!\mathbf{G}, \dot{a}_{j}^{1}) \upsilon^{0} (\delta_{i}'\!+\!1, \bm{\delta}_{\backslash\!\{\!i\!\}}^{+}).
        \label{eq: prop convexity, case a.1}
    \end{align}
    Next, we derive that
    \begin{align}
        & \!\hspace{-0.6cm} \alpha \upsilon^{1}(\bm{\delta}'') + (1-\alpha) \upsilon^{1}(\bm{\delta}') - \upsilon^{1}(\bm{\delta}) 
        \\
        &\hspace{-0.6cm} \!\geq  \alpha Z(\bm{\delta}'', \check{\mathbf{a}}^{1}; \upsilon^{0}) + (1-\alpha) Z(\bm{\delta}', \hat{\mathbf{a}}^{1}; \upsilon^{0}) 
        \\
        &\hspace{-0.5cm} \!- \alpha Z(\bm{\delta}, \check{\mathbf{a}}^{1}; \upsilon^{0}) - (1-\alpha) Z(\bm{\delta}, \dot{\mathbf{a}}^{1}; \upsilon^{0}) 
        \\
        &\hspace{-0.6cm} \!= \left[\alpha c(\bm{\delta}'') + (1-\alpha) c(\bm{\delta}') - \alpha c(\bm{\delta}) - (1-\alpha) c(\bm{\delta}) \right] 
        \\
        &\hspace{-0.5cm} \!\!+\! \alpha \! \sum_{{\delta_{i}''}^{+}} \sum_{\bm{\delta}_{\backslash\!\{\!i\!\}}^{+}} \!\operatorname{P} ({\delta_{i}''}^{+}|\delta_{i}'',\! \mathbf{G}, \!\check{a}_{i}^{1}) \!\operatorname{P} (\bm{\delta}_{\backslash\!\{\!i\!\}}^{+}|\bm{\delta}_{\backslash\!\{\!i\!\}},\! \mathbf{G}, \!\check{\mathbf{a}}_{\backslash\!\{\!i\!\}}^{1}) \upsilon^{0} ({\bm{\delta}''}^{+})  \!\!\!\!\!
        \\
        &\hspace{-0.5cm} \!+ (1\!-\!\alpha) \sum_{{\delta_{i}'}^{+}} \sum_{\delta_{j}^{+}} \!\sum_{\bm{\delta}_{\backslash\!\{\!i,j\!\}}^{+}} \!
        \operatorname{P} ({\delta_{i}'}^{+}|\delta_{i}', \mathbf{G}, \hat{a}_{i}^{1}) 
        \operatorname{P} (\delta_{j}^{+}|\delta_{j}, \mathbf{G}, \hat{a}_{j}^{1}) 
        \\
        & \times \operatorname{P} (\bm{\delta}_{\backslash\!\{\!i,j\!\}}^{+}|\bm{\delta}_{\backslash\!\{\!i,j\!\}}, \mathbf{G}, \hat{\mathbf{a}}_{\backslash\!\{\!i,j\!\}}^{1}) \upsilon^{0} ({\bm{\delta}'}^{+}) 
        \\
        &\hspace{-0.5cm} \!-\! \alpha \! \sum_{\delta_{i}^{+}} \! \sum_{\bm{\delta}_{\backslash\!\{\!i\!\}}^{+}} \!\operatorname{P} (\delta_{i}^{+}|\delta_{i}, \!\mathbf{G}, \!\check{a}_{i}^{1}) \!\operatorname{P} (\bm{\delta}_{\backslash\!\{i\}}^{+}|\bm{\delta}_{\backslash\!\{\!i\!\}}, \!\mathbf{G}, \!\check{\mathbf{a}}_{\backslash\!\{\!i\!\}}^{1}) \upsilon^{0} (\bm{\delta}^{+}) 
        \\
        & \hspace{-0.5cm} \!- (1-\alpha) \!\sum_{\delta_{i}^{+}}\! \sum_{\delta_{j}^{+}} \sum_{\bm{\delta}_{\backslash\!\{\!i,j\!\}}^{+}}
        \operatorname{P} (\delta_{i}^{+}|\delta_{i},\! \mathbf{G}, \dot{a}_{i}^{1}) 
        \operatorname{P} (\delta_{j}^{+}|\delta_{j},\! \mathbf{G}, \dot{a}_{j}^{1}) 
        \\
        & \times \operatorname{P} (\bm{\delta}_{\backslash\!\{\!i,j\!\}}^{+}|\bm{\delta}_{\backslash\!\{\!i,j\!\}},\! \mathbf{G}, \dot{\mathbf{a}}_{\backslash\!\{\!i,j\!\}}^{1}) \upsilon^{0} (\bm{\delta}^{+}) \!\!\!
        \\
        & \hspace{-0.6cm} \!\geq\! \alpha\! \operatorname{P} ({\delta_{i}''}\!+\!1|\delta_{i}'', \!\mathbf{G}, \check{a}_{i}^{1}) \! \sum_{\bm{\delta}_{\backslash\!\{\!i\!\}}^{+}} \! \operatorname{P} (\bm{\delta}_{\backslash\!\{\!i\!\}}^{+}|\bm{\delta}_{\backslash\!\{\!i\!\}}, \!\mathbf{G}, \check{\mathbf{a}}_{\backslash\!\{\!i\!\}}^{1}) \upsilon^{0} ({\bm{\delta}''}^{+})  \!\!\!\!
        \\
        & \hspace{-0.5cm} \!+ (1\!-\!\alpha) \operatorname{P} ({\delta_{i}'}+1|\delta_{i}', \mathbf{G}, \dot{a}_{i}^{1}) \sum_{\delta_{j}^{+}} \sum_{\bm{\delta}_{\backslash\!\{\!i,j\!\}}^{+}}
        \operatorname{P} (\delta_{j}^{+}|\delta_{j}, \mathbf{G}, \dot{a}_{j}^{1}) 
        \\
        & \times \operatorname{P} (\bm{\delta}_{\backslash\!\{\!i,j\!\}}^{+}|\bm{\delta}_{\backslash\!\{\!i,j\!\}}, \mathbf{G}, \dot{\mathbf{a}}_{\backslash\!\{\!i,j\!\}}^{1}) \upsilon^{0} ({\bm{\delta}'}^{+}) 
        \\
        & \hspace{-0.5cm} \!- \alpha \!\operatorname{P} (\delta_{i}\!+\!1|\delta_{i}, \!\mathbf{G}, \check{a}_{i}^{1}) \! \sum_{\bm{\delta}_{\backslash\!\{\!i\!\}}^{+}}  \operatorname{P} (\bm{\delta}_{\backslash\!\{\!i\!\}}^{+}|\bm{\delta}_{\backslash\!\{\!i\!\}}, \!\mathbf{G}, \check{\mathbf{a}}_{\backslash\!\{\!i\!\}}^{1}) \upsilon^{0} (\bm{\delta}^{+}) \!\!
        \\ 
        & \hspace{-0.5cm} \!- (1-\alpha) \operatorname{P} (\delta_{i}+1|\delta_{i},\! \mathbf{G}, \dot{a}_{i}^{1})  \! \sum_{\delta_{j}^{+}} \sum_{\bm{\delta}_{\backslash\!\{\!i,j\!\}}^{+}}
        \operatorname{P} (\delta_{j}^{+}|\delta_{j},\! \mathbf{G}, \dot{a}_{j}^{1}) 
        \\
        & \times \operatorname{P} (\bm{\delta}_{\backslash\!\{\!i,j\!\}}^{+}|\bm{\delta}_{\backslash\!\{\!i,j\!\}},\! \mathbf{G}, \dot{\mathbf{a}}_{\backslash\!\{\!i,j\!\}}^{1}) \upsilon^{0} (\bm{\delta}^{+}) \!\!\!
        \\
        & \hspace{-0.6cm} \geq 0,
    \end{align}
    where the first inequality is from~\eqref{ineq:V&Z}, and the first equality is derived based on~\eqref{eq: Z-function}, and the second inequality is from~\eqref{ineq: def cost convexity},~\eqref{eq: prop convexity, case a.1}, and the following equality
    \begin{align}
        & \operatorname{P} ({\delta_{i}''}^{+}=1|\delta_{i}'', \mathbf{G}, \check{a}_{i}^{1}) \upsilon^{0} (1, \bm{\delta}_{\backslash\!\{i\}}^{+})
        \\
        & = \operatorname{P} ({\delta_{i}}^{+}=1|\delta_{i}, \mathbf{G}, \check{a}_{i}^{1}) \upsilon^{0} (1, \bm{\delta}_{\backslash\!\{i\}}^{+}),
    \end{align}
    and
    \begin{align}
        & \operatorname{P} (\delta_{i}'+1|\delta_{i}', \mathbf{G}, \hat{a}_{i}^{1}) \upsilon^{0} (\delta_{i}'+1, \bm{\delta}_{\backslash\!\{i\}}^{+}) 
        \\
        & \gg \operatorname{P} ({\delta_{i}'}^{+}=1|\delta_{i}', \mathbf{G}, \hat{a}_{i}^{1}) \upsilon^{0} (1, \bm{\delta}_{\backslash\!\{i\}}^{+}),
    \end{align}
    and
    \begin{align}
        & \operatorname{P} (\delta_{i}+1|\delta_{i}, \mathbf{G}, \dot{a}_{i}^{1}) \upsilon^{0} (\delta_{i}+1, \bm{\delta}_{\backslash\!\{i\}}^{+}) 
        \\
        & \gg \operatorname{P} ({\delta_{i}}^{+}=1|\delta_{i}, \mathbf{G}, \dot{a}_{i}^{1}) \upsilon^{0} (1, \bm{\delta}_{\backslash\!\{i\}}^{+}),
        \label{ineq: prop convexity case a.1 much larger}
    \end{align}
    achieved by Lemma~\ref{lemma:much larger}, $\delta_{i}'+1 \gg 1$, and $\delta_{i}+1 \gg 1$, and the last inequality is derived based on Lemma~\ref{lemma: for prop convexity} and the following equality
    \begin{align}
        \operatorname{P} ({\delta_{i}''}+1|\delta_{i}'', \mathbf{G}, \check{a}_{i}^{1}) = \operatorname{P} ({\delta_{i}'}+1|\delta_{i}', \mathbf{G}, \dot{a}_{i}^{1})
        \\
        = \operatorname{P} ({\delta_{i}}+1|\delta_{i}, \mathbf{G}, \check{a}_{i}^{1}) = \operatorname{P} ({\delta_{i}}+1|\delta_{i}, \mathbf{G}, \dot{a}_{i}^{1}),
    \end{align}
    achieved by $\check{a}_{i}^{1} = \dot{a}_{i}^{1}$.

\item[(a.2)] If $\psi_{i,m_{1}} > \psi_{i,m_{2}}$ and $\check{a}_{j}^{1}= m_{1}$, we define another action $\dot{\mathbf{a}}^{1}$, where $\dot{a}_{i}^{1} = m_{2}, \dot{a}_{j}^{1} = m_{1}$, and $\dot{\mathbf{a}}^{1}_{\backslash\!\{\!i,j\!\}} = \check{\mathbf{a}}^{1}_{\backslash\!\{\!i,j\!\}}$, then we have $\psi_{j,m_{1}} > \psi_{j,m_{2}}$ and
    \begin{align}
        & \hspace{-1cm}\operatorname{P} (\delta_{i}''\!\!+\!1|\delta_{i}'', \!\mathbf{G}, \check{a}_{i}^{1}) \sum_{\delta_{j}^{+}}  \operatorname{P} (\delta_{j}^{+}|\delta_{j}, \!\mathbf{G}, \check{a}_{j}^{1}) \upsilon^{0} (\delta_{i}''\!+\!1, \!\bm{\delta}_{\backslash\!\{\!i\!\}}^{+}) 
        \\ 
        & \hspace{-1cm}\geq\! \operatorname{P} (\delta_{i}''\!\!+\!1|\delta_{i}'', \!\mathbf{G}, \!\dot{a}_{i}^{1}) \!\sum_{\delta_{j}^{+}} \!\operatorname{P} (\delta_{j}^{+}|\delta_{j}, \!\mathbf{G}, \!\dot{a}_{j}^{1}) \upsilon^{0} (\delta_{i}''\!\!+\!1, \!\bm{\delta}_{\backslash\!\{\!i\!\}}^{+}).
        \label{eq: prop convexity, case a.2}
    \end{align}
    Next, we derive that
    \begin{align}
        & \!\hspace{-0.6cm} \alpha \upsilon^{1}(\bm{\delta}'') + (1-\alpha) \upsilon^{1}(\bm{\delta}') - \upsilon^{1}(\bm{\delta}) 
        \\
        &\hspace{-0.6cm} \!\geq  \alpha Z(\bm{\delta}'', \check{\mathbf{a}}^{1}; \upsilon^{0}) + (1-\alpha) Z(\bm{\delta}', \hat{\mathbf{a}}^{1}; \upsilon^{0}) 
        \\
        &\hspace{-0.4cm} \!- \alpha Z(\bm{\delta}, \dot{\mathbf{a}}^{1}; \upsilon^{0}) - (1-\alpha) Z(\bm{\delta}, \hat{\mathbf{a}}^{1}; \upsilon^{0}) 
        \\
        & \hspace{-0.6cm} \!\geq\! \alpha \operatorname{P} ({\delta_{i}''}\!+\!1|\delta_{i}'', \!\mathbf{G}, \!\dot{a}_{i}^{1}) \! \sum_{\bm{\delta}_{\backslash\!\{\!i\!\}}^{+}}  \!\operatorname{P} (\bm{\delta}_{\backslash\!\{\!i\!\}}^{+}|\bm{\delta}_{\backslash\!\{\!i\!\}}, \!\mathbf{G},\! \dot{\mathbf{a}}_{\backslash\!\{\!i\!\}}^{1}) \upsilon^{0} ({\bm{\delta}''}^{+})  \!\!
        \\ 
        & \hspace{-0.4cm} \!+ (1\!-\!\alpha) \operatorname{P} ({\delta_{i}'}+1|\delta_{i}', \mathbf{G}, \hat{a}_{i}^{1}) \sum_{\delta_{j}^{+}} \sum_{\bm{\delta}_{\backslash\!\{\!i,j\!\}}^{+}}
        \operatorname{P} (\delta_{j}^{+}|\delta_{j}, \mathbf{G}, \hat{a}_{j}^{1}) 
        \\
        & \times \operatorname{P} (\bm{\delta}_{\backslash\!\{\!i,j\!\}}^{+}|\bm{\delta}_{\backslash\!\{\!i,j\!\}}, \mathbf{G}, \hat{\mathbf{a}}_{\backslash\!\{\!i,j\!\}}^{1}) \upsilon^{0} ({\bm{\delta}'}^{+}) 
        \\
        & \hspace{-0.4cm} \!- \alpha \operatorname{P} (\delta_{i}\!+\!1|\delta_{i},\! \mathbf{G},\! \dot{a}_{i}^{1}) \! \sum_{\bm{\delta}_{\backslash\!\{\!i\!\}}^{+}} \! \operatorname{P} (\bm{\delta}_{\backslash\!\{\!i\!\}}^{+}|\bm{\delta}_{\backslash\!\{\!i\!\}}, \!\mathbf{G}, \!\dot{\mathbf{a}}_{\backslash\!\{\!i\!\}}^{1}) \upsilon^{0} (\bm{\delta}^{+}) 
        \\ 
        & \hspace{-0.4cm} \!- (1-\alpha) \operatorname{P} (\delta_{i}+1|\delta_{i},\! \mathbf{G}, \hat{a}_{i}^{1})  \! \sum_{\delta_{j}^{+}} \sum_{\bm{\delta}_{\backslash\!\{\!i,j\!\}}^{+}}
        \operatorname{P} (\delta_{j}^{+}|\delta_{j},\! \mathbf{G}, \hat{a}_{j}^{1}) 
        \\
        & \times \operatorname{P} (\bm{\delta}_{\backslash\!\{\!i,j\!\}}^{+}|\bm{\delta}_{\backslash\!\{\!i,j\!\}},\! \mathbf{G}, \hat{\mathbf{a}}_{\backslash\!\{\!i,j\!\}}^{1}) \upsilon^{0} (\bm{\delta}^{+}) \!\!\!
        \\
        & \hspace{-0.6cm} \geq 0,
    \end{align}
    where the second inequality is from ~\eqref{ineq: def cost convexity},~\eqref{eq: prop convexity, case a.2}, and the following equality
    \begin{align}
        & \operatorname{P} ({\delta_{i}'}^{+}=1|\delta_{i}', \mathbf{G}, \hat{a}_{i}^{1}) \upsilon^{0} (1, \bm{\delta}_{\backslash\!\{i\}}^{+})
        \\
        & = \operatorname{P} ({\delta_{i}}^{+}=1|\delta_{i}, \mathbf{G}, \hat{a}_{i}^{1}) \upsilon^{0} (1, \bm{\delta}_{\backslash\!\{i\}}^{+}),
    \end{align}
    and
    \begin{align}
        & \operatorname{P} (\delta_{i}''+1|\delta_{i}'', \mathbf{G}, \check{a}_{i}^{1}) \upsilon^{0} (\delta_{i}''+1, \bm{\delta}_{\backslash\!\{i\}}^{+}) 
        \\
        & \gg \operatorname{P} ({\delta_{i}''}^{+}=1|\delta_{i}'', \mathbf{G}, \check{a}_{i}^{1}) \upsilon^{0} (1, \bm{\delta}_{\backslash\!\{i\}}^{+}),
    \end{align}
    and~\eqref{ineq: prop convexity case a.1 much larger} achieved by Lemma~\ref{lemma:much larger}, $\delta_{i}''+1 \gg 1$, and $\delta_{i}+1 \gg 1$, and the last inequality is derived based on Lemma~\ref{lemma: for prop convexity} and the following equality
    \begin{align}
        \operatorname{P} ({\delta_{i}''}+1|\delta_{i}'', \mathbf{G}, \dot{a}_{i}^{1}) = \operatorname{P} ({\delta_{i}'}+1|\delta_{i}', \mathbf{G}, \hat{a}_{i}^{1})
        \\
        = \operatorname{P} ({\delta_{i}}+1|\delta_{i}, \mathbf{G}, \dot{a}_{i}^{1}) = \operatorname{P} ({\delta_{i}}+1|\delta_{i}, \mathbf{G}, \hat{a}_{i}^{1}),
    \end{align}
    achieved by $\check{a}_{i}^{1} = \dot{a}_{i}^{1}$.
    \item[(b)] If $\check{a}_{i}^{1} = m, \hat{a}_{i}^{1}= 0$, then we assume that $\hat{a}_{j}^{1}= m$ and we have
    \begin{align}
        & \!\hspace{-0.7cm} \alpha \upsilon^{1}(\bm{\delta}'') + (1-\alpha) \upsilon^{1}(\bm{\delta}') - \upsilon^{1}(\bm{\delta}) 
        \\
        &\hspace{-0.6cm} \!\geq  \alpha Z(\bm{\delta}'', \check{\mathbf{a}}^{1}; \upsilon^{0}) + (1-\alpha) Z(\bm{\delta}', \hat{\mathbf{a}}^{1}; \upsilon^{0}) 
        \\
        &\hspace{-0.7cm} \!- \alpha Z(\bm{\delta}, \check{\mathbf{a}}^{1}; \upsilon^{0}) - (1-\alpha) Z(\bm{\delta}, \hat{\mathbf{a}}^{1}; \upsilon^{0}) 
        \\
        &\hspace{-0.7cm} \!= \left[\alpha c(\bm{\delta}'') + (1-\alpha) c(\bm{\delta}') - \alpha c(\bm{\delta}) - (1-\alpha) c(\bm{\delta}) \right] 
        \\
        &\hspace{-0.6cm} \!\!+ \! \alpha \! \sum_{{\delta_{i}''}^{+}} \! \sum_{\bm{\delta}_{\backslash\!\{\!i\!\}}^{+}} \!\operatorname{P} ({\delta_{i}''}^{+}\!|\delta_{i}'', \!\mathbf{G}, \check{a}_{i}^{1}) \!\operatorname{P} (\bm{\delta}_{\backslash\!\{\!i\!\}}^{+}\!|\bm{\delta}_{\backslash\!\{\!i\!\}}, \!\mathbf{G}, \check{\mathbf{a}}_{\backslash\!\{\!i\!\}}^{1}) \upsilon^{0} ({\bm{\delta}''}^{+})  \!\!\!
        \\
        &\hspace{-0.6cm} \!\!+\! (1\!-\!\alpha) \!\sum_{{\delta_{i}'}^{+}} \!\! \sum_{\bm{\delta}_{\backslash\!\{\!i\!\}}^{+}} \!\!\operatorname{P} ({\delta_{i}'}^{+}|\delta_{i}', \!\mathbf{G}, \!\hat{a}_{i}^{1}) \!\operatorname{P} (\bm{\delta}_{\backslash\!\{\!i\!\}}^{+}|\bm{\delta}_{\backslash\!\{\!i\!\}}, \!\mathbf{G}, \!\hat{\mathbf{a}}_{\backslash\!\{\!i\!\}}^{1}) \upsilon^{0} ({\bm{\delta}'}^{+}) \!\!\!
        \\
        &\hspace{-0.6cm} \!- \alpha \! \sum_{\delta_{i}^{+}} \! \sum_{\bm{\delta}_{\backslash\!\{\!i\!\}}^{+}}\! \operatorname{P} (\delta_{i}^{+}|\delta_{i}, \!\mathbf{G}, \!\check{a}_{i}^{1}) \!\operatorname{P} (\bm{\delta}_{\backslash\!\{\!i\!\}}^{+}|\bm{\delta}_{\backslash\!\{\!i\!\}},\! \mathbf{G}, \!\check{\mathbf{a}}_{\backslash\!\{\!i\!\}}^{1}) \upsilon^{0} (\bm{\delta}^{+}) 
        \\
        & \hspace{-0.6cm} \!-\! (1\!-\!\alpha) \!\sum_{\delta_{i}^{+}}\!\! \sum_{\bm{\delta}_{\backslash\!\{\!i\!\}}^{+}}
        \!\operatorname{P} (\delta_{i}^{+}|\delta_{i},\! \mathbf{G}, \!\hat{a}_{i}^{1}) 
        \!\operatorname{P} (\bm{\delta}_{\backslash\!\{\!i\!\}}^{+}|\bm{\delta}_{\backslash\!\{\!i\!\}},\! \mathbf{G}, \!\hat{\mathbf{a}}_{\backslash\!\{\!i\!\}}^{1}) \upsilon^{0} (\bm{\delta}^{+}) \!\!\!
        \\
        & \hspace{-0.7cm} \geq\! \alpha\! \operatorname{P} ({\delta_{i}''}\!\!+\!1|\delta_{i}'', \!\mathbf{G}, \!\check{a}_{i}^{1}) \!\! \sum_{\bm{\delta}_{\backslash\!\{\!i\!\}}^{+}} \!\operatorname{P} (\bm{\delta}_{\backslash\!\{\!i\!\}}^{+}|\bm{\delta}_{\backslash\!\{\!i\!\}}, \!\mathbf{G}, \!\check{\mathbf{a}}_{\backslash\!\{\!i\!\}}^{1}) \upsilon^{0} ({\bm{\delta}''}^{+})
        \\
        & \hspace{-0.7cm} + (1\!-\!\alpha) \operatorname{P} ({\delta_{i}'}+1|\delta_{i}', \mathbf{G}, \hat{a}_{i}^{1}) \operatorname{P} ({\delta_{j}}+1|\delta_{j}, \mathbf{G}, \hat{a}_{j}^{1}) 
        \\
        & \times \sum_{\bm{\delta}_{\backslash\!\{\!i,j\!\}}^{+}} \operatorname{P} (\bm{\delta}_{\backslash\!\{\!i,j\!\}}^{+}|\bm{\delta}_{\backslash\!\{\!i\!\}}, \mathbf{G}, \hat{\mathbf{a}}_{\backslash\!\{\!i,j\!\}}^{1}) \upsilon^{0} ({\bm{\delta}'}^{+})
        \\
        & \hspace{-0.6cm} - \!\alpha\! \operatorname{P} ({\delta_{i}}\!+\!1|\delta_{i}, \!\mathbf{G}, \!\check{a}_{i}^{1}) \!\! \sum_{\bm{\delta}_{\backslash\!\{\!i\!\}}^{+}} \!\operatorname{P} (\bm{\delta}_{\backslash\!\{\!i\!\}}^{+}|\bm{\delta}_{\backslash\!\{\!i\!\}}, \!\mathbf{G}, \!\check{\mathbf{a}}_{\backslash\!\{\!i\!\}}^{1}) \upsilon^{0} ({\bm{\delta}}^{+})
        \\
        & \hspace{-0.6cm} + (1\!-\!\alpha) \operatorname{P} ({\delta_{i}}+1|\delta_{i}, \mathbf{G}, \hat{a}_{i}^{1}) \operatorname{P} ({\delta_{j}}+1|\delta_{j}, \mathbf{G}, \hat{a}_{j}^{1}) 
        \\
        & \times \sum_{\bm{\delta}_{\backslash\!\{\!i,j\!\}}^{+}} \operatorname{P} (\bm{\delta}_{\backslash\!\{\!i,j\!\}}^{+}|\bm{\delta}_{\backslash\!\{\!i\!\}}, \mathbf{G}, \hat{\mathbf{a}}_{\backslash\!\{\!i,j\!\}}^{1}) \upsilon^{0} ({\bm{\delta}}^{+})
        \\
        & \hspace{-0.7cm} \geq 0,
    \end{align}
    where the second inequality is derived from~\eqref{ineq: def cost convexity} and the following equality
    \begin{align}
        & \!\!\!\!\!\!\!\operatorname{P} ({\delta_{i}''}^{+}\!\!=\!1|\delta_{i}'', \!\mathbf{G}, \!\check{a}_{i}^{1}) \!\!\sum_{\bm{\delta}_{\backslash\!\{\!i\!\}}^{+}} \!\!\operatorname{P} (\bm{\delta}_{\backslash\!\{\!i\!\}}^{+}|\bm{\delta}_{\backslash\!\{\!i\!\}}, \!\mathbf{G}, \!\check{\mathbf{a}}_{\backslash\!\{\!i\!\}}^{1}) \upsilon^{0} (1, \!\bm{\delta}_{\backslash\!\{\!i\!\}}^{+}) \!\!\!\!
        \\
        & \!\!\!\!\!\!\!\!\! =\! \operatorname{P} ({\delta_{i}}^{+}\!\!\!=\!1|\delta_{i}, \!\mathbf{G}, \check{a}_{i}^{1}) \!\!\!\sum_{\bm{\delta}_{\backslash\!\{\!i\!\}}^{+}} \!\!\operatorname{P} (\bm{\delta}_{\backslash\!\{\!i\!\}}^{+}\!|\bm{\delta}_{\backslash\!\{\!i\!\}}, \!\mathbf{G}, \check{\mathbf{a}}_{\backslash\!\{\!i\!\}}^{1}) \upsilon^{0} (1, \!\bm{\delta}_{\backslash\!\{\!i\!\}}^{+}), \!\!\!\!
    \end{align}
    and
    \begin{equation}
        \operatorname{P} ({\delta_{i}'}^{+}=1|\delta_{i}', \mathbf{G}, \hat{a}_{i}^{1}) = \operatorname{P} ({\delta_{i}}^{+}=1|\delta_{i}, \mathbf{G}, \hat{a}_{i}^{1}) = 0
    \end{equation}
    from $\hat{a}_{i}^{1}= 0$, and the inequality
    \begin{align}
        & \!\!\!\operatorname{P} ({\delta_{i}'}\!+\!1|\delta_{i}', \!\mathbf{G}, \hat{a}_{i}^{1}) \operatorname{P} ({\delta_{j}}^{+}\!=\!1|\delta_{j}, \!\mathbf{G}, \hat{a}_{j}^{1}) \upsilon^{0} ({\delta_{i}'}\!+\!1, 1, {\bm{\delta}'}_{\backslash\!\{\!i,j\!\}}^{+}) \!\!
        \\
        & \!\!\!\!\geq\! \operatorname{P} ({\delta_{i}}\!+\!1|\delta_{i}, \!\mathbf{G}, \hat{a}_{i}^{1}) \operatorname{P} ({\delta_{j}}^{+}\!=\!1|\delta_{j}, \!\mathbf{G}, \hat{a}_{j}^{1}) \upsilon^{0} ({\delta_{i}'}\!+\!1, 1, {\bm{\delta}}_{\backslash\!\{\!i,j\!\}}^{+})\!\!\!
    \end{align}
    achieved from Lemma~\ref{lemma:monotone V} with $\delta_{i}' \geq \delta_{i}$, and the last inequality is based on Lemma~\ref{lemma: for prop convexity} with the equality
    \begin{equation}
        \operatorname{P} ({\delta_{i}'}+1|\delta_{i}, \mathbf{G}, \hat{a}_{i}^{1}) = \operatorname{P} ({\delta_{i}}+1|\delta_{i}, \mathbf{G}, \hat{a}_{i}^{1}) = 1, 
    \end{equation}
    and 
    \begin{equation}
        \!\operatorname{P} ({\delta_{i}''}\!+\!1|\delta_{i}, \!\mathbf{G}, \check{a}_{i}^{1}) \!=\! \operatorname{P} ({\delta_{i}}\!+\!1|\delta_{i}, \!\mathbf{G}, \check{a}_{i}^{1}) \!=\! \operatorname{P} ({\delta_{j}}\!+\!1|\delta_{j}, \!\mathbf{G}, \hat{a}_{j}^{1})
    \end{equation}
    as $\check{a}_{i}^{1} = \hat{a}_{j}^{1}$.
    \item[(c)] If $\check{a}_{i}^{1} = 0, \hat{a}_{i}^{1}= m$ and $\check{a}_{j}^{1}= m$, we define another action $\dot{\mathbf{a}}^{1}$ where $\dot{a}_{i}^{1} = m, \dot{a}_{j}^{1} = 0$, and $\dot{\mathbf{a}}_{\backslash\!\{\!i,j\!\}} = \check{\mathbf{a}}_{\backslash\!\{\!i,j\!\}}$, then we have
    \begin{align}
        & \!\hspace{-0.7cm} \alpha \upsilon^{1}(\bm{\delta}'') + (1-\alpha) \upsilon^{1}(\bm{\delta}') - \upsilon^{1}(\bm{\delta}) 
        \\
        &\hspace{-0.7cm} \!\geq  \alpha Z(\bm{\delta}'', \check{\mathbf{a}}^{1}; \upsilon^{0}) + (1-\alpha) Z(\bm{\delta}', \hat{\mathbf{a}}^{1}; \upsilon^{0}) 
        \\
        &\hspace{-0.6cm} \!- \alpha Z(\bm{\delta}, \dot{\mathbf{a}}^{1}; \upsilon^{0}) - (1-\alpha) Z(\bm{\delta}, \hat{\mathbf{a}}^{1}; \upsilon^{0}) 
        \\
        &\hspace{-0.7cm} \!= \left[\alpha c(\bm{\delta}'') + (1-\alpha) c(\bm{\delta}') - \alpha c(\bm{\delta}) - (1-\alpha) c(\bm{\delta}) \right] 
        \\
        &\hspace{-0.6cm} \!+ \alpha \! \sum_{{\delta_{i}''}^{+}} \sum_{{\delta_{j}}^{+}} \! \sum_{\bm{\delta}_{\backslash\!\{\!i,j\!\}}^{+}} \operatorname{P} ({\delta_{i}''}^{+}|\delta_{i}'', \mathbf{G}, \check{a}_{i}^{1}) \operatorname{P} ({\delta_{j}}^{+}|\delta_{j}, \mathbf{G}, \check{a}_{j}^{1})
        \\
        &\times \operatorname{P} (\bm{\delta}_{\backslash\!\{\!i,j\!\}}^{+}|\bm{\delta}_{\backslash\!\{\!i,j\!\}}, \mathbf{G}, \check{\mathbf{a}}_{\backslash\!\{\!i,j\!\}}^{1}) \upsilon^{0} ({\bm{\delta}''}^{+}) 
        \\
        &\hspace{-0.6cm} \!+\! (1\!-\!\alpha) \!\sum_{{\delta_{i}'}^{+}} \! \sum_{\bm{\delta}_{\backslash\!\{\!i\!\}}^{+}} \operatorname{P} ({\delta_{i}'}^{+}|\delta_{i}', \!\mathbf{G}, \!\hat{a}_{i}^{1}) \operatorname{P} (\bm{\delta}_{\backslash\!\{\!i\!\}}^{+}|\bm{\delta}_{\backslash\!\{\!i\!\}}, \!\mathbf{G}, \hat{\mathbf{a}}_{\backslash\!\{\!i\!\}}^{1}) \upsilon^{0} ({\bm{\delta}'}^{+}) \!\!\!
        \\
        &\hspace{-0.6cm} \!- \alpha \! \sum_{{\delta_{i}}^{+}} \! \sum_{{\delta_{j}}^{+}} \! \sum_{\bm{\delta}_{\backslash\!\{\!i,j\!\}}^{+}} \operatorname{P} ({\delta_{i}}^{+}|\delta_{i}, \mathbf{G}, \dot{a}_{i}^{1}) \operatorname{P} ({\delta_{j}}^{+}|\delta_{j}, \mathbf{G}, \dot{a}_{j}^{1})
        \\
        & \times \operatorname{P} (\bm{\delta}_{\backslash\!\{\!i,j\!\}}^{+}|\bm{\delta}_{\backslash\!\{\!i,j\!\}}, \mathbf{G}, \dot{\mathbf{a}}_{\backslash\!\{\!i,j\!\}}^{1}) \upsilon^{0} ({\bm{\delta}}^{+}) 
        \\
        &\hspace{-0.6cm} \!-\! (1\!-\!\alpha) \!\sum_{{\delta_{i}}^{+}} \! \sum_{\bm{\delta}_{\backslash\!\{\!i\!\}}^{+}} \operatorname{P} ({\delta_{i}}^{+}|\delta_{i}, \!\mathbf{G}, \!\hat{a}_{i}^{1}) \operatorname{P} (\bm{\delta}_{\backslash\!\{\!i\!\}}^{+}|\bm{\delta}_{\backslash\!\{\!i\!\}}, \!\mathbf{G}, \hat{\mathbf{a}}_{\backslash\!\{\!i\!\}}^{1}) \upsilon^{0} ({\bm{\delta}}^{+}) \!\!\!
        \\
        &\hspace{-0.7cm} \!\geq \alpha \operatorname{P} ({\delta_{i}''}+1|\delta_{i}'', \mathbf{G}, \check{a}_{i}^{1}) \operatorname{P} ({\delta_{j}}+1|\delta_{j}, \mathbf{G}, \check{a}_{j}^{1})
        \\
        &\times \sum_{\bm{\delta}_{\backslash\!\{\!i,j\!\}}^{+}} \operatorname{P} (\bm{\delta}_{\backslash\!\{\!i,j\!\}}^{+}|\bm{\delta}_{\backslash\!\{\!i,j\!\}}, \mathbf{G}, \check{\mathbf{a}}_{\backslash\!\{\!i,j\!\}}^{1}) \upsilon^{0} ({\bm{\delta}''}^{+}) 
        \\
        &\hspace{-0.6cm} \!+\! (1\!-\!\alpha) \operatorname{P} ({\delta_{i}'}+1|\delta_{i}', \!\mathbf{G}, \hat{a}_{i}^{1}) \sum_{\bm{\delta}_{\backslash\!\{\!i\!\}}^{+}} \operatorname{P} (\bm{\delta}_{\backslash\!\{\!i\!\}}^{+}|\bm{\delta}_{\backslash\!\{\!i\!\}}, \!\mathbf{G}, \hat{\mathbf{a}}_{\backslash\!\{\!i\!\}}^{1}) \upsilon^{0} ({\bm{\delta}'}^{+})\!\!\!\!
        \\
        &\hspace{-0.6cm} - \!\alpha \operatorname{P} ({\delta_{i}}+1|\delta_{i}, \mathbf{G}, \dot{a}_{i}^{1}) \operatorname{P} ({\delta_{j}}+1|\delta_{j}, \mathbf{G}, \dot{a}_{j}^{1})
        \\
        & \times \sum_{\bm{\delta}_{\backslash\!\{\!i,j\!\}}^{+}} (\bm{\delta}_{\backslash\!\{\!i,j\!\}}^{+}|\bm{\delta}_{\backslash\!\{\!i,j\!\}}, \mathbf{G}, \dot{\mathbf{a}}_{\backslash\!\{\!i,j\!\}}^{1}) \upsilon^{0} ({\bm{\delta}}^{+}) 
        \\
        &\hspace{-0.6cm} \!-\! (1\!-\!\alpha) \operatorname{P} ({\delta_{i}}\!+\!1|\delta_{i}, \!\mathbf{G}, \hat{a}_{i}^{1}) \! \sum_{\bm{\delta}_{\backslash\!\{\!i,j\!\}}^{+}} \operatorname{P} ({\delta_{i}}\!+\!1|\delta_{i}, \!\mathbf{G}, \hat{a}_{i}^{1}) \upsilon^{0} ({\bm{\delta}}^{+}) 
        \\
        &\hspace{-0.6cm} + \alpha \operatorname{P} ({\delta_{i}''}+1|\delta_{i}'', \mathbf{G}, \check{a}_{i}^{1}) \operatorname{P} ({\delta_{j}}^{+}=1|\delta_{j}, \mathbf{G}, \check{a}_{j}^{1}) 
        \\
        &\times \sum_{\bm{\delta}_{\backslash\!\{\!i,j\!\}}^{+}} \operatorname{P} (\bm{\delta}_{\backslash\!\{\!i,j\!\}}^{+}|\bm{\delta}_{\backslash\!\{\!i,j\!\}}, \mathbf{G}, \check{\mathbf{a}}_{\backslash\!\{\!i,j\!\}}^{1}) \upsilon^{0} ({\bm{\delta}''}^{+}) 
        \\
        &\hspace{-0.7cm} \geq 0,
    \end{align}
    where the second inequality is derived from~\eqref{ineq: def cost convexity} and the following equality
    \begin{align}
        &\!\!\!\!\!\!\!\!\! \operatorname{P} ({\delta_{i}'}^{+}\!\!\!=\!1|\delta_{i}', \mathbf{G}, \!\hat{a}_{i}^{1}) \! \!\sum_{\bm{\delta}_{\backslash\!\{\!i\!\}}^{+}} \!\operatorname{P} (\bm{\delta}_{\backslash\!\{\!i\!\}}^{+}|\bm{\delta}_{\backslash\!\{\!i\!\}}, \!\mathbf{G}, \!\hat{\mathbf{a}}_{\backslash\!\{\!i\!\}}^{1}) \upsilon^{0} (1, \!{\bm{\delta}}_{\backslash\!\{\!i\!\}}^{+})
        \\
        &\!\!\!\!\!\!\!\!\! =\! \operatorname{P} ({\delta_{i}}^{+}\!\!\!=\!1|\delta_{i}, \!\mathbf{G}, \!\hat{a}_{i}^{1}) \!\! \sum_{\bm{\delta}_{\backslash\!\{\!i\!\}}^{+}} \!\operatorname{P} (\bm{\delta}_{\backslash\!\{\!i\!\}}^{+}\!|\bm{\delta}_{\backslash\!\{\!i\!\}}, \!\mathbf{G}, \!\hat{\mathbf{a}}_{\backslash\!\{\!i\!\}}^{1}) \upsilon^{0} (1, \!{\bm{\delta}}_{\backslash\!\{\!i\!\}}^{+}),\!\!\!\!\!
    \end{align}
    and
    \begin{equation}
        \operatorname{P} ({\delta_{i}''}^{+}=1|\delta_{i}'', \mathbf{G}, \check{a}_{i}^{1}) = 0
    \end{equation}
    from $\check{a}_{i}^{1} = 0$, and the inequality
    \begin{align}
        & \operatorname{P} ({\delta_{i}}+1|\delta_{i}, \mathbf{G}, \dot{a}_{i}^{1}) \upsilon^{0} ({\delta_{i}}+1, {\bm{\delta}}_{\backslash\!\{\!i\!\}}^{+}) 
        \\
        & \gg \operatorname{P} ({\delta_{i}}^{+}=1|\delta_{i}, \mathbf{G}, \dot{a}_{i}^{1}) \upsilon^{0} (1, {\bm{\delta}}_{\backslash\!\{\!i\!\}}^{+}),
    \end{align}
    achieved from Lemma~\ref{lemma:much larger} and ${\delta_{i}}+1 \gg 1$, and the last inequality is based on Lemma~\ref{lemma: for prop convexity} with the equality
    \begin{equation}
        \operatorname{P} ({\delta_{i}'}+1|\delta_{i}', \mathbf{G}, \hat{a}_{i}^{1}) = \operatorname{P} ({\delta_{j}}+1|\delta_{j}, \mathbf{G}, \dot{a}_{j}^{1}) = 1, 
    \end{equation}
    and
    \begin{align}
        & \operatorname{P} ({\delta_{j}}+1|\delta_{j}, \mathbf{G}, \check{a}_{j}^{1}) = \operatorname{P} ({\delta_{i}'}+1|\delta_{i}', \mathbf{G}, \hat{a}_{i}^{1})
        \\
        & = \operatorname{P} ({\delta_{i}}+1|\delta_{i}, \mathbf{G}, \dot{a}_{i}^{1}) = \operatorname{P} ({\delta_{i}}+1|\delta_{i}, \mathbf{G}, \hat{a}_{i}^{1}),
    \end{align}
    and the inequality $\upsilon^{0} ({\bm{\delta}''}^{+}) \geq 0$.
\end{itemize}
Therefore, the Bellman operation~\eqref{eq: bellman operator} preserve this convexity of the value function $\upsilon^{0}(\mathbf{s})$ to the optimal V function $\upsilon^{*}(\mathbf{s})$.

\section{Proof of Theorem~\ref{theo: structure gateway}} \label{proof: structure gateway}

To prove Theorem~\ref{theo: structure gateway}, it is sufficient to prove that for each action $\mathbf{a}$, where $a_{i} = 0, a_{j} = m$ for $i \in \mathcal{I}, j \notin \mathcal{I}$, we can find a better action $\dot{\mathbf{a}}$, where $\dot{a}_{i} = m$, $\dot{a}_{j} = 0$ and $\dot{\mathbf{a}}_{\backslash\!\{\!i,j\!\}} = \mathbf{a}_{\backslash\!\{\!i,j\!\}}$.
Therefore, based on~\eqref{ineq:V&Q}, the actions $\mathbf{a}$ and $\dot{\mathbf{a}}$ follows the inequality
\begin{equation}
    Q(\mathbf{s}, \dot{\mathbf{a}}) \leq Q(\mathbf{s}, \mathbf{a}),
    \label{ineq: Q(bara) < Q2(a), structure gateway}
\end{equation}
which implies that the optimal action of the device $i$ cannot be idle, i.e., $a^{*} \neq 0$ as in Theorem~\ref{theo: structure gateway}.
% The inequality~\eqref{ineq: Q(bara) < Q2(a), structure gateway} means that for any action $\mathbf{a}$, we can always find another better action $\dot{\mathbf{a}}$ with lower Q value, so that the optimal action of the device $i$ cannot be idle, i.e., $a^{*} \neq 0$ as in Theorem~\ref{theo: structure gateway}.
To prove~\eqref{ineq: Q(bara) < Q2(a), structure gateway}, we develop the following lemma of the optimal V function in an asymptotic form.
For writing simplicity, we write the AoI state as $\bm{\delta} = (\delta_{i}, \delta_{j}, \bm{\delta}_{\backslash\!\{\!i, j\!\}})$ in the following.

By using~\eqref{eq:Bellman_Q}, we have
\begin{align}
    \!\!\!\!\! Q(&\mathbf{s}, \dot{\mathbf{a}}) \!=\!\!\!  \sum_{{\bm{\delta}_{\backslash\!\{\!i,j\!\}}^{+}}} \!\!\sum_{\mathbf{G}^{+}} \!\sum_{\delta_{i}^{+}} \!\sum_{\delta_{j}^{+}} \operatorname{P}\left(\bm{\delta}_{\backslash\!\{\!i,j\!\}}^{+}| \bm{\delta}_{\backslash\!\{\!i,j\!\}}, \!\dot{\mathbf{a}}_{\backslash\!\{\!i,j\!\}},\! \mathbf{G}_{\backslash\!\{\!i,j\!\}}\right) \\
    & \times \! \operatorname{P} (\mathbf{G}^{+}) \operatorname{P} (\delta_{i}^{+}| \delta_{i}, \dot{a}_{i}, \mathbf{G}_{i}) \operatorname{P} (\delta_{j}^{+}| \delta_{j}, \dot{a}_{j}, \mathbf{G}_{j}) \upsilon(\mathbf{s}^{+}),
    \label{eq: Q(s, bar(a)), gateway theo}
\end{align}
and 
\begin{align}
    \!\!\!\!\! Q(&\mathbf{s}, \mathbf{a}) \!=\!\!\!  \sum_{{\bm{\delta}_{\backslash\!\{\!i,j\!\}}^{+}}} \!\sum_{\mathbf{G}^{+}}\! \sum_{\delta_{i}^{+}} \!\sum_{\delta_{j}^{+}} \operatorname{P}\left(\bm{\delta}_{\backslash\!\{\!i,j\!\}}^{+}| \bm{\delta}_{\backslash\!\{\!i,j\!\}}, \!\mathbf{a}_{\backslash\!\{\!i,j\!\}}, \!\mathbf{G}_{\backslash\!\{\!i,j\!\}}\right) \\
    & \times \! \operatorname{P} (\mathbf{G}^{+}) \operatorname{P} (\delta_{i}^{+}| \delta_{i}, a_{i}, \mathbf{G}_{i}) \operatorname{P} (\delta_{j}^{+}| \delta_{j}, a_{j}, \mathbf{G}_{j}) \upsilon(\mathbf{s}^{+}).
    \label{eq: Q(s, a), gateway theo}
\end{align}
Since $\dot{\mathbf{a}}_{\backslash\!\{\!i,j\!\}} = \mathbf{a}_{\backslash\!\{\!i,j\!\}}$ and $a_{i} = \dot{a}_{j} = 0$, we derive that
\begin{align}
    & \operatorname{P}(\bm{\delta}_{\backslash\!\{\!i,j\!\}}^{+}| \bm{\delta}_{\backslash\!\{\!i,j\!\}}, \dot{\mathbf{a}}_{\backslash\!\{\!i,j\!\}}, \mathbf{G}_{\backslash\!\{\!i,j\!\}}) \\
    & = \operatorname{P}(\bm{\delta}_{\backslash\!\{\!i,j\!\}}^{+}| \bm{\delta}_{\backslash\!\{\!i,j\!\}}, \mathbf{a}_{\backslash\!\{\!i,j\!\}}, \mathbf{G}_{\backslash\!\{\!i,j\!\}}), \!\!
    \label{eq: Pr(bar(a)) = Pr(a) 1}
\end{align}
and 
\begin{equation}
    \operatorname{P} (\delta_{j}^{+} \!= \delta_{j}+1| \delta_{j}, \dot{a}_{j}, \mathbf{G}_{j}) \!=\! \operatorname{P} (\delta_{i}^{+} \!= \delta_{i}+1| \delta_{i}, a_{i}, \mathbf{G}_{i}) = 1.
    \label{eq: Pr(bar(a)) = Pr(a) 2}
\end{equation}
Based on~\eqref{eq: Q(s, bar(a)), gateway theo},~\eqref{eq: Q(s, a), gateway theo},~\eqref{eq: Pr(bar(a)) = Pr(a) 1}, and~\eqref{eq: Pr(bar(a)) = Pr(a) 2}, the inequality~\eqref{ineq: Q(bara) < Q2(a), structure gateway} is equivalent to
\begin{equation}
    \sum_{\delta_{i}^{+}}  \operatorname{P} (\delta_{i}^{+}| \delta_{i}, \dot{a}_{i}, \mathbf{G}_{i})  \upsilon(\mathbf{s}^{+}) 
    \leq \sum_{\delta_{j}^{+}}  \operatorname{P} (\delta_{j}^{+}| \delta_{j}, a_{j}, \mathbf{G}_{j}) \upsilon(\mathbf{s}^{+}).
    \label{ineq: Q(bara) < Q2(a) equivalent, structure gateway}
\end{equation}
Next, the following equality
\begin{align}
    &\!\!\!\!\!\!\! \operatorname{P} (\delta_{i}^{+}\!=\delta_{i}\!+\!1| \delta_{i}, \dot{a}_{i}, \mathbf{G}_{i}) \upsilon(\delta_{i}\!+\!1, \delta_{j}\!+\!1, \bm{\delta}_{\backslash\!\{\!i,j\!\}}^{+}, \mathbf{G}^{+}) \\
    &\!\!\!\!\!\!\! =\! \operatorname{P} (\delta_{j}^{+}\!=\delta_{j}\!+\!1| \delta_{j}, a_{j}, \mathbf{G}_{j}) \upsilon(\delta_{i}\!+\!1, \delta_{j}\!+\!1, \bm{\delta}_{\backslash\!\{\!i,j\!\}}^{+}, \mathbf{G}^{+})
    \label{eq: Pr(bar(a)) = Pr(a) 3}
\end{align}
and 
\begin{equation}
    \operatorname{P} (\delta_{i}^{+}=1| \delta_{i}, \dot{a}_{i}, \mathbf{G}_{i}) = \operatorname{P} (\delta_{j}^{+}=1| \delta_{j}, a_{j}, \mathbf{G}_{j})
    \label{eq: Pr(bar(a)) = Pr(a) 4}
\end{equation}
are obtained from $\dot{a}_{i} = a_{j} = m$, ${g}_{i,m} = {g}_{j,m}$ and~\eqref{eq: Pr(bar(a)) = Pr(a) 2}.
From~\eqref{eq: Pr(bar(a)) = Pr(a) 3} and~\eqref{eq: Pr(bar(a)) = Pr(a) 4}, to prove~\eqref{ineq: Q(bara) < Q2(a) equivalent, structure gateway}, it is sufficient to show that
\begin{equation}
    \upsilon(1, \delta_{j}\!+\!1, \bm{\delta}_{\backslash\!\{\!i,j\!\}}^{+},\mathbf{G}^{+}) \leq \upsilon(\delta_{i}\!+\!1, 1, \bm{\delta}_{\backslash\!\{\!i,j\!\}}^{+}, \mathbf{G}^{+}),
    \label{ineq: V comparison with two different aoi delta=1}
\end{equation}
when $\delta_{i} > \bar{\delta} \gg 1$.

The inquality~\eqref{ineq: V comparison with two different aoi delta=1} is equivalent to the following Lemma.
\begin{lemma}
Consider a multi-device-multi-channel system with co-located devices. 
For states $\mathbf{s} = (\bm{\delta}, \mathbf{G})$ and $\mathbf{s}^{\circ} = (\bm{\delta}^{\circ}, \mathbf{G})$, where $\bm{\delta} = (\delta_{i}, \delta_{j}, \bm{\delta}_{\backslash\!\{\!i,j\!\}})$, and $\bm{\delta}^{\circ} = (\delta_{i}', \delta''_{j}, \bm{\delta}_{\backslash\!\{\!i,j\!\}})$ with $\delta'_{i} \geq \delta_{j} \geq \delta''_{j}$ and $\delta'_{i} \gg \delta_{i}, \forall i \in \mathcal{I}, j \notin \mathcal{I}$, the following inequality hold
    \begin{equation}
        \upsilon^{*}(\mathbf{s}^{\circ}) \geq \upsilon^{*}(\mathbf{s}).
        \label{ineq: V comparison with two different aoi}
    \end{equation}
    \label{lemma: V comparison with two different aoi}
\end{lemma}
\begin{proof}
    See Appendix~\ref{proof: V comparison with two different aoi}.
\end{proof}

Therefore, the inquality~\eqref{ineq: Q(bara) < Q2(a), structure gateway} holds under Lemma~\ref{lemma: V comparison with two different aoi}, which is exactly Theorem~\ref{theo: structure gateway}.

\section{Proof of Lemma~\ref{lemma: V comparison with two different aoi}} \label{proof: V comparison with two different aoi}
Based on the monotonicity of the optimal V function in Lemma~\ref{lemma:monotone V}, to prove Lemma~\ref{lemma: V comparison with two different aoi}, it is sufficient to prove~\eqref{ineq: V comparison with two different aoi} hold when the state $\mathbf{s} = \hat{\mathbf{s}}_{(j)} = (\hat{\bm{\delta}}^{(j)}, \mathbf{G})$, where $\hat{\bm{\delta}}^{(j)} = (\delta_{i}, \delta_{j}, \bm{\delta}_{\backslash\!\{\!i,j\!\}})$ and $\delta'_{j} = \delta'_{i} \geq \delta_{j}$, i.e., 
\begin{equation}
    \upsilon^{*}(\hat{\mathbf{s}}_{(j)}) \leq \upsilon^{*}(\mathbf{s}^{\circ}).
    \label{ineq: V comparison with two different aoi 2}
\end{equation}
Similar to the proof of Theorem~\ref{theo: convexity, N-M}, proving~\eqref{ineq: V comparison with two different aoi 2} is equivalent to proving
\begin{equation}
    \upsilon^{1}(\hat{\mathbf{s}}_{(j)}) \leq \upsilon^{1}(\mathbf{s}^{\circ})
    \label{ineq: V comparison with two different aoi, 1}
\end{equation}
under the assumption of
\begin{equation}
    \upsilon^{0}(\hat{\mathbf{s}}_{(j)}) \leq \upsilon^{0}(\mathbf{s}^{\circ}).
    \label{ineq: V comparison with two different aoi, 0}
\end{equation}
From the cost function $c_{i}(\delta) \geq c_{j}(\delta) \forall i \in \mathcal{I}, j \notin \mathcal{I}, \delta \gg 1$, and $\delta'_{i} = \delta'_{j} \gg \delta_{i}$, we have
\begin{equation}
    c(\hat{\mathbf{s}}_{(j)}) \leq c(\mathbf{s}^{\circ}).
    \label{ineq: cost comparison with two different aoi}
\end{equation}
Similar to the proof of Theorem~\ref{theo: structure gateway}, in the 1st iteration, we can prove that the optimal action of the device $i$ w.r.t. the state $\mathbf{s}^{\circ}$ should be scheduled, i.e., $\mathbf{a}^{1} = \pi (\mathbf{s}^{\circ})$ and $a^{1}_{i} \neq 0$, under the assumption of the $\upsilon^{0}$ in~\eqref{ineq: V comparison with two different aoi, 0}.
Since the devices are co-located and the packet drop rate is independent of the scheduled device, we represent the packet drop rate as $\psi_{m} = \psi_{i,m}$ during the proof.

Thus, in the following, we write the state as $\mathbf{s} = \bm{\delta}$, $\mathbf{s}'_{\text{AoI}} = \bm{\delta}'$, and $\mathbf{s}''_{\text{AoI}} = \bm{\delta}''$ and prove~\eqref{ineq: V comparison with two different aoi, 1} based on different cases with different optimal actions where $a^{1}_{i} \neq 0$.
Also, for writing simplicity, we write the AoI state as $\bm{\delta} = (\delta_{i}, \delta_{j}, \bm{\delta}_{\backslash\!\{\!i, j\!\}})$.
\begin{itemize}
    \item[(a)] If $a^{1}_{i} = m_{1}$ and $a^{1}_{j} = m_{2}$, then we define another action $\dot{\mathbf{a}}^{1}$ where $\dot{a}^{1}_{i} = m_{2}$, $\dot{a}^{1}_{j} = m_{1}$, and $\dot{\mathbf{a}}^{1}_{\backslash\!\{\!i,j\!\}} = \mathbf{a}^{1}_{\backslash\!\{\!i,j\!\}}$ for the state $\hat{\bm{\delta}}^{(j)}$ in first value iteration.
    By using the actions $\mathbf{a}^{1}$ and $\dot{\mathbf{a}}^{1}$, we derive the AoI state transition probability as:
    \begin{align}
        & \!\!\!\!\!\!\!\!\!\!\!\!\!\!\!\!\!\!\operatorname{P} ({\delta'_{i}}^{+} \!\!\!=\! 1| \delta'_{i},\! a^{1}_{i}, \!\mathbf{G}_{i}) \!=\! \operatorname{P} ({\delta'_{j}}^{+} \!\!\!=\! 1| \delta'_{j}, \dot{a}^{1}_{j}, \!\mathbf{G}_{j}) \!=\! (1\!-\!\psi_{m_{1}}),
        \label{eq: P(aoi = 1), 1}
        \\
        & \!\!\!\!\!\!\!\!\!\!\!\!\!\!\!\!\!\!\operatorname{P} ({\delta''_{j}}^{+} \!\!\!=\! 1| \delta''_{j}, a^{1}_{j},\!\mathbf{G}_{j}) \!=\! \operatorname{P} ({\delta}_{i}^{+} \!\!=\! 1| \delta_{i}, \dot{a}^{1}_{i}, \!\mathbf{G}_{i}) \!=\! (1\!-\!\psi_{m_{2}}),
        \label{eq: P(aoi = 1), 2}
    \end{align}
    and
    \begin{align}
        &\!\!\!\!\!\! \operatorname{P} (\delta'_{i} \!+\! 1| \delta'_{i}, a^{1}_{i}, \mathbf{G}_{i}) \!=\! \operatorname{P} (\delta'_{j}\!+\!1| \delta'_{j}, \dot{a}^{1}_{j}, \mathbf{G}_{j}) \!=\! \psi_{m_{1}},
        \label{eq: P(aoi += 1), 1}
        \\
        &\!\!\!\!\!\! \operatorname{P} (\delta''_{j} \!+\! 1| \delta''_{j}, a^{1}_{j}, \mathbf{G}_{j}) \!=\! \operatorname{P} (\delta_{i}\!+\!1| \delta_{i}, \dot{a}^{1}_{i}, \mathbf{G}_{i}) \!=\! \psi_{m_{2}}.
        \label{eq: P(aoi += 1), 2}
    \end{align}
    Next, we have
    \begin{align}
        &\!\!\!\!\!\!\! \upsilon^{1} (\bm{\delta}^{\circ}) - \upsilon^{1}(\hat{\bm{\delta}}^{(j)}) 
        \\
        &\!\!\!\!\!\!\! \geq Z(\bm{\delta}^{\circ}, \mathbf{a}^{1}, \upsilon_{0}) - Z(\hat{\bm{\delta}}^{(j)}, \dot{\mathbf{a}}^{1}, \upsilon_{0}) \\
        &\!\!\!\!\!\!\! = \left[ c(\bm{\delta}^{\circ}) - c(\hat{\bm{\delta}}^{(j)}) \right] \!+\!\! \sum_{{\bm{\delta}_{\backslash\!\{\!i,j\!\}}^{+}}}  \operatorname{P}(\bm{\delta}_{\backslash\!\{\!i,j\!\}}^{+}| \bm{\delta}_{\backslash\!\{\!i,j\!\}}, \mathbf{a}^{1}_{\backslash\!\{\!i,j\!\}}, \mathbf{G}_{\backslash\!\{\!i,j\!\}})
        \\
        &\!\!\!\!\! \times \Bigg[ \sum_{{\delta'_{i}}^{+}} \sum_{{\delta''_{j}}^{+}} \operatorname{P} ({\delta'_{i}}^{+}| \delta'_{i}, a^{1}_{i}, \!\mathbf{G}_{i}) \operatorname{P} ({\delta''_{j}}^{+}| \delta_{j}, a^{1}_{j}, \!\mathbf{G}_{j}) \upsilon^{0}({\bm{\delta}^{\circ}}^{+}) \!\!
        \\
        &\!\!\! - \sum_{\delta_{i}^{+}} \sum_{{\delta'_{j}}^{+}} \operatorname{P} (\delta_{i}^{+}| \delta_{i}, \dot{a}^{1}_{i}, \mathbf{G}_{i}) \operatorname{P} ({\delta'_{j}}^{+}| \delta'_{j}, \dot{a}^{1}_{j}, \mathbf{G}_{j}) \upsilon^{0}{(\hat{\bm{\delta}}^{(j)}}^{+}) \Bigg] 
        \\
        &\!\!\!\!\!\!\! \geq 0,
        \label{ineq: v1 - v0 > 0, lemma4, case a}
        \end{align}
        where the first inequality is derived from~\eqref{eq: Z-function}, the first equality is from~\eqref{eq: Z-function} and $\dot{\mathbf{a}}^{1}_{\backslash\!\{\!i,j\!\}} = \mathbf{a}^{1}_{\backslash\!\{\!i,j\!\}}$, and the last inequality is from~\eqref{ineq: cost comparison with two different aoi},~\eqref{eq: P(aoi = 1), 1},~\eqref{eq: P(aoi = 1), 2},~\eqref{eq: P(aoi += 1), 1},~\eqref{eq: P(aoi += 1), 2}, and the following inequality:
        \begin{align}
            &\!\!\!\!\!\!\!\!\!\! (1 \!-\! \psi_{m_{1}}) (1 \!-\! \psi_{m_{2}}) \!\left[\upsilon^{0}(1,\!1, \bm{\delta}_{\backslash\!\{\!i,j\!\}}^{+}) \!-\! \upsilon^{0}(1,\!1, \bm{\delta}_{\backslash\!\{\!i,j\!\}}^{+})\right] \!\!\!
            \\
            &\!\!\!\!\!\!\!\!\!\! + \!\psi_{m_{1}} (1 \!-\!\psi_{m_{2}}) \!\left[\upsilon^{0}\!\left(\!\delta'_{i} \!+\! 1,\!1, \!\bm{\delta}_{\backslash\!\{\!i,j\!\}}^{+}\!\right) \!-\! \upsilon^{0}\!\left(\!1,\delta'_{j}\!+\!1,\!\bm{\delta}_{\backslash\!\{\!i,j\!\}}^{+}\!\right)\right] \!\!\!
            \\
            &\!\!\!\!\!\!\!\!\!\! +\! (1 \!-\! \psi_{m_{1}}) \psi_{m_{2}} \!\left[\upsilon^{0}\!\left(\!1,\delta''_{j} \!+\!\! 1, \bm{\delta}_{\backslash\!\{\!i,j\!\}}^{+}\!\right) \!-\! \upsilon^{0}\!\left(\!\delta_{i}\!+\!1,\!1, \bm{\delta}_{\backslash\!\{\!i,j\!\}}^{+}\!\right)\right] \!\!\!\!
            \\
            &\!\!\!\!\!\!\!\!\!\! +\! \psi_{m_{1}} \psi_{m_{2}} \!\left[\upsilon^{0}\!\!\left(\!\delta'_{i} \!+\!\! 1, \delta''_{j} \!+\!\! 1, \!\bm{\delta}_{\backslash\!\{\!i,j\!\}}^{+}\!\right) \!-\! \upsilon^{0}\!\!\left
            (\!\delta_{i}\!+\!\!1,\delta'_{j}\!+\!\!1, \!\bm{\delta}_{\backslash\!\{\!i,j\!\}}^{+}\!\right)\right] \!\!\!\!\!
            \\
            &\!\!\!\!\!\!\!\!\!\! \geq 0,
        \end{align}
        achieved from~\eqref{ineq: V comparison with two different aoi, 0} and $\upsilon^{0}(\delta'_{i} \!+\!\! 1, \delta''_{j} \!+\!\! 1, \bm{\delta}_{\backslash\!\{\!i,j\!\}}^{+}) \gg \upsilon^{0}(\delta_{i}\!+\!1,1, \bm{\delta}_{\backslash\!\{\!i,j\!\}}^{+})$ with Lemma~\ref{lemma:much larger} and $\delta'_{i} + 1 \gg \delta_{i} + 1$.
    \item[(b)] If $a^{1}_{i} = m$ and $a^{1}_{j} = 0$, then we also define the action $\dot{\mathbf{a}}^{1}$ where $\dot{a}^{1}_{i} = 0$, $\dot{a}^{1}_{j} = m$, and $\dot{\mathbf{a}}^{1}_{\backslash\!\{\!i,j\!\}} = \mathbf{a}^{1}_{\backslash\!\{\!i,j\!\}}$.
    In case (b), the transition probability of the AoI states $\delta''_{j}$ and $\delta_{i}$ in the inequality~\eqref{ineq: v1 - v0 > 0, lemma4, case a} of case (a) are replaced by
    \begin{align}
        \operatorname{P} (\delta''_{j} \!+\! 1| \delta''_{j}, a^{1}_{j}, \mathbf{G}_{j}) = \operatorname{P} (\delta_{i} \!+\! 1| \delta_{i}, a^{1}_{i}, \mathbf{G}_{i}) = 1, 
        \\
        \operatorname{P} ({\delta''_{j}}^{+} \!=\! 1| \delta''_{j}, a^{1}_{j}, \mathbf{G}_{j}) = \operatorname{P} (\delta_{i}^{+} \!=\! 1| \delta_{i}, a^{1}_{i}, \mathbf{G}_{i}) = 0,
    \end{align}
    and the other parts are kept constant.
    Therefore, to prove $\upsilon^{1} (\bm{\delta}^{\circ}) - \upsilon^{1}(\hat{\bm{\delta}}^{(j)})$ in this case based on the proof of case (a), it is sufficient to prove that 
    \begin{align}
        & (1 - \psi_{m}) \!\left[\upsilon^{0}(1,\delta''_{j} \!+\! 1, \bm{\delta}_{\backslash\!\{\!i,j\!\}}^{+}) \!-\! \upsilon^{0}(\delta_{i}\!+\!1,1, \bm{\delta}_{\backslash\!\{\!i,j\!\}}^{+})\right]
        \\
        & + \psi_{m} \!\left[\upsilon^{0}(\delta'_{i} \!+\!\! 1, \delta''_{j} \!+\!\! 1, \bm{\delta}_{\backslash\!\{\!i,j\!\}}^{+}) \!-\! \upsilon^{0}(\delta_{i}\!+\!\!1,\delta'_{j}\!+\!\!1, \bm{\delta}_{\backslash\!\{\!i,j\!\}}^{+})\right] 
        \\
        & \geq 0,
    \end{align}
    which is derived based on~\eqref{ineq: V comparison with two different aoi, 0} and $\upsilon^{0}(\delta'_{i} \!+\!\! 1, \delta''_{j} \!+\!\! 1, \bm{\delta}_{\backslash\!\{\!i,j\!\}}^{+}) \gg \upsilon^{0}(\delta_{i}\!+\!1,1, \bm{\delta}_{\backslash\!\{\!i,j\!\}}^{+})$ by using Lemma~\ref{lemma:much larger} and $\delta'_{i} + 1 \gg \delta_{i} + 1$.
\end{itemize}

Therefore, the property of the value function $\upsilon^{0}(\mathbf{s})$ as shown in~\eqref{ineq: V comparison with two different aoi, 0} can be preserved by the Bellman operation~\eqref{eq: bellman operator} to the optimal V function $\upsilon^{*}(\mathbf{s})$.

\section{Proof of Lemma~\ref{lemma: for prop convexity}} \label{proof: for prop convexity}
Similar to the proof of Theorem~\ref{theo: convexity, 2-1}, to prove Proposition~\ref{prop: V convexity} based on Lemma~\ref{lemma:converge of V*}, it is sufficient to prove that $\upsilon^{1}(\mathbf{s})$ holds the following inequality
\begin{equation}\label{ineq: for prop convex, 1}
    \alpha \upsilon^{1}(\dot{\mathbf{s}}'') + (1-\alpha) \upsilon^{1}(\ddot{\mathbf{s}}') \geq \alpha \upsilon^{1}(\dot{\mathbf{s}}) + (1-\alpha) \upsilon^{1}(\ddot{\mathbf{s}}),
\end{equation}
under the assumption that the value function $\upsilon^{0}(\mathbf{s})$ holds the inequality
\begin{equation}\label{ineq: for prop convex, 0}
    \alpha \upsilon^{0}(\dot{\mathbf{s}}'') + (1-\alpha) \upsilon^{0}(\ddot{\mathbf{s}}') \geq \alpha \upsilon^{0}(\dot{\mathbf{s}}) + (1-\alpha) \upsilon^{0}(\ddot{\mathbf{s}}),
\end{equation}
where $\dot{\mathbf{s}}'' = (\delta''_{i}, \dot{\bm{\delta}}_{\backslash\!\{\!i\!\}}, \mathbf{G})$, $\ddot{\mathbf{s}}' = (\delta'_{i}, \ddot{\bm{\delta}}_{\backslash\!\{\!i\!\}}, \mathbf{G})$, $\dot{\mathbf{s}} = (\delta_{i}, \dot{\bm{\delta}}_{\backslash\!\{\!i\!\}}, \mathbf{G})$, $\ddot{\mathbf{s}} = (\delta_{i}, \ddot{\bm{\delta}}_{\backslash\!\{\!i\!\}}, \mathbf{G})$, $\alpha\delta''_{i} + (1-\alpha)\delta'_{i} = \delta_{i}$, $\alpha \in [0, 1]$, and $\delta'_{i} \geq \delta_{i} \geq \delta''_{i} \gg 1$.

In the following, we write the state as $\dot{\mathbf{s}}'' = \dot{\bm{\delta}}''$, $\ddot{\mathbf{s}}' = \ddot{\bm{\delta}}'$, $\dot{\mathbf{s}} = \dot{\bm{\delta}}$, and $\ddot{\mathbf{s}} = \ddot{\bm{\delta}}$ and prove~\eqref{ineq: for prop convex, 1} based on different cases with different optimal actions $\check{\mathbf{a}}^{1} = {\pi^{1}} (\dot{\bm{\delta}}''), \hat{\mathbf{a}}^{1} = {\pi^{1}} (\ddot{\bm{\delta}}')$.

\begin{itemize}
    \item[(a)] If $\check{a}_{i}^{1} = m_{1}, \hat{a}_{i}^{1}= m_{2}$, then there are 2 cases with different packet drop rate: (a.1) $\psi_{i,m_{1}} \leq \psi_{i,m_{2}}$ and (a.2) $\psi_{i,m_{1}} > \psi_{i,m_{2}}$.
    \item[(a.1)] If $\psi_{i,m_{1}} \leq \psi_{i,m_{2}}$ and $\hat{a}_{j}^{1}= m_{1}$, we define another action $\dot{\mathbf{a}}^{1}$, where $\dot{a}_{i}^{1} = m_{1}, \dot{a}_{j}^{1} = m_{2}$, and $\dot{\mathbf{a}}^{1}_{\backslash\!\{\!i,j\!\}} = \hat{\mathbf{a}}^{1}_{\backslash\!\{\!i,j\!\}}$, then we have $\psi_{j,m_{1}} < \psi_{j,m_{2}}$ and
    \begin{align}
        & \hspace{-1cm}\operatorname{P} (\delta_{i}'\!+\!1|\delta_{i}', \!\mathbf{G}, \hat{a}_{i}^{1}) \sum_{\ddot{\delta}_{j}^{+}}  \operatorname{P} (\ddot{\delta}_{j}^{+}|\ddot{\delta}_{j}, \!\mathbf{G}, \hat{a}_{j}^{1}) \upsilon^{0} (\delta_{i}'\!+\!1, \bm{\delta}_{\backslash\!\{\!i\!\}}^{+}) 
        \\ 
        & \hspace{-1cm}\geq\! \operatorname{P} (\delta_{i}'\!+\!1|\delta_{i}', \!\mathbf{G}, \dot{a}_{i}^{1}) \!\!\sum_{\ddot{\delta}_{j}^{+}} \operatorname{P} (\ddot{\delta}_{j}^{+}|\ddot{\delta}_{j}, \!\mathbf{G}, \!\dot{a}_{j}^{1}) \upsilon^{0} (\delta_{i}'\!+\!1, \!\bm{\delta}_{\backslash\!\{\!i\!\}}^{+}).
        \label{eq: for prop convexity, case a.1}
    \end{align}
    Next, we derive that
        \begin{align}
            & \!\hspace{-0.7cm} \alpha \upsilon^{1}(\dot{\bm{\delta}}'') + (1-\alpha) \upsilon^{1}(\ddot{\bm{\delta}}') - \alpha \upsilon^{1}(\dot{\bm{\delta}}) - (1-\alpha) \upsilon^{1}(\ddot{\bm{\delta}})
            \\
            &\hspace{-0.7cm} \!\geq  \alpha Z(\dot{\bm{\delta}}'', \check{\mathbf{a}}^{1}; \upsilon^{0}) + (1-\alpha) Z(\ddot{\bm{\delta}}', \hat{\mathbf{a}}^{1}; \upsilon^{0}) 
            \\
            &\hspace{-0.6cm} \!- \alpha Z(\dot{\bm{\delta}}, \check{\mathbf{a}}^{1}; \upsilon^{0}) - (1-\alpha) Z(\ddot{\bm{\delta}}, \dot{\mathbf{a}}^{1}; \upsilon^{0}) 
            \\
            &\hspace{-0.7cm} \!= \left[\alpha c(\dot{\bm{\delta}}'') + (1-\alpha) c(\ddot{\bm{\delta}}') - \alpha c(\dot{\bm{\delta}}) - (1-\alpha) c(\ddot{\bm{\delta}}) \right] 
            \\
            &\hspace{-0.6cm} \!+ \alpha \! \sum_{{\delta_{i}''}^{+}} \! \sum_{{\dot{\bm{\delta}}_{\backslash\!\{\!i\!\}}}^{+}} \operatorname{P} ({\delta_{i}''}^{+}|\delta_{i}'', \mathbf{G}, \check{a}_{i}^{1}) \! \operatorname{P} (\dot{\bm{\delta}}_{\backslash\!\{\!i\!\}}^{+}|\dot{\bm{\delta}}_{\backslash\!\{\!i\!\}}, \mathbf{G}, \check{\mathbf{a}}_{\backslash\!\{\!i\!\}}^{1}) \upsilon^{0} (\dot{\bm{\delta}}''^{+})  \!\!
            \\
            &\hspace{-0.6cm} \!+ (1\!-\!\alpha) \sum_{{\delta_{i}'}^{+}} \sum_{\delta_{j}^{+}} \sum_{{\ddot{\bm{\delta}}_{\backslash\!\{\!i,j\!\}}}^{+}}
            \operatorname{P} ({\delta_{i}'}^{+}|\delta_{i}', \mathbf{G}, \hat{a}_{i}^{1}) 
            \operatorname{P} (\dot{\delta}_{j}^{+}|\dot{\delta}_{j}, \mathbf{G}, \hat{a}_{j}^{1}) 
            \\
            & \times \operatorname{P} (\ddot{\bm{\delta}}_{\backslash\!\{\!i,j\!\}}^{+}|\ddot{\bm{\delta}}_{\backslash\!\{\!i,j\!\}}, \mathbf{G}, \hat{\mathbf{a}}_{\backslash\!\{\!i,j\!\}}^{1}) \upsilon^{0} (\ddot{\bm{\delta}}'^{+}) 
            \\
            &\hspace{-0.6cm} \!- \alpha \! \sum_{\delta_{i}^{+}} \! \sum_{{\dot{\bm{\delta}}_{\backslash\!\{\!i\!\}}}^{+}} \operatorname{P} (\delta_{i}^{+}|\delta_{i}, \mathbf{G}, \check{a}_{i}^{1}) \! \operatorname{P} (\dot{\bm{\delta}}_{\backslash\!\{\!i\!\}}^{+}|\dot{\bm{\delta}}_{\backslash\!\{\!i\!\}}, \mathbf{G}, \check{\mathbf{a}}_{\backslash\!\{\!i\!\}}^{1}) \upsilon^{0} ({\dot{\bm{\delta}}}^{+}) 
            \\
            & \hspace{-0.6cm} \!- (1-\alpha) \!\sum_{\delta_{i}^{+}}\! \sum_{\ddot{\delta}_{j}^{+}} \sum_{{\ddot{\bm{\delta}}_{\backslash\!\{\!i,j\!\}}}^{+}}
            \operatorname{P} (\delta_{i}^{+}|\delta_{i},\! \mathbf{G}, \dot{a}_{i}^{1}) 
            \operatorname{P} (\ddot{\delta}_{j}^{+}|\ddot{\delta}_{j},\! \mathbf{G}, \dot{a}_{j}^{1}) 
            \\
            & \times \operatorname{P} (\ddot{\bm{\delta}}_{\backslash\!\{\!i,j\!\}}^{+}|\ddot{\bm{\delta}}_{\backslash\!\{\!i,j\!\}},\! \mathbf{G}, \dot{\mathbf{a}}_{\backslash\!\{\!i,j\!\}}^{1}) \upsilon^{0} (\ddot{\bm{\delta}}^{+}) \!\!\!
            \\
            & \hspace{-0.7cm} \!\geq\! \alpha \operatorname{P} ({\delta_{i}''}+1|\delta_{i}'', \mathbf{G}, \check{a}_{i}^{1}) \! \sum_{{\dot{\bm{\delta}}_{\backslash\!\{\!i\!\}}}^{+}} \operatorname{P} (\dot{\bm{\delta}}_{\backslash\!\{\!i\!\}}^{+}|\dot{\bm{\delta}}_{\backslash\!\{\!i\!\}}, \mathbf{G}, \check{\mathbf{a}}_{\backslash\!\{\!i\!\}}^{1}) \upsilon^{0} (\dot{\bm{\delta}}''^{+})  \!\!
            \\
            & \hspace{-0.6cm} \!+ (1\!-\!\alpha) \operatorname{P} ({\delta_{i}'}+1|\delta_{i}', \mathbf{G}, \dot{a}_{i}^{1}) \sum_{\ddot{\delta}_{j}^{+}} \sum_{{\ddot{\bm{\delta}}_{\backslash\!\{\!i,j\!\}}}^{+}}
            \operatorname{P} (\ddot{\delta}_{j}^{+}|\ddot{\delta}_{j}, \mathbf{G}, \dot{a}_{j}^{1}) 
            \\
            & \times \operatorname{P} (\ddot{\bm{\delta}}_{\backslash\!\{\!i,j\!\}}^{+}|\ddot{\bm{\delta}}_{\backslash\!\{\!i,j\!\}}, \mathbf{G}, \dot{\mathbf{a}}_{\backslash\!\{\!i,j\!\}}^{1}) \upsilon^{0} (\ddot{\bm{\delta}}'^{+}) 
            \\
            & \hspace{-0.6cm} \!- \alpha \operatorname{P} (\delta_{i}+1|\delta_{i}, \mathbf{G}, \check{a}_{i}^{1}) \! \sum_{{\dot{\bm{\delta}}_{\backslash\!\{\!i\!\}}}^{+}}\operatorname{P} (\dot{\bm{\delta}}_{\backslash\!\{\!i\!\}}^{+}|\dot{\bm{\delta}}_{\backslash\!\{\!i\!\}}, \mathbf{G}, \check{\mathbf{a}}_{\backslash\!\{\!i\!\}}^{1}) \upsilon^{0} (\dot{\bm{\delta}}^{+}) 
            \\ 
            & \hspace{-0.6cm} \!- (1-\alpha) \operatorname{P} (\delta_{i}+1|\delta_{i},\! \mathbf{G}, \dot{a}_{i}^{1})  \! \sum_{\ddot{\delta}_{j}^{+}} \sum_{{\ddot{\bm{\delta}}_{\backslash\!\{\!i,j\!\}}}^{+}}
            \operatorname{P} (\ddot{\delta}_{j}^{+}|\ddot{\delta}_{j},\! \mathbf{G}, \dot{a}_{j}^{1}) 
            \\
            & \times \operatorname{P} (\ddot{\bm{\delta}}_{\backslash\!\{\!i,j\!\}}^{+}|\ddot{\bm{\delta}}_{\backslash\!\{\!i,j\!\}},\! \mathbf{G}, \dot{\mathbf{a}}_{\backslash\!\{\!i,j\!\}}^{1}) \upsilon^{0} (\ddot{\bm{\delta}}^{+}) \!\!\!
            \\
            & \hspace{-0.7cm} \geq 0,
        \end{align}
         where the first inequality is from~\eqref{ineq:V&Z}, and the first equality is derived based on~\eqref{eq: Z-function}, and the second inequality is from~\eqref{ineq: def cost convexity},~\eqref{eq: for prop convexity, case a.1}, and the following equality
         \begin{align}
            & \operatorname{P} ({\delta_{i}''}^{+}=1|\delta_{i}'', \mathbf{G}, \check{a}_{i}^{1}) \upsilon^{0} (1, \dot{\bm{\delta}}_{\backslash\!\{\!i\!\}}^{+})
            \\
            & = \operatorname{P} ({\delta_{i}}^{+}=1|\delta_{i}, \mathbf{G}, \check{a}_{i}^{1}) \upsilon^{0} (1, \dot{\bm{\delta}}_{\backslash\!\{\!i\!\}}^{+}),
        \end{align}
        and inequality
        \begin{align}
            & \operatorname{P} (\delta_{i}'+1|\delta_{i}', \mathbf{G}, \hat{a}_{i}^{1}) \upsilon^{0} (\delta_{i}'+1, \ddot{\bm{\delta}}_{\backslash\!\{\!i\!\}}^{+}) 
            \\
            & \gg \operatorname{P} ({\delta_{i}'}^{+}=1|\delta_{i}', \mathbf{G}, \hat{a}_{i}^{1}) \upsilon^{0} (1, \ddot{\bm{\delta}}_{\backslash\!\{\!i\!\}}^{+}),
        \end{align}
        and
        \begin{align}
            & \operatorname{P} (\delta_{i}+1|\delta_{i}, \mathbf{G}, \dot{a}_{i}^{1}) \upsilon^{0} (\delta_{i}+1, \ddot{\bm{\delta}}_{\backslash\!\{\!i\!\}}^{+}) 
            \\
            & \gg \operatorname{P} ({\delta_{i}}^{+}=1|\delta_{i}, \mathbf{G}, \dot{a}_{i}^{1}) \upsilon^{0} (1, \ddot{\bm{\delta}}_{\backslash\!\{\!i\!\}}^{+}),
            \label{ineq: for prop convexity case a.1 much larger}
        \end{align}
        achieved by Lemma~\ref{lemma:much larger}, $\delta_{i}'+1 \gg 1$, and $\delta_{i}+1 \gg 1$, and the last inequality is derived based on~\eqref{ineq: for prop convex, 0} and the following equality
        \begin{align}
            &\operatorname{P} ({\delta_{i}''}+1|\delta_{i}'', \mathbf{G}, \check{a}_{i}^{1}) = \operatorname{P} ({\delta_{i}'}+1|\delta_{i}', \mathbf{G}, \dot{a}_{i}^{1})
            \\
            & = \operatorname{P} ({\delta_{i}}+1|\delta_{i}, \mathbf{G}, \check{a}_{i}^{1}) = \operatorname{P} ({\delta_{i}}+1|\delta_{i}, \mathbf{G}, \dot{a}_{i}^{1})
        \end{align}
        achieved by $\check{a}_{i}^{1} = \dot{a}_{i}^{1}$.

        \item[(a.2)] If $\psi_{i,m_{1}} > \psi_{i,m_{2}}$ and $\check{a}_{j}^{1}= m_{1}$, we define another action $\dot{\mathbf{a}}^{1}$, where $\dot{a}_{i}^{1} = m_{2}, \dot{a}_{j}^{1} = m_{1}$, and $\dot{\mathbf{a}}^{1}_{\backslash\!\{\!i,j\!\}} = \check{\mathbf{a}}^{1}_{\backslash\!\{\!i,j\!\}}$, then we have $\psi_{j,m_{1}} > \psi_{j,m_{2}}$ and
        \begin{align}
            & \hspace{-1.1cm}\operatorname{P} (\delta_{i}''\!+\!1|\delta_{i}'', \!\mathbf{G}, \check{a}_{i}^{1}) \sum_{\dot{\delta}_{j}^{+}}  \operatorname{P} (\dot{\delta}_{j}^{+}|\dot{\delta}_{j}, \!\mathbf{G}, \check{a}_{j}^{1}) \upsilon^{0} (\delta_{i}''\!+\!1, \dot{\bm{\delta}}_{\backslash\!\{\!i\!\}}^{+}) 
            \\ 
            & \hspace{-1.1cm}\!\geq\! \operatorname{P} (\delta_{i}''\!\!+\!1|\delta_{i}'', \!\mathbf{G}, \!\dot{a}_{i}^{1}) \!\sum_{\dot{\delta}_{j}^{+}} \operatorname{P} (\dot{\delta}_{j}^{+}|\dot{\delta}_{j}, \!\mathbf{G}, \!\dot{a}_{j}^{1}) \upsilon^{0} (\delta_{i}''\!+\!1, \dot{\bm{\delta}}_{\backslash\!\{\!i\!\}}^{+}).
            \label{eq: for prop convexity, case a.2}
        \end{align}
        Next, we derive that
        \begin{align}
            & \!\hspace{-0.7cm} \alpha \upsilon^{1}(\dot{\bm{\delta}}'') + (1-\alpha) \upsilon^{1}(\ddot{\bm{\delta}}') - \alpha \upsilon^{1}(\dot{\bm{\delta}}) - (1-\alpha) \upsilon^{1}(\ddot{\bm{\delta}})
            \\
            &\hspace{-0.7cm} \!\geq  \alpha Z(\dot{\bm{\delta}}'', \check{\mathbf{a}}^{1}; \upsilon^{0}) + (1-\alpha) Z(\ddot{\bm{\delta}}', \hat{\mathbf{a}}^{1}; \upsilon^{0}) 
            \\
            &\hspace{-0.6cm} \!- \alpha Z(\dot{\bm{\delta}}, \dot{\mathbf{a}}^{1}; \upsilon^{0}) - (1-\alpha) Z(\ddot{\bm{\delta}}, \hat{\mathbf{a}}^{1}; \upsilon^{0}) 
            \\
            & \hspace{-0.7cm} \!\geq\! \alpha \!\operatorname{P} ({\delta_{i}''}\!+\!1|\delta_{i}'', \!\mathbf{G}, \!\dot{a}_{i}^{1}) \! \sum_{{\dot{\bm{\delta}}_{\backslash\!\{\!i\!\}}}^{+}}  \!\operatorname{P} (\dot{\bm{\delta}}_{\backslash\!\{\!i\!\}}^{+}|\dot{\bm{\delta}}_{\backslash\!\{\!i\!\}}, \!\mathbf{G}, \!\dot{\mathbf{a}}_{\backslash\!\{\!i\!\}}^{1}) \upsilon^{0} (\dot{\bm{\delta}}''^{+})  \!\!
            \\ 
            & \hspace{-0.6cm} \!+ (1\!-\!\alpha) \operatorname{P} ({\delta_{i}'}+1|\delta_{i}', \mathbf{G}, \hat{a}_{i}^{1}) \sum_{\ddot{\delta}_{j}^{+}} \sum_{{\ddot{\bm{\delta}}_{\backslash\!\{\!i,j\!\}}}^{+}}
            \operatorname{P} (\ddot{\delta}_{j}^{+}|\ddot{\delta}_{j}, \mathbf{G}, \hat{a}_{j}^{1}) 
            \\
            & \times \operatorname{P} (\ddot{\bm{\delta}}_{\backslash\!\{\!i,j\!\}}^{+}|\ddot{\bm{\delta}}_{\backslash\!\{\!i,j\!\}}, \mathbf{G}, \hat{\mathbf{a}}_{\backslash\!\{\!i,j\!\}}^{1}) \upsilon^{0} (\ddot{\bm{\delta}}'^{+}) 
            \\
            & \hspace{-0.6cm} \!- \!\alpha \operatorname{P} (\delta_{i}\!+\!1|\delta_{i}, \!\mathbf{G}, \!\dot{a}_{i}^{1}) \! \sum_{{\dot{\bm{\delta}}_{\backslash\!\{\!i\!\}}}^{+}}  \!\operatorname{P} (\dot{\bm{\delta}}_{\backslash\!\{\!i\!\}}^{+}|\dot{\bm{\delta}}_{\backslash\!\{\!i\!\}},\! \mathbf{G}, \!\dot{\mathbf{a}}_{\backslash\!\{\!i\!\}}^{1}) \upsilon^{0} (\dot{\bm{\delta}}^{+}) 
            \\ 
            & \hspace{-0.6cm} \!- (1-\alpha) \operatorname{P} (\delta_{i}+1|\delta_{i},\! \mathbf{G}, \hat{a}_{i}^{1})  \! \sum_{\ddot{\delta}_{j}^{+}} \sum_{{\ddot{\bm{\delta}}_{\backslash\!\{\!i,j\!\}}}^{+}}
            \operatorname{P} (\ddot{\delta}_{j}^{+}|\ddot{\delta}_{j},\! \mathbf{G}, \hat{a}_{j}^{1}) 
            \\
            & \times \operatorname{P} (\ddot{\bm{\delta}}_{\backslash\!\{\!i,j\!\}}^{+}|\ddot{\bm{\delta}}_{\backslash\!\{\!i,j\!\}},\! \mathbf{G}, \hat{\mathbf{a}}_{\backslash\!\{\!i,j\!\}}^{1}) \upsilon^{0} (\ddot{\bm{\delta}}^{+}) \!\!\!
            \\
            & \hspace{-0.7cm} \geq 0,
        \end{align}
        where the second inequality is from ~\eqref{ineq: def cost convexity},~\eqref{eq: for prop convexity, case a.2}, and the following equality
        \begin{align}
            & \operatorname{P} ({\delta_{i}'}^{+}=1|\delta_{i}', \mathbf{G}, \hat{a}_{i}^{1}) \upsilon^{0} (1, \ddot{\bm{\delta}}_{\backslash\!\{\!i\!\}}^{+})
            \\
            & = \operatorname{P} ({\delta_{i}}^{+}=1|\delta_{i}, \mathbf{G}, \hat{a}_{i}^{1}) \upsilon^{0} (1, \ddot{\bm{\delta}}_{\backslash\!\{\!i\!\}}^{+}),
        \end{align}
        and
        \begin{align}
            & \operatorname{P} (\delta_{i}''+1|\delta_{i}'', \mathbf{G}, \check{a}_{i}^{1}) \upsilon^{0} (\delta_{i}''+1, \dot{\bm{\delta}}_{\backslash\!\{\!i\!\}}^{+}) 
            \\
            & \gg \operatorname{P} ({\delta_{i}''}^{+}=1|\delta_{i}'', \mathbf{G}, \check{a}_{i}^{1}) \upsilon^{0} (1, \dot{\bm{\delta}}_{\backslash\!\{\!i\!\}}^{+}),
        \end{align}
        and~\eqref{ineq: for prop convexity case a.1 much larger} achieved by Lemma~\ref{lemma:much larger}, $\delta_{i}''+1 \gg 1$, and $\delta_{i}+1 \gg 1$, and the last inequality is derived based on~\eqref{ineq: for prop convex, 0} and the following equality
        \begin{align}
            & \operatorname{P} ({\delta_{i}''}+1|\delta_{i}'', \mathbf{G}, \dot{a}_{i}^{1}) = \operatorname{P} ({\delta_{i}'}+1|\delta_{i}', \mathbf{G}, \hat{a}_{i}^{1})
            \\
            & = \operatorname{P} ({\delta_{i}}+1|\delta_{i}, \mathbf{G}, \dot{a}_{i}^{1}) = \operatorname{P} ({\delta_{i}}+1|\delta_{i}, \mathbf{G}, \hat{a}_{i}^{1})
        \end{align}
        achieved by $\check{a}_{i}^{1} = \dot{a}_{i}^{1}$.
        \item[(b)] If $\check{a}_{i}^{1} = m, \hat{a}_{i}^{1}= 0$, then we assume that $\hat{a}_{j}^{1} = m$ and we have
        \begin{align}
            & \hspace{-0.8cm} \alpha \upsilon^{1}(\dot{\bm{\delta}}'') + (1-\alpha) \upsilon^{1}(\ddot{\bm{\delta}}') - \alpha \upsilon^{1}(\dot{\bm{\delta}}) - (1-\alpha) \upsilon^{1}(\ddot{\bm{\delta}})
            \\
            &\hspace{-0.8cm} \!\geq  \alpha Z(\dot{\bm{\delta}}'', \check{\mathbf{a}}^{1}; \upsilon^{0}) + (1-\alpha) Z(\ddot{\bm{\delta}}', \hat{\mathbf{a}}^{1}; \upsilon^{0}) 
            \\
            &\hspace{-0.7cm} \!- \alpha Z(\dot{\bm{\delta}}, \check{\mathbf{a}}^{1}; \upsilon^{0}) - (1-\alpha) Z(\ddot{\bm{\delta}}, \hat{\mathbf{a}}^{1}; \upsilon^{0}) 
            \\
            &\hspace{-0.8cm} \!= \left[\alpha c(\dot{\bm{\delta}}'') + (1-\alpha) c(\ddot{\bm{\delta}}') - \alpha c(\dot{\bm{\delta}}) - (1-\alpha) c(\ddot{\bm{\delta}}) \right] 
            \\
            &\hspace{-0.7cm} \!+\! \alpha \! \sum_{{\delta_{i}''}^{+}} \! \sum_{{\dot{\bm{\delta}}_{\backslash\!\{\!i\!\}}}^{+}} \!\operatorname{P} ({\delta_{i}''}^{+}|\delta_{i}'', \!\mathbf{G}, \!\check{a}_{i}^{1}) \!\operatorname{P} (\dot{\bm{\delta}}_{\backslash\!\{\!i\!\}}^{+}|\dot{\bm{\delta}}_{\backslash\!\{\!i\!\}}, \!\mathbf{G}, \!\check{\mathbf{a}}_{\backslash\!\{\!i\!\}}^{1}) \upsilon^{0} (\dot{\bm{\delta}}''^{+})  \!\!
            \\
            &\hspace{-0.7cm} \!+\! (1\!-\!\alpha) \!\!\sum_{{\delta_{i}'}^{+}} \! \sum_{{\ddot{\bm{\delta}}_{\backslash\!\{\!i\!\}}}^{+}} \!\operatorname{P} ({\delta_{i}'}^{+}|\delta_{i}', \!\mathbf{G}, \hat{a}_{i}^{1}) \! \operatorname{P} (\ddot{\bm{\delta}}_{\backslash\!\{\!i\!\}}^{+}|\ddot{\bm{\delta}}_{\backslash\!\{\!i\!\}}, \!\mathbf{G}, \hat{\mathbf{a}}_{\backslash\!\{\!i\!\}}^{1}) \upsilon^{0} (\ddot{\bm{\delta}}'^{+}) \!\!\!
            \\
            &\hspace{-0.7cm} \!- \!\alpha \! \sum_{\delta_{i}^{+}} \! \sum_{{\dot{\bm{\delta}}_{\backslash\!\{\!i\!\}}}^{+}}\! \operatorname{P} (\delta_{i}^{+}|\delta_{i},\! \mathbf{G}, \!\check{a}_{i}^{1}) \operatorname{P} (\dot{\bm{\delta}}_{\backslash\!\{\!i\!\}}^{+}|\dot{\bm{\delta}}_{\backslash\!\{\!i\!\}}, \!\mathbf{G}, \!\check{\mathbf{a}}_{\backslash\!\{\!i\!\}}^{1}) \upsilon^{0} (\dot{\bm{\delta}}^{+}) 
            \\
            & \hspace{-0.7cm} \!-\! (1\!-\!\alpha) \!\sum_{\delta_{i}^{+}}\! \sum_{{\ddot{\bm{\delta}}_{\backslash\!\{\!i\!\}}}^{+}}
            \operatorname{P} (\delta_{i}^{+}|\delta_{i},\! \mathbf{G}, \hat{a}_{i}^{1}) \operatorname{P} (\ddot{\bm{\delta}}_{\backslash\!\{\!i\!\}}^{+}|\ddot{\bm{\delta}}_{\backslash\!\{\!i\!\}},\! \mathbf{G}, \hat{\mathbf{a}}_{\backslash\!\{\!i\!\}}^{1}) \upsilon^{0} (\ddot{\bm{\delta}}^{+}) \!\!\!
            \\
            & \hspace{-0.8cm} \geq \!\alpha \!\operatorname{P} ({\delta_{i}''}\!+\!1|\delta_{i}'', \!\mathbf{G}, \check{a}_{i}^{1}) \!\! \sum_{{\dot{\bm{\delta}}_{\backslash\!\{\!i\!\}}}^{+}} \operatorname{P} (\dot{\bm{\delta}}_{\backslash\!\{\!i\!\}}^{+}|\dot{\bm{\delta}}_{\backslash\!\{\!i\!\}}, \!\mathbf{G}, \check{\mathbf{a}}_{\backslash\!\{\!i\!\}}^{1}) \upsilon^{0} (\dot{\bm{\delta}}''^{+})
            \\
            & \hspace{-0.7cm} + (1\!-\!\alpha) \operatorname{P} ({\delta_{i}'}+1|\delta_{i}', \mathbf{G}, \hat{a}_{i}^{1}) \operatorname{P} ({\ddot{\delta}_{j}}+1|\ddot{\delta}_{j}, \mathbf{G}, \hat{a}_{j}^{1}) 
            \\
            & \times \sum_{{\ddot{\bm{\delta}}_{\backslash\!\{\!i,j\!\}}}^{+}} \operatorname{P} (\ddot{\bm{\delta}}_{\backslash\!\{\!i,j\!\}}^{+}|\ddot{\bm{\delta}}_{\backslash\!\{\!i,j\!\}}, \mathbf{G}, \hat{\mathbf{a}}_{\backslash\!\{\!i,j\!\}}^{1}) \upsilon^{0} (\ddot{\bm{\delta}}'^{+})
            \\
            & \hspace{-0.7cm} - \!\alpha \!\operatorname{P} ({\delta_{i}}\!+\!1|\delta_{i}, \!\mathbf{G}, \check{a}_{i}^{1}) \!\! \sum_{{\dot{\bm{\delta}}_{\backslash\!\{\!i\!\}}}^{+}} \operatorname{P} (\dot{\bm{\delta}}_{\backslash\!\{\!i\!\}}^{+}|\dot{\bm{\delta}}_{\backslash\!\{\!i\!\}},\! \mathbf{G}, \check{\mathbf{a}}_{\backslash\!\{\!i\!\}}^{1}) \upsilon^{0} ({\dot{\bm{\delta}}}^{+})
            \\
            & \hspace{-0.7cm} + (1\!-\!\alpha) \operatorname{P} ({\delta_{i}}+1|\delta_{i}, \mathbf{G}, \hat{a}_{i}^{1}) \operatorname{P} ({\ddot{\delta}_{j}}+1|\ddot{\delta}_{j}, \mathbf{G}, \hat{a}_{j}^{1}) 
            \\
            & \times \sum_{{\ddot{\bm{\delta}}_{\backslash\!\{\!i,j\!\}}}^{+}} \operatorname{P} (\ddot{\bm{\delta}}_{\backslash\!\{\!i,j\!\}}^{+}|\ddot{\bm{\delta}}_{\backslash\!\{\!i,j\!\}}, \mathbf{G}, \hat{\mathbf{a}}_{\backslash\!\{\!i,j\!\}}^{1}) \upsilon^{0} ({\ddot{\bm{\delta}}}^{+})
            \\
            & \hspace{-0.8cm} \geq 0,
        \end{align}
        where the second inequality is derived from~\eqref{ineq: def cost convexity} and the following equality
        \begin{align}
            &\!\!\!\!\!\!\!\!\! \operatorname{P} ({\delta_{i}''}^{+}\!\!\!=\!1|\delta_{i}'', \!\mathbf{G}, \!\check{a}_{i}^{1}) \!\sum_{{\dot{\bm{\delta}}_{\backslash\!\{\!i\!\}}}^{+}} \operatorname{P} (\dot{\bm{\delta}}_{\backslash\!\{\!i\!\}}^{+}|\dot{\bm{\delta}}_{\backslash\!\{\!i\!\}}, \!\mathbf{G}, \!\check{\mathbf{a}}_{\backslash\!\{\!i\!\}}^{1}) \upsilon^{0} (1, \!\dot{\bm{\delta}}_{\backslash\!\{\!i\!\}}^{+}) \!\!
            \\
            &\!\!\!\!\!\!\!\!\!\! =\! \operatorname{P} ({\delta_{i}}^{+}\!\!\!=\!1|\delta_{i}, \!\mathbf{G}, \!\check{a}_{i}^{1}) \!\sum_{{\dot{\bm{\delta}}_{\backslash\!\{\!i\!\}}}^{+}}\!\! \operatorname{P} (\dot{\bm{\delta}}_{\backslash\!\{\!i\!\}}^{+}|\dot{\bm{\delta}}_{\backslash\!\{\!i\!\}}, \!\mathbf{G}, \!\check{\mathbf{a}}_{\backslash\!\{\!i\!\}}^{1}) \upsilon^{0} (1, \!\dot{\bm{\delta}}_{\backslash\!\{\!i\!\}}^{+}), \!\!\!\!
        \end{align}
        and
        \begin{equation}
            \operatorname{P} ({\delta_{i}'}^{+}=1|\delta_{i}', \mathbf{G}, \hat{a}_{i}^{1}) = \operatorname{P} ({\delta_{i}}^{+}=1|\delta_{i}, \mathbf{G}, \hat{a}_{i}^{1}) = 0
        \end{equation}
        from $\hat{a}_{i}^{1}= 0$, and the inequality
        \begin{align}
            & \!\!\!\!\!\operatorname{P} ({\delta_{i}'}\!+\!1|\delta_{i}', \!\mathbf{G}, \hat{a}_{i}^{1}) \operatorname{P} ({\ddot{\delta}_{j}}^{+}\!=\!1|\ddot{\delta}_{j}, \!\mathbf{G}, \hat{a}_{j}^{1}) \upsilon^{0} ({\delta_{i}'}\!+\!1, 1, \ddot{\bm{\delta}}_{\backslash\!\{\!i,j\!\}}^{+}) \!\!
            \\
            & \!\!\!\!\!\!\geq\! \operatorname{P} ({\delta_{i}}\!+\!1|\delta_{i}, \!\mathbf{G}, \hat{a}_{i}^{1}) \operatorname{P} ({\ddot{\delta}_{j}}^{+}\!\!=\!1|\ddot{\delta}_{j}, \!\mathbf{G}, \hat{a}_{j}^{1}) \upsilon^{0} ({\delta_{i}}\!+\!1, \!1, {\ddot{\bm{\delta}}}_{\backslash\!\{\!i,j\!\}}^{+})\!\!\!
        \end{align}
        achieved from Lemma~\ref{lemma:monotone V} with $\delta_{i}' \geq \delta_{i}$, and the last inequality is based on~\eqref{ineq: for prop convex, 0} with the equality
        \begin{equation}
            \operatorname{P} ({\delta_{i}'}+1|\delta_{i}, \mathbf{G}, \hat{a}_{i}^{1}) = \operatorname{P} ({\delta_{i}}+1|\delta_{i}, \mathbf{G}, \hat{a}_{i}^{1}) = 1, 
        \end{equation}
        and 
        \begin{align}
            & \operatorname{P} ({\delta_{i}''}+1|\delta_{i}, \mathbf{G}, \check{a}_{i}^{1}) = \operatorname{P} ({\delta_{i}}+1|\delta_{i}, \mathbf{G}, \check{a}_{i}^{1}) 
            \\
            & = \operatorname{P} ({\dot{\delta}_{j}}+1|\dot{\delta}_{j}, \mathbf{G}, \hat{a}_{j}^{1}) = \operatorname{P} ({\ddot{\delta}_{j}}+1|\ddot{\delta}_{j}, \mathbf{G}, \hat{a}_{j}^{1}),
        \end{align}
        as $\check{a}_{i}^{1} = \hat{a}_{j}^{1}$.
        \item[(c)] If $\check{a}_{i}^{1} = 0, \hat{a}_{i}^{1}= m$ and $\check{a}_{j}^{1}= m$, we define another action $\dot{\mathbf{a}}^{1}$ where $\dot{a}_{i}^{1} = m, \dot{a}_{j}^{1} = 0$, and $\dot{\mathbf{a}}_{\backslash\!\{\!i,j\!\}} = \check{\mathbf{a}}_{\backslash\!\{\!i,j\!\}}$, then we have
        \begin{align}
            & \!\hspace{-0.8cm} \upsilon^{1}(\dot{\bm{\delta}}'') + (1-\alpha) \upsilon^{1}(\ddot{\bm{\delta}}') - \alpha \upsilon^{1}(\dot{\bm{\delta}}) - (1-\alpha) \upsilon^{1}(\ddot{\bm{\delta}})
            \\
            &\hspace{-0.8cm} \!\geq  \alpha Z(\dot{\bm{\delta}}'', \check{\mathbf{a}}^{1}; \upsilon^{0}) + (1-\alpha) Z(\ddot{\bm{\delta}}', \hat{\mathbf{a}}^{1}; \upsilon^{0}) 
            \\
            &\hspace{-0.7cm} \!- \alpha Z(\dot{\bm{\delta}}, \dot{\mathbf{a}}^{1}; \upsilon^{0}) - (1-\alpha) Z(\ddot{\bm{\delta}}, \hat{\mathbf{a}}^{1}; \upsilon^{0}) 
            \\
            &\hspace{-0.8cm} \!= \left[\alpha c(\dot{\bm{\delta}}'') + (1-\alpha) c(\ddot{\bm{\delta}}') - \alpha c(\dot{\bm{\delta}}) - (1-\alpha) c(\ddot{\bm{\delta}}) \right] 
            \\
            &\hspace{-0.7cm} \!+ \alpha \! \sum_{{\delta_{i}''}^{+}} \! \sum_{{\dot{\delta}_{j}}^{+}} \! \sum_{{\dot{\bm{\delta}}_{\backslash\!\{\!i,j\!\}}}^{+}} \operatorname{P} ({\delta_{i}''}^{+}|\delta_{i}'', \mathbf{G}, \check{a}_{i}^{1}) \operatorname{P} ({\dot{\delta}_{j}}^{+}|\dot{\delta}_{j}, \mathbf{G}, \check{a}_{j}^{1})
            \\
            &\times \operatorname{P} (\dot{\bm{\delta}}_{\backslash\!\{\!i,j\!\}}^{+}|\dot{\bm{\delta}}_{\backslash\!\{\!i,j\!\}}, \mathbf{G}, \check{\mathbf{a}}_{\backslash\!\{\!i,j\!\}}^{1}) \upsilon^{0} (\dot{\bm{\delta}}''^{+}) 
            \\
            &\hspace{-0.7cm} \!\!+\! (1\!-\!\alpha) \!\sum_{{\delta_{i}'}^{+}} \! \sum_{{\ddot{\bm{\delta}}_{\backslash\!\{\!i\!\}}}^{+}} \!\operatorname{P} ({\delta_{i}'}^{+}|\delta_{i}', \!\mathbf{G}, \hat{a}_{i}^{1})  \operatorname{P} (\ddot{\bm{\delta}}_{\backslash\!\{\!i\!\}}^{+}|\ddot{\bm{\delta}}_{\backslash\!\{\!i\!\}}, \!\mathbf{G}, \hat{\mathbf{a}}_{\backslash\!\{\!i\!\}}^{1}) \upsilon^{0} (\ddot{\bm{\delta}}'^{+})\!\!
            \\
            &\hspace{-0.7cm} \!- \alpha \! \sum_{{\delta_{i}}^{+}} \! \sum_{{\dot{\delta}_{j}}^{+}} \! \sum_{{\dot{\bm{\delta}}_{\backslash\!\{\!i,j\!\}}}^{+}} \operatorname{P} ({\delta_{i}}^{+}|\delta_{i}, \mathbf{G}, \dot{a}_{i}^{1}) \operatorname{P} ({\dot{\delta}_{j}}^{+}|\dot{\delta}_{j}, \mathbf{G}, \dot{a}_{j}^{1})
            \\
            & \times \operatorname{P} (\dot{\bm{\delta}}_{\backslash\!\{\!i,j\!\}}^{+}|\dot{\bm{\delta}}_{\backslash\!\{\!i,j\!\}}, \mathbf{G}, \dot{\mathbf{a}}_{\backslash\!\{\!i,j\!\}}^{1}) \upsilon^{0} (\dot{\bm{\delta}}^{+}) 
            \\
            &\hspace{-0.7cm} \!-\! (1\!-\!\alpha) \!\sum_{{\delta_{i}}^{+}} \! \sum_{{\ddot{\bm{\delta}}_{\backslash\!\{\!i\!\}}}^{+}} \!\operatorname{P} ({\delta_{i}}^{+}|\delta_{i}, \!\mathbf{G}, \hat{a}_{i}^{1})  \operatorname{P} (\ddot{\bm{\delta}}_{\backslash\!\{\!i\!\}}^{+}|\ddot{\bm{\delta}}_{\backslash\!\{\!i\!\}}, \!\mathbf{G}, \hat{\mathbf{a}}_{\backslash\!\{\!i\!\}}^{1}) \upsilon^{0} (\ddot{\bm{\delta}}^{+})\!\!\!
            \\
            &\hspace{-0.8cm} \!\geq \alpha \operatorname{P} ({\delta_{i}''}+1|\delta_{i}'', \mathbf{G}, \check{a}_{i}^{1}) \operatorname{P} ({\dot{\delta}_{j}}+1|\dot{\delta}_{j}, \mathbf{G}, \check{a}_{j}^{1})
            \\
            &\times \sum_{{\dot{\bm{\delta}}_{\backslash\!\{\!i,j\!\}}}^{+}} \operatorname{P} (\dot{\bm{\delta}}_{\backslash\!\{\!i,j\!\}}^{+}|\dot{\bm{\delta}}_{\backslash\!\{\!i,j\!\}}, \mathbf{G}, \check{\mathbf{a}}_{\backslash\!\{\!i,j\!\}}^{1}) \upsilon^{0} (\dot{\bm{\delta}}''^{+}) 
            \\
            &\hspace{-0.7cm} \!+\! (1\!-\!\alpha) \!\operatorname{P} ({\delta_{i}'}\!+\!1|\delta_{i}', \!\mathbf{G}, \!\hat{a}_{i}^{1}) \!\! \sum_{{\ddot{\bm{\delta}}_{\backslash\!\{\!i\!\}}}^{+}} \!\operatorname{P} (\ddot{\bm{\delta}}_{\backslash\!\{\!i\!\}}^{+}|\ddot{\bm{\delta}}_{\backslash\!\{\!i\!\}}, \!\mathbf{G}, \!\hat{\mathbf{a}}_{\backslash\!\{\!i\!\}}^{1}) \upsilon^{0} (\ddot{\bm{\delta}}'^{+}) \!\!\!\!
            \\
            &\hspace{-0.7cm} - \!\alpha \operatorname{P} ({\delta_{i}}+1|\delta_{i}, \mathbf{G}, \dot{a}_{i}^{1}) \operatorname{P} ({\dot{\delta}_{j}}+1|\dot{\delta}_{j}, \mathbf{G}, \dot{a}_{j}^{1})
            \\
            & \times \sum_{{\dot{\bm{\delta}}_{\backslash\!\{\!i,j\!\}}}^{+}} (\dot{\bm{\delta}}_{\backslash\!\{\!i,j\!\}}^{+}|\dot{\bm{\delta}}_{\backslash\!\{\!i,j\!\}}, \mathbf{G}, \dot{\mathbf{a}}_{\backslash\!\{\!i,j\!\}}^{1}) \upsilon^{0} (\dot{\bm{\delta}}^{+}) 
            \\
            &\hspace{-0.7cm} \!-\! (1\!-\!\alpha) \!\operatorname{P} ({\delta_{i}}\!+\!1|\delta_{i}, \!\mathbf{G}, \!\hat{a}_{i}^{1}) \!\!\!\sum_{{\ddot{\bm{\delta}}_{\backslash\!\{\!i\!\}}}^{+}} \!\operatorname{P} (\ddot{\bm{\delta}}_{\backslash\!\{\!i\!\}}^{+}|\ddot{\bm{\delta}}_{\backslash\!\{\!i\!\}}, \!\mathbf{G}, \!\hat{\mathbf{a}}_{\backslash\!\{\!i\!\}}^{1}) \upsilon^{0} (\ddot{\bm{\delta}}^{+}) \!\!\! 
            \\
            &\hspace{-0.7cm} + \alpha \operatorname{P} ({\delta_{i}''}+1|\delta_{i}'', \mathbf{G}, \check{a}_{i}^{1}) \operatorname{P} ({\dot{\delta}_{j}}^{+}=1|\dot{\delta}_{j}, \mathbf{G}, \check{a}_{j}^{1}) 
            \\
            &\times \sum_{{\dot{\bm{\delta}}_{\backslash\!\{\!i,j\!\}}}^{+}} \operatorname{P} (\dot{\bm{\delta}}_{\backslash\!\{\!i,j\!\}}^{+}|\dot{\bm{\delta}}_{\backslash\!\{\!i,j\!\}}, \mathbf{G}, \check{\mathbf{a}}_{\backslash\!\{\!i,j\!\}}^{1}) \upsilon^{0} (\dot{\bm{\delta}}''^{+}) 
            \\
            &\hspace{-0.8cm} \geq 0,
        \end{align}
        where the second inequality is derived from~\eqref{ineq: def cost convexity} and the following equality
        \begin{align}
            & \!\!\!\!\operatorname{P} ({\delta_{i}'}^{+}\!\!=\!1|\delta_{i}', \!\mathbf{G}, \hat{a}_{i}^{1}) \! \sum_{{\ddot{\bm{\delta}}_{\backslash\!\{\!i\!\}}}^{+}} \operatorname{P} (\ddot{\bm{\delta}}_{\backslash\!\{\!i\!\}}^{+}|\ddot{\bm{\delta}}_{\backslash\!\{\!i\!\}}, \!\mathbf{G}, \hat{\mathbf{a}}_{\backslash\!\{\!i\!\}}^{1}) \upsilon^{0} (1, \ddot{\bm{\delta}}_{\backslash\!\{\!i\!\}}^{+})
            \\
            & \!\!\!\!= \operatorname{P} ({\delta_{i}}^{+}\!\!=\!1|\delta_{i}, \!\mathbf{G}, \hat{a}_{i}^{1}) \! \sum_{{\ddot{\bm{\delta}}_{\backslash\!\{\!i\!\}}}^{+}} \operatorname{P} (\ddot{\bm{\delta}}_{\backslash\!\{\!i\!\}}^{+}|\ddot{\bm{\delta}}_{\backslash\!\{\!i\!\}}, \!\mathbf{G}, \hat{\mathbf{a}}_{\backslash\!\{\!i\!\}}^{1}) \upsilon^{0} (1, \ddot{\bm{\delta}}_{\backslash\!\{\!i\!\}}^{+}),\!\!
        \end{align}
        and
        \begin{equation}
            \operatorname{P} ({\delta_{i}''}^{+}=1|\delta_{i}'', \mathbf{G}, \check{a}_{i}^{1}) = 0
        \end{equation}
        from $\check{a}_{i}^{1} = 0$, and the inequality
        \begin{align}
            & \operatorname{P} ({\delta_{i}}+1|\delta_{i}, \mathbf{G}, \dot{a}_{i}^{1}) \upsilon^{0} ({\delta_{i}}+1, {\dot{\bm{\delta}}}_{\backslash\!\{\!i\!\}}^{+}) 
            \\
            & \gg \operatorname{P} ({\delta_{i}}^{+}=1|\delta_{i}, \mathbf{G}, \dot{a}_{i}^{1}) \upsilon^{0} (1, \dot{\bm{\delta}}_{\backslash\!\{\!i\!\}}^{+}),
        \end{align}
        achieved from Lemma~\ref{lemma:much larger} and ${\delta_{i}}+1 \gg 1$, and the last inequality is based on~\eqref{ineq: for prop convex, 0} with the equality
        \begin{equation}
            \operatorname{P} ({\delta_{i}''}+1|\delta_{i}'', \mathbf{G}, \check{a}_{i}^{1}) = \operatorname{P} ({\dot{\delta}_{j}}+1|\dot{\delta}_{j}, \mathbf{G}, \dot{a}_{j}^{1}) = 1, 
        \end{equation}
        and
        \begin{align}
            & \operatorname{P} ({\dot{\delta}_{j}}+1|\dot{\delta}_{j}, \mathbf{G}, \check{a}_{j}^{1}) = \operatorname{P} ({\delta_{i}'}+1|\delta_{i}', \mathbf{G}, \hat{a}_{i}^{1})
            \\
            & = \operatorname{P} ({\delta_{i}}+1|\delta_{i}, \mathbf{G}, \dot{a}_{i}^{1}) = \operatorname{P} ({\delta_{i}}+1|\delta_{i}, \mathbf{G}, \hat{a}_{i}^{1}),
        \end{align}
        and the inequality $\upsilon^{0} (\dot{\bm{\delta}}''^{+}) \geq 0$.
        Therefore, the Bellman operation~\eqref{eq: bellman operator} preserve this property of the value function $\upsilon^{0}(\mathbf{s})$ to the optimal V function $\upsilon^{*}(\mathbf{s})$.
\end{itemize}

\bibliographystyle{IEEEtran}

\end{document}